\documentclass[preprint,10pt]{elsarticle} 

\makeatletter
\def\ps@pprintTitle{%
	\let\@oddhead\@empty
	\let\@evenhead\@empty
	\def\@oddfoot{}%
	\let\@evenfoot\@oddfoot}
\makeatother
\usepackage[left=15mm,right = 15mm,top = 15mm,bottom = 15mm]{geometry}
\usepackage[titletoc]{appendix}
\usepackage[table]{xcolor}
\usepackage{float}
\usepackage{amssymb, amsmath,amsthm,graphics,psfrag,graphicx,color,fancyhdr}
\usepackage{appendix}
\usepackage{verbatim}
\usepackage{blindtext}
\usepackage{hyperref}
\usepackage{cleveref}
\usepackage{stmaryrd}
\usepackage{soul}
\usepackage{sidecap}
\usepackage{tikz}
\usepackage{rotating}
\usepackage{booktabs}
\usepackage{comment}
\usepackage{todonotes}
\usepackage{multirow}
\usepackage{changepage}
\usepackage{float}
\usepackage{wrapfig}
\usepackage{stmaryrd}
\usepackage{mdframed}
\usepackage{mathtools}
\usepackage{subcaption}

\SetSymbolFont{stmry}{bold}{U}{stmry}{b}{n}

\definecolor{lightblue}{rgb}{0.32,0.45,0.90}
\definecolor{lightgreen}{rgb}{0.42,0.7,0.40}
\numberwithin{equation}{section}
\numberwithin{figure}{section}
\hypersetup{colorlinks=true,linkcolor = blue,citecolor=blue}

\usetikzlibrary{shapes,arrows}
\usetikzlibrary{decorations.markings}
\numberwithin{figure}{section}

\def\b{\boldsymbol}

\makeatletter
\newcommand{\vast}{\bBigg@{4}}
\newcommand{\Vast}{\bBigg@{5}}
\makeatother\def\b{\boldsymbol}

\colorlet{cgray}{gray!20!white}

\theoremstyle{definition}
\newtheorem{definition}{Definition}[section]

\newtheorem{proposition}{Proposition}[section]
\newtheorem*{remark}{Remark}
\newtheorem{example}{Example}[section]

\tikzset{every label/.style={font=\footnotesize,inner sep=1pt}}

\allowdisplaybreaks

\journal{Journal of note taking}
\begin{document}
	
\begin{frontmatter}
\title{Numerically modelling semidirect product geodesics}

\author[inst4]{J Woodfield}

\affiliation[inst4]{organization={Department of Mathematics, Imperial College London},
addressline={South Kensington Campus}, 
city={London},
postcode={SW7 2AZ}, 
country={United Kingdom}}

\begin{abstract}
This paper numerically investigates Euler-Poincaré equations arising from a self-semidirect product group structure. Nonlinearly coupled systems of equations emerge from the semidirect product action where one set of dynamics can be considered in the frame of another. A monolithic energy-preserving continuous Galerkin finite element method is used to study geodesic equations associated with the semidirect product of the diffeomorphism group on a circle with itself. Theoretically predicted peakon solutions are observed as an emergent behaviour. In addition, complicated nonlinear transfers of energy are associated with the semidirect product coupling, where amongst various nonlinear interactions, we observe coupled peakon behaviour. A mimetic (C-grid) finite difference method is used to study the geodesic flow of the semidirect product of the volume preserving diffeomorphism group with itself, where similar coupling behaviour is observed in the vorticity variables. We also investigate coadjoint and Lie-Poisson structures in the context of geodesic equations on semidirect product groups, where the underlying group is first extended by central extension or semidirect product.

\end{abstract}

\end{frontmatter}

\tableofcontents


 \smallskip

\section{Introduction and literature review}

As demonstrated by Arnold \cite{arnold1966geometrie}, Euler's equations describing an incompressible fluid can be considered geodesic equations on the group of volume-preserving diﬀeomorphisms for the $L^2$ kinetic energy \cite{arnold2014topological,ebin1970groups,ebin1968space,omori2006infinite}. The \emph{metric} plays a determining role in the equations, for example, by changing from the $L^2$ metric to the
$\dot{H}^{1/2}$ metric, the 2D Euler's equations become the Surface Quasi Geostrophic equations \cite{washabaugh2015sqg}. Similarly,
\emph{Group and algebra extensions}  have played a historical role in determining geophysical fluid and wave models can be regarded as geodesic equations on an extended group \cite{vizman2008geodesic,zeitlin1994differential,roger1995extensions}. Typical group extensions include: (i) the extension by direct product with another group, $H\times K$; (ii) the semidirect product with another group, $H\ltimes K$; (iii) the semidirect product with a vector space, $G\ltimes V$ and (iv) the central extension of a group $G$, denoted as $\widehat{G}$. This paper investigates the group structures $G\ltimes G$, $\widehat{G}\ltimes \widehat{G}$, $(G\ltimes V)\ltimes  (G\ltimes V)$, and the resulting geodesic equations arising from their corresponding variational principles.

\smallskip

In \cite{escher2011euler} self-semidirect product geodesics associated with $\operatorname{Diff}(S^1)\ltimes\operatorname{Diff}(S^1)$ are considered, and for a specific metric/inertia operator a coupled system of Camassa-Holm equations are derived. This coupled system is subsequently untangled and shown to admit peakon solutions in the uncoupled variables. In this paper, we numerically solve the coupled CH-CH system of equations when choosing weighted $H_{\alpha}^1$, $H_{\beta}^1$ norms. We observe theoretically predicted peakon solutions as an emergent behaviour in addition to other complicated nonlinear interactions. In particular, we observe peakons in each variable travelling together with the same speed. 

To determine if this type of coupling behaviour is associated more generally with semidirect product coupled systems, we investigate the geodesic equations arising from the self-semidirect product group $\operatorname{SDiff}(\mathbb{T}^2)\ltimes \operatorname{SDiff}(\mathbb{T}^2)$, where $\operatorname{SDiff}(\mathbb{T}^2)$ is the volume-preserving diffeomorphism group on the torus. This gave a coupled system of 2D-Euler's equations, where numerically we observe similar coupling between the vorticity variables.

We also investigate group structures arising when considering the self-semidirect product of centrally extended groups.  The Camassa-Holm equation and the KdV equations are integrable equations originally derived to model shallow-water waves over a flat bed \cite{camassa1993integrable,korteweg1895xli} supporting peakon and soliton solutions respectively. The central extension has allowed the Camassa-Holm equation and the KdV equation to be interpreted as geodesic flows on the group of orientation-preserving diffeomorphisms on the circle in the presence of dispersion. More specifically
Ovsienko and Khesin \cite{ovsienko1987korteweg} showed the solutions of the KdV equation can be interpreted as geodesics of the right invariant $L^2$ metric on the Virasoro-Bott group. 
Misiolek \cite{misiolek1998shallow} extended this to show that solutions of the Camasa Holm equation can be interpreted as geodesics of the right invariant $H^1_{\alpha}$ metric on the Virasoro-Bott group. We consider coupled systems of such equations, through the semidirect product of centrally extended groups, where additional dispersive coupling is noted.

We derive self-semidirect product coupled systems of equations arising when a vector space has extended the underlying group. 
The Semi-direct product with the vector space of functions was shown to give a 2-component generalization of the Camassa-Holm equation (CH2) from the geodesic motion for the $H^1$ metric on the semidirect product space $\operatorname{Diff}\left(S^1\right) \ltimes C^{\infty}\left(S^1\right)$ or the Itô \cite{ito1982symmetries} equations for the $L^2$ metric. We consider a coupled system of such equations arising from the self-semidirect product.
This semidirect product space was centrally extended for modified 2-component Cammassa-Holm equations in \cite{ovsienko1994extensions,marcel1997extension,guha2008geodesic} and extended to higher spatial dimensions in \cite{holm2009geodesic}.

The semidirect product of two non-Albelian groups denoted $H \ltimes K$ is a fundamental construction considered in the reduction by stages framework \cite{ortega2013momentum,marsden2007hamiltonian,gay2009geometric,cendra2001lagrangian,cendra1987lin,perlmutter1999symplectic} and has previously received attention in other modelling applications; such as 
multiscale image registration \cite{bruveris2012mixture} hybrid plasma models \cite{holm2010euler}, Hall Magnetohydrodynamics \cite{araki2015differential}, Landau-Lifshitz equations \cite{patrick1999landau}. In \cite{holm2023lagrangian} the Lagrangian reduction by stages framework has been proposed as a coupling mechanism for waves and fluids. This work employs an analogous coupling mechanism, discussed in \cite{holm202431,holm2019stochastic} where one considers one fluid in the frame of another, this leads to adjoint actions associated with the semidirect product of groups \cite{bruveris2012mixture}, and elaborated on in \cref{sec: composition of maps}.

\smallskip



\subsection{Overview}\label{sec:Overview}
The rest of this paper is organised as follows.
\smallskip

Section \ref{sec:Background Material} reviews background material for reference and use in later sections of the paper. 
\begin{itemize}
\item
Section \ref{sec:Euler-Poincaré-Arnold} reviews the variational derivation of the Euler-Poincaré equations. 
Section \ref{sec:diffeomorphisms} reviews the diffeomorphism group, and how the Euler-Poincaré equations give dispersionless Camassa-Holm equations.
Section \ref{sec: volume preserving diffeomorphisms}  reviews Euler-Poincaré equations for the volume-preserving diffeomorphisms and how they give Euler's fluid equation.
\item Section \ref{sec:Direct products} reviews the various adjoint and co-adjoint actions of the direct product of groups $H\times K$ needed later in the paper. 
Section \ref{sec: semidirect products} reviews the adjoint actions associated with the semidirect product of a group with itself $G_H\ltimes G_K$ and how these actions produce different structures in Euler-Poincaré equations.
\end{itemize}
\smallskip

In \cref{sec:Numerical methodology} we numerically investigate the semidirect product geodesics $G\ltimes G$. 
\begin{itemize}
\item
In \cref{sec:geodesics_on_diff_semi_diff} we use an energy preserving finite element method to solve the geodesic equation arising from $\operatorname{Diff}(S^1)\ltimes \operatorname{Diff}(S^1)$ under weighted $H^1_{\alpha}(S)$, $H^1_{\beta}(S)$ norms. The resulting equations are self-semidirect product coupled system of Camassa-Holm equations. 
\item
In \cref{sec:geodesics_on_Sdiff_semi_Sdiff} we solve geodesics arising from $\operatorname{Diff}(\mathbb{T}^2)\ltimes \operatorname{Diff}(\mathbb{T}^2)$ under the $L^2(\mathbb{T}^2)$ norms, a self-semidirect product coupled system of Euler's equations. We observe that geodesics of self-semidirect product systems tend to strongly couple the nonlinear flow dynamics. 
\end{itemize}
\smallskip

In \cref{sec:Further extensions} we derive the semidirect product of extended groups. 
\begin{itemize}
\item
In \cref{sec: semidirect product of centrall extension} we review the central extension of a group $\widehat{G}$, the resulting Euler-Poincaré equations and the Euler-Poincaré equations for $\widehat{G}\ltimes \widehat{G}$. In \cref{sec:Example 0: Geodesics} geodesics of $\widehat{\operatorname{Diff}}(S^1)$ are reviewed, and in \cref{sec:Example 1: Geodesics} geodesics of $\widehat{\operatorname{Diff}}(S^1)\ltimes  \widehat{\operatorname{Diff}}(S^1)$ are introduced. 
\item
In \cref{sec:semi_direct_GV} we review the Euler-Poincaré equation for $G\ltimes V$, and introduce its self-semidirect product $(G\ltimes V)\ltimes (G\ltimes V)$. The group $(\operatorname{Diff}(S^1)\ltimes C^{\infty}(S^1))\ltimes (\operatorname{Diff}(S^1)\ltimes C^{\infty}(S^1))$ is given as an example of semidirect product coupled Itô equations \cite{ito1982symmetries} and $(G\ltimes\mathfrak{g})\ltimes(G\ltimes\mathfrak{g})$, $G=\operatorname{SDiff}(M)$ is given an example of semidirect product coupled magnetohydrodynamics (MHD) \cite{holm202431}. 
\end{itemize}
In \cref{sec:conclusion} we present the conclusion, we summarise the main results, and open problems.
\smallskip
There are 5 appendices. These appendices treat the following topics. In \cref{def:group homomorphism} we define a homomorphism. In \cref{sec: semidirect product of centrall extension} the group axioms for centrally extended groups are verified. In \cref{sec:group-axioms} the group axioms for self-semidirect product extended groups are discussed and the Adjoint actions are in \cref{sec:adjoints for H semi K}. 
In \cref{sec: composition of maps} the composition of maps modelling framework and relationship to the usual Euler-Poincaré equations \cref{sec: composition of maps} are discussed. In \cref{Centrally extending the semidirect product group} a self semidirect product Lie algebra is centrally extended, leading to different coupled systems of equations than the self-semidirect product of centrally extended Lie algebras.

\section{Background Material}\label{sec:Background Material}
\subsection{Euler-Poincaré}\label{sec:Euler-Poincaré-Arnold}
Let $(G,\cdot_G)$ be a Lie group, with group multiplication $g_1 \cdot_G g_2=g_1 g_2$, identity $e_G$, and inverse $g^{-1}$. The group $G$ acts on itself
under inner automorphism, a left action defined as
\begin{align}
\operatorname{AD}_{g_1} g_2: =g_1 \cdot g_2 = g_1 g_2 g_1^{-1}. \label{eq:inner-automorphism}
\end{align}
Here we have adopted the notation $\cdot$ without a subscript,  to denote left conjugation by group action. The derivative of the inner automorphism at the identity gives the Ad-joint representation of the group on its Lie algebra 
\begin{align}
\operatorname{Ad}_{g}\eta := \frac{d}{d\epsilon}\bigg\vert_{\epsilon=0} \operatorname{AD}_{g_1}g_2(\epsilon), \quad\text{where}\quad \eta := \frac{d }{d\epsilon}\bigg|_{\epsilon=0} g_2(\epsilon)\in \mathfrak{g}=T_e G.
\end{align}
The derivative of this Adjoint representation at the identity of the group defines the adjoint representation of the Lie algebra 
\begin{align}
\operatorname{ad}_{\xi}\eta := \frac{d}{d\epsilon}\bigg|_{\epsilon=0} \operatorname{Ad}_{g_1(\epsilon)} \eta\quad\text{where}\quad \xi := \frac{d }{d\epsilon}\bigg|_{\epsilon=0} g_1(\epsilon)\in \mathfrak{g}=T_e G.
\end{align}
A Lie algebra $\mathfrak{g}$ is a vector space and consequnently has a dual vector space denoted $\mathfrak{g}^*$, 
we denote the inner product pairing between the Lie algebra and its dual space as $\langle m,u \rangle_{\mathfrak{g}^*\times\mathfrak{g}}$. The co-Adjoint (anti)representation $\operatorname{Ad}_{g}^*\eta$ and the co-adjoint representation $\operatorname{ad}_{\xi}^*\eta$ are operators defined through the pairing as follows
\begin{align}
\langle 
\operatorname{Ad}_{g}^*\gamma,\eta
\rangle_{\mathfrak{g}^*\times\mathfrak{g}} = \langle \gamma, \operatorname{Ad}_{g}\eta \rangle_{\mathfrak{g}^*\times\mathfrak{g}},\\
\langle 
\operatorname{ad}_{\xi}^*\gamma,\eta
\rangle_{\mathfrak{g}^*\times\mathfrak{g}} = \langle \gamma, \operatorname{ad}_{\xi}\eta \rangle_{\mathfrak{g}^*\times\mathfrak{g}}.
\end{align}
When considering a constrained variational principle of the type
\begin{align}
   0=\delta \int\ell(u)dt,\quad\text{where}\quad  \delta u = \frac{d}{d\epsilon}\bigg|_{\epsilon=0} (\partial_t {g}_{\epsilon})g^{-1}_{\epsilon} = \partial_t v - \operatorname{ad}_{u}v,\quad v(t)\in \mathfrak{g}.
\end{align}
One attains the following Euler-Poincaré-Arnold equation in $\mathfrak{g}^*$
\begin{align}
    \frac{d}{dt}\frac{\delta \ell}{\delta u} + \operatorname{ad}^*_{u}\frac{\delta \ell}{\delta u} = 0, \label{eq:euler-poincare}
\end{align}
 from the variational principle
\begin{align}
 0&= \delta \int \ell(u) dt= \int \left\langle \frac{\delta \ell}{\delta u},\delta u \right\rangle_{\mathfrak{g}^*\times\mathfrak{g}} dt
 = \int \left\langle \frac{\delta \ell}{\delta u}, \frac{dv}{dt}  - \operatorname{ad}_{u}v \right\rangle_{\mathfrak{g}^*\times\mathfrak{g}} dt = -\int \left\langle \frac{d}{dt}\frac{\delta \ell}{\delta u} + \operatorname{ad}^*_{u}\frac{\delta \ell}{\delta u},v \right\rangle_{\mathfrak{g}^*\times\mathfrak{g}} dt. \label{eq:eulerarnold_variational_principle}
\end{align}
In this work we shall refer to a geodesic equation, as the Euler-Poincaré or Euler-Arnold equation \cref{eq:euler-poincare}, arising from critical points of a variational principle with any quadratic Lagrangian associated with a group $G$.

\subsection{Diffeomorphisms}\label{sec:diffeomorphisms}
To derive the Euler-Poincaré equations for fluids, we treat the diffeomorphism group operationally as a Lie group. The group action of the diffeomorphism group $(\operatorname{Diff}(M),\cdot_{\operatorname{Diff}})$, is composition $g_1 \cdot_{\operatorname{Diff}} g_2 = g_1 \circ g_2$. The corresponding algebra consists of vector fields $\operatorname{Vect}(M)$. The geometric dual space to vector fields consists of one-form densities $\Lambda^1(M)\otimes \Lambda^n(M)$ \cite{kirillov2006infinite}. 
\smallskip

The Adjoint representation of $\operatorname{Diff}(M)$, on the space of vectorfields $\operatorname{Vect}(M)$, is given by push-forward $g_*v$. The co-Adjoint antirepresentation of $\operatorname{Diff}(M)$ on $\alpha\otimes D \in \Lambda^1(M)\otimes \Lambda^n(M)$ is given by pull-back $\operatorname{Ad}_g^* (\alpha \otimes D)=g^* (\alpha \otimes D)$, from the consistency of the representation. The $\operatorname{ad}$-joint representation of $\operatorname{Vect}(M)$ on itself, is given by the negative Lie derivative, equal to the negative commutator of vectorfields.
The coadjoint representation of $\operatorname{Vect}(M)$ on the dual space $\Lambda^1(M)\otimes \Lambda^n(M)$ is the Lie derivative. These are written as 
\begin{align}
\operatorname{ad}_{u}v &= -[u,v] = - \mathcal{L}_{u}v, \\
\operatorname{ad}^*_{u}(\alpha\otimes D) &= - \mathcal{L}_{u}^{T}(\alpha\otimes D) = \mathcal{L}_{u}(\alpha\otimes D)\,.
\end{align}
Where the contragradient representation of the Lie derivative is denoted $\mathcal{L}_{u}^{T}$ and is equal to the negative Lie derivative, verifiable through integration by parts. The formal treatment of the diffeomorphism group as a Lie group is beyond the scope of this paper and is considered in \cite{ebin1970groups,ebin1968space,omori2006infinite,michor2008topics}.

\begin{example}[Camassa-Holm equation]
In one dimension with the Lagrangian
\begin{align}
\ell(u) = \frac{1}{2}||u||^2_{H^1_{\alpha}(S)}=\frac{1}{2} \int \left(u^2 + \alpha^2 u^2_x \right)\mathrm{d}x,
\end{align}
the variational derivative is calculated as the one-form density
\begin{align}
 \frac{\delta \ell}{\delta u} = (u - \alpha^2 u_{xx}) \mathrm{d}x \otimes \mathrm{d}x =m \mathrm{d}x\otimes \mathrm{d}x \in \Lambda^1(S) \otimes \Lambda^1(S).
\end{align}
The co-adjoint representation is the Lie derivative of the 1-form density
\begin{align}
\operatorname{ad}^{*}_{u}(m\mathrm{d}x\otimes \mathrm{d}x) = \mathcal{L}_{u}(m\mathrm{d}x\otimes \mathrm{d}x)= \left( (um)_x+u_xm \right) \mathrm{d}x \otimes \mathrm{d}x.
\end{align}
The geodesic equation (\cref{eq:eulerarnold_variational_principle}) associated with the above Lagrangian associated with the $H^1_{\alpha}(S^1)$ norm is the Camassa-Holm equation, 
\begin{align}
\partial_t m + (mu)_x + u_x m = 0, \quad \text{where}\quad m = (u - \alpha^2 u_{xx}).
\end{align}
\end{example}

The above example admits generalisation. For example, by using a weighted $H^k$ norm to define the Lagrangian
\begin{align}
\ell(u) = 1/2 ||u||_{H^k_{\b \alpha}}^2  = 1/2 \int \left( \alpha_0 u^2 + \alpha^2_1 u_x^2 + ... +\alpha_{k}^2 (\partial^{k}_{x} u)^2 \right) \mathrm{d}x
\end{align}
one obtains the momentum $m = u - \alpha^2_1 u_{xx} + ... +(-1)^{k} \alpha_k^2\partial^{2k}_{x}u$ and geometric pairing $\langle m \mathrm{d}x\otimes \mathrm{d}x,u\partial_x\rangle = \ell(u)$, such that ($\alpha_0=1$,$\lbrace\alpha_i=0\rbrace_{i\neq 0,1}$) gives rescaled inviscid Burgers equation, and ($\alpha_1=1$,$\lbrace\alpha_i=0\rbrace_{i\neq 1}$) gives the Hunter-Saxton equation. This example generalises to higher spatial dimensions where under the $L^2$, $H^1$ norm one attains the (Averaged) template matching equations \cite{hirani2001averaged} or the EPDiff equation. Other well-known fluid equations are observed by restricting to volume-preserving diffeomorphisms.


\subsection{Volume preserving Diffeomorphisms}\label{sec: volume preserving diffeomorphisms}






The group of volume-preserving diffeomorphisms $G= \operatorname{SDiff}(M)$, also has group multiplication given by composition (a right action).
The Lie algebra of the group $G= \operatorname{SDiff}(M)$ is the space of divergence-free vector fields $\mathfrak{g}=\operatorname{SVect}(M)$. The dual space $\mathfrak{g}^{*}$ is isomorphic to the space of 1-forms modulo exact 1-forms $\Lambda^1/ \mathrm{d}\Lambda^0$ \cite{khesin2009geometry}.%
\footnote{The additional exact 1-form is necessary because the variational derivative with respect to a divergence-free vector field is in the full dual space of $\operatorname{Vect}(M)$, not necessarily in the dual space of $\operatorname{SVect}(M)$.}
Hence, each element in the dual space may be written as $[\alpha] = \alpha+\mathrm{d}f$ where $f\in \Lambda^0$ is an arbitrary function, $\alpha\in \Lambda^1$ is a 1-form, and $\mathrm{d}$ denotes the exterior derivative. The Adjoint-actions of $g\in \operatorname{SDiff}(M)$ on $v\in \operatorname{SVect}(M)$, and $\alpha \in \operatorname{SVect}(M)^*$ are the usual push-forward and pull-back actions (i.e., representations) associated with the diffeomorphism group discussed in \cref{sec:diffeomorphisms}. 
\smallskip

The adjoint and coadjoint actions of $u\in \operatorname{SVect}(M)$ are given by
\begin{align}
\operatorname{ad}_{u} v &= -\mathcal{L}_{u}v = -[u,v],\\
\operatorname{ad}^*_{u} [\alpha] &= -\mathcal{L}^T_{u}[\alpha]= \mathcal{L}_{u}[\alpha] = \mathcal{L}_u(\alpha+\mathrm{d}f) = \mathcal{L}_u\alpha + \mathrm{d}\mathcal{L}_{u}f = [\mathcal{L}_{u}\alpha].
\end{align}
As introduced above, the notation $[\,\cdot\,]$ refers to 1-forms, modulo exact 1-forms. The notation 
$\mathcal{L}_{u}v$ denotes the Lie derivative, given by the commutator of vector fields $u$ and $v$ denoted $[u,v]$. The notation $\mathcal{L}_{u}[\alpha]$ denotes Lie derivative of the oneform $[\alpha]$. 
\smallskip

The Euler-Poincaré equations are derived modulo exact 1-forms under the pairing $\langle X,[\alpha] \rangle_{\operatorname{SVect(M)}\times \Lambda^1(M)/\mathrm{d}\Lambda^0(M)}$ from the following variational principle
\begin{align}
 0&= \delta \int \ell(u) dt= \int \left\langle \left[\frac{\delta \ell}{\delta u}\right],\delta u \right\rangle dt= \int \left\langle \left[\frac{\delta \ell}{\delta u}\right],\dot{v} - \operatorname{ad}_{u}v \right \rangle dt = -\int \left\langle \frac{d}{dt}\left[\frac{\delta \ell}{\delta u}\right] + \operatorname{ad}^*_{u}\left[\frac{\delta \ell}{\delta u}\right],v \right\rangle dt. \label{eq:ep_oneform}
\end{align}
We refer to \cite{arnol'd2009topological,khesin2009geometry} for information regarding this pairing. We refer to \cite{misiolek2010fredholm}, for the co-ordinate free derivation of the co-adjoint actions in dimensions $(2,3,n)$. In \cite{modin2010geodesics}, it is elaborated in more detail how the same adjoint action can arise under several choices of (inertia, geometric and other) pairings using different isomorphisms in the de-Rham complex. Additional information regarding geometric structures in two-dimensional incompressible flow can be found in \cite{marsden1983coadjoint} and \cite{khesin2009geometry}.

\begin{example}[Euler's equations] Given the Lagrangian
\begin{align}
\ell(u) = \frac{1}{2}||u||_2^2 =\frac{1}{2}\int_M u^2 \mathrm{d}^n x,
\end{align}
the variational derivative is calculable as
\begin{align}
 \frac{\delta \ell}{\delta u} = m \mathrm{d}x = u^{\flat} \in \Lambda^1,
\end{align}
Where $\flat$ denotes the musical isomorphism taking vectors to one forms, locally acting in coordinates as $\flat: (u^i \partial_{i}) \mapsto u_i dx^{i}$. 
The Euler-Poincaré equation \cref{eq:ep_oneform} takes the 1-form (modulo exact 1forms) description 
\begin{align}
\frac{d}{dt} \left[\frac{\delta \ell}{\delta u}\right] + \mathcal{L}_{u} \left[\frac{\delta \ell}{\delta u}\right] = 0.
\end{align}

Which upon introducing the exact 1-form pressure is equivalent to the co-ordinate free incompressible Eulers equation
\begin{align}
\frac{d}{dt} u^{\flat} + \mathcal{L}_{u}u^{\flat} = \mathrm{d}\left(\frac{1}{2}|u|^2 + p\right).
\end{align}
The exterior derivative commutes with the Lie derivative and removes exact 1-forms, such that the exterior derivative gives the 2-form vorticity conservation law 
\begin{align}
0= \left(\frac{d}{dt}+\mathcal{L}_{u}\right)\mathrm{d} \frac{\delta \ell}{\delta u} = \frac{d}{dt} g_t^{*}\mathrm{d}\left(\frac{\delta \ell}{\delta u}\right)=\frac{d}{dt} \operatorname{Ad}^*_{g_t}\mathrm{d}\left(\frac{\delta \ell}{\delta u}\right)= \frac{d}{dt} \omega + \mathcal{L}_u \omega = 0.
\end{align}  
\end{example} 

In two-dimensional incompressible flow, there is no line stretching term and there exists a stream function such that, $\omega = \operatorname{curl}(-\nabla^{\perp}\psi)$ specifies an elliptic equation, that can be solved, relating (potential)vorticity to stream-function, and one can define the two-dimensional incompressible Euler's equation
\begin{align}
\partial_t\omega + \mathcal{L}_u \omega = 0, \quad u = -\nabla^{\perp}\psi, \quad \omega = -\Delta \psi, \quad  \omega(0,x) = \omega_0(x).
\end{align}

Where conveniently for two-dimensional incompressible flow many equivalent representations exist since $\mathcal{L}_{u}\omega = u \cdot \nabla \omega = \nabla \cdot(u\omega)  = \lbrace \psi,\omega\rbrace= J(\psi,\omega)=[\psi,\omega]= -\nabla^{\perp}
\psi\cdot \nabla\omega$. In \cref{sec:geodesics_on_Sdiff_semi_Sdiff}, we discretise the divergence form, giving a flux form discretisation capable of locally preserving vorticity.


\subsection{Direct products \texorpdfstring{$H\times K$}{}}\label{sec:Direct products}
Before embarking on non-trivial group extensions, we review the trivial extension associated with the direct product.

\begin{definition}[Direct product]Given two Lie groups $(H,\cdot_{H}),(K,\cdot_{K})$ one can define a larger direct product Lie group
$G = H\times K$, where multiplication in the larger group is defined as 
\begin{align}
g_1g_2 := (h_1,k_1)(h_2,k_2) &:= (h_1\cdot_{H} h_2, k_1 \cdot_{K}k_2), \quad \forall h_1,h_2\in H, \quad k_1,k_2\in K. \label{eq:dp multiplication}
\end{align}
Associativity, inverse and identity follow from the previous group structures on $(H,\cdot_{H})$, and $(K,\cdot_{K})$. The inverse and identity are $g_1^{-1}=(h_1,k_1)^{-1} = (h_1^{-1},k_1^{-1})$, $e_G = (e_H,e_K)$, and the direct product group is also a Lie group. 
\end{definition}

The direct product of groups is a canonical example of a trivial group extension, satisfying the following equivalency of short exact sequences \footnote{Exactness refers to ($\operatorname{Im}(i) = \operatorname{Ker}(s)$), and equivalency between sequences refers to the existence of a homomorphism between G and $K\times H$}
$$
1 \rightarrow K \xhookrightarrow{i} G \xrightarrow{s} H \rightarrow 1,\quad
1 \rightarrow K \xhookrightarrow{i} K \times H \xrightarrow{s} H \rightarrow 1. 
$$ Where $i$ denotes an injective Lie group homomorphism, and $s$ denotes a surjective Lie group homomorphism. A group homomorphism and its key properties is defined in \cref{def:group homomorphism}. The corresponding AD-joint Adjoint, ad-joint, and coadjoint operators for the larger group $G$ inherit the direct product structure, 
\begin{align}
\operatorname{AD}_{(h_1,k_1)}(h_2,k_2) &= (\operatorname{AD}_{h_1}h_2,\operatorname{AD}_{k_1}k_2),\\
\operatorname{Ad}_{(h,k)}(\eta_h,\eta_k) &= (\operatorname{Ad}_{h}\eta_h,\operatorname{Ad}_{k}\eta_k),\\
\operatorname{ad}_{(\xi_h,\xi_k)}(\eta_h,\eta_k) &= (\operatorname{ad}_{\xi_{h}} \eta_h,\operatorname{ad}_{\xi_{k}}\eta_k),\\
\operatorname{ad}^*_{(\xi_h,\xi_k)}(m,n) &= (\operatorname{ad}^{*}_{\xi_{h}} m,\operatorname{ad}^{*}_{\xi_{k}}n),
\label{eq:direct prod co}
\end{align}
under the paring $$ \langle (m,n),(\eta_k,\eta_h)\rangle_{\mathfrak{g}^*\times\mathfrak{g}} = \langle m,\eta_h \rangle_{\mathfrak{h}^*\times\mathfrak{h}}+\langle n,\eta_k \rangle_{\mathfrak{k}^*\times\mathfrak{k}}.$$
The Euler-Poincaré equations \cref{eq:eulerarnold_variational_principle} for coadjoint motion \cref{eq:direct prod co} acting on  $(m,n)=(\frac{\delta\ell}{\delta\xi_h},\frac{\delta\ell}{\delta\xi_k}) \in \mathfrak{h}^*\times\mathfrak{k}^*,$ or in Lie-Poisson form acting on $(\xi_h,\xi_k) \in \mathfrak{h}\times \mathfrak{k}$ may be written as follows
\begin{align}
\mp
\partial_t \begin{bmatrix}
    m\\
    n
\end{bmatrix} =
\left( \operatorname{ad}^*_{(\xi_h,\xi_k)}(m,n) \right)^{T} = 
\begin{bmatrix}
\operatorname{ad}_{\xi_h}^*(\cdot) & 0\\
0&\operatorname{ad}_{\xi_k}^{*} (\cdot)
\end{bmatrix}
\begin{bmatrix}
    m\\
    n
\end{bmatrix} 
=
\begin{bmatrix}
\operatorname{ad}_{(\cdot)}^*m  & 0 \\
0& \operatorname{ad}_{(\cdot)}^* n
\end{bmatrix}
\begin{bmatrix}
    \xi_h\\
    \xi_k
\end{bmatrix},
\end{align}
where the sign discrepancy arises from the group acting on the left or the right. We have turned the tuples of momenta and algebra into vectors to represent the coadjoint action and Lie-Poisson structure through matrix multiplication as a notational convenience. 

\begin{remark}
These resulting uncoupled equations on $\mathfrak{h}^*$, $\mathfrak{k}^*$ do not interact. Whilst these equations are uncoupled, it is possible to consider a coupled system by modifying the Lagrangian, to include an interaction term. If $\ell(u,v) = \ell^1(u) + \ell^2(v) + \ell^{12}(u,v)$, then the resulting system takes the form 
\begin{align}
\partial_t (m+l) + \operatorname{ad}^*_{u}(m +l) = 0 ,\quad m =\frac{\delta \ell^1}{\delta u},\quad l =\frac{\delta \ell^{12}}{\delta u},\\
\partial_t (n+o) + \operatorname{ad}^*_{v}(n+o) = 0, \quad n = \frac{\delta \ell^1}{\delta v},\quad o =\frac{\delta \ell^{12}}{\delta v}.
\end{align}
Coupling of this type will be referred to as Direct Product Lagrangian Coupling since the coupling comes from modifying the Lagrangian, and leads to Euler-Poincaré equations in transformed momentum variables. 
\end{remark}

\subsection{Semidirect products \texorpdfstring{$H\ltimes K$}
{}}\label{sec: semidirect products}
The semidirect product of groups is an example of a nontrivial group extension and the semidirect product action is motivated in \cref{sec: composition of maps}. 

Let $G = H\ltimes K$ be the outer semidirect product of two Lie groups $(H,\cdot_H)$, $(K,\cdot_K)$, with group multiplication given by 
\begin{align}
(h_1,k_1)\cdot_{\ltimes}(h_2,k_2) &= \left(h_1 h_2, k_1 \phi_{h_{1}} (k_{2})\right) =  (h_1 h_2, k_1 (h_{1} \cdot k_{2}) ), \label{eq:group product structure}
\end{align} 
where the left homomorphism action 
\begin{align}
    \phi_h(k):= h \cdot k :=  hkh^{-1}, \quad H\times K \mapsto K,\label{eq:phihk}
\end{align}  
has the key properties
\begin{align}
\phi_{h_1}(\phi_{h_2}k) = \phi_{(h_1 \cdot_{H} h_2)} k = h_1 \cdot (h_2 \cdot k) &= (h_1 h_2) \cdot k, \label{eq:property1}\\
\phi_{h}(k_1 k_2) = \phi_{h} (k_1) \cdot_{K} \phi_{h}(k_2)  = h\cdot(k_1 k_2) &= (h\cdot k_1)(h\cdot k_2).\label{eq:property2}
\end{align}
The first of these, is the homomorphism property, and the second is the automorphism property. The group $G=H\ltimes K$, inherits the usual Lie-group properties from the previously assumed Lie-Group structures in $H,K$, and the homomorphism properties of $\phi$. More specifically, the identity is $e_{G} = (e_H,e_K)$ and inverse $g^{-1} = (h,h^{-1}\cdot k^{-1})$, and associativity are verified in \cref{sec:group-axioms}. 
We shall now suppose $H$ is a copy of $K$, and by a calculation elaborated upon in \cite{marsden2007hamiltonian} and \cref{sec:adjoints for H semi K} we attain the following AD-Ad-ad-ad*-joint actions. Where it is assumed that the inner product structure is of the form
\begin{align}
    \langle (m,n),(u,v) \rangle_{\mathfrak{g}^*\times\mathfrak{g}} = \langle m,u\rangle_{\mathfrak{h}^*\times\mathfrak{h}}+ \langle n,v\rangle_{\mathfrak{k}^*\times\mathfrak{k}}.
\end{align}

\begin{proposition}[AD-Ad-adjoint for semidirect product of groups]\label{prop:adjoint actions for a semidirect product}
Let $G=H\ltimes K$, where $K$ is a copy of $H$ then the AD-Ad-ad-ad$^*$-joint actions are
\begin{align}
\operatorname{AD}_{g_1}g_2 = \operatorname{AD}_{(h_1,k_1)}(h_2,k_2)
&= \left(\operatorname{AD}_{h_1}h_2,k_1 \cdot_{K} (\phi_{h_1} k_2 ) \cdot_{K}(\phi_{\operatorname{AD}_{h_{1}} h_{2} } k_{1}^{-1} )\right),\\ 
\operatorname{Ad}_{g}{\eta_g} = \operatorname{Ad}_{(h,k)}{(\eta_h,\eta_k)} &= (\operatorname{Ad}_{h}\eta_{h},\operatorname{Ad}_{kh}(\eta_k+\eta_h) -\operatorname{Ad}_{h}\eta_h),\\
\operatorname{ad}_{\xi_{g}}\eta_{g} = \operatorname{ad}_{(\xi_{h},\xi_{k})}(\eta_{h},\eta_{k})
&= (\operatorname{ad}_{\xi_h}\eta_h, \operatorname{ad}_{\xi_h+\xi_k} (\eta_{h}+\eta_{k})- \operatorname{ad}_{\xi_h}\eta_h ),\\
\operatorname{ad}^{*}_{(\xi_h,\xi_k)} (m,n) &= (\operatorname{ad}^{*}_{\xi_h} m+ \operatorname{ad}_{\xi_k}^{*} n, \operatorname{ad}_{(\xi_k + \xi_h)}^{*} n  ). 
\end{align}
\end{proposition}

The resulting Euler-Poincaré equations are then given by 
\begin{align}
\pm \partial_t (m,n) + \operatorname{ad}^{*}_{(\xi_h,\xi_k)} (m,n) = (\pm \partial_t m + \operatorname{ad}^{*}_{\xi_h} m+ \operatorname{ad}_{\xi_k}^{*} n,\pm\partial_t n + \operatorname{ad}_{(\xi_k + \xi_h)}^{*} n  ), \label{eq:euler_arnold}
\end{align}
which we can write in co-adjoint matrix form and Lie-Poisson form respectively as, 
\begin{align}
\mp
\partial_t \begin{bmatrix}
    m\\
    n
\end{bmatrix} 
=
\left(\operatorname{ad}^{*}_{(\xi_h,\xi_k)} (m,n) \right)^{T}= \begin{bmatrix}
\operatorname{ad}_{\xi_h}^*(\cdot) & \operatorname{ad}_{\xi_k}^*(\cdot)\\
0&\operatorname{ad}_{\xi_h+\xi_k}^{*} (\cdot)
\end{bmatrix}
\begin{bmatrix}
    m\\
    n
\end{bmatrix} 
= 
\begin{bmatrix}
\operatorname{ad}_{(\cdot)}^*m & \operatorname{ad}_{(\cdot)}^*n \\
\operatorname{ad}_{(\cdot)}^*n& \operatorname{ad}_{(\cdot)}^*n
\end{bmatrix}
\begin{bmatrix}
    \xi_h\\
    \xi_k
\end{bmatrix}. \label{eq:euler_arnold_matrix}
\end{align}
Here, the positive/negative choice arises from whether the group acts on the left or the right. The diffeomorphism group acts by right action and corresponds to the upper choice of operation of $\pm,\mp$ in \cref{eq:euler_arnold,eq:euler_arnold_matrix}, and will be chosen as default throughout the rest of this paper for readability. One can diagonalise the co-adjoint matrix representation and the Lie-Poisson representation, as follows,
\begin{align}
- \partial_t
\begin{bmatrix}
    m-n\\
    n
\end{bmatrix} 
=
\begin{bmatrix}
\operatorname{ad}_{\xi_h}^*(\cdot) & 0\\
0&\operatorname{ad}_{\xi_h+\xi_k}^{*} (\cdot)
\end{bmatrix}
\begin{bmatrix}
    m-n\\
    n
\end{bmatrix}
= 
\begin{bmatrix}
\operatorname{ad}_{(\cdot)}^*(m-n) & 0 \\
        0       & \operatorname{ad}_{(\cdot)}^*n
\end{bmatrix}
\begin{bmatrix}
    \xi_h\\
    \xi_h+\xi_k
\end{bmatrix}. 
\label{eq:euler_arnold_matrix2}
\end{align}

\smallskip

Adjoint actions in the more general case of semidirect product Lie groups and Lie algebra's when $H\neq K$ are also well known. Hochschild \cite{hochschild1965structure} in particular has a discussion specifically on the Ad-joint relationship, describing the induced adjoint representation of the semidirect product Lie group $H\ltimes K$ on its Lie algebra, when $H\neq K$. One can consult the books \cite{ortega2013momentum,marsden2007hamiltonian} for a more detailed exposition into how the semidirect product of different groups fits into reduction and mechanics.

\section{Numerical Experiments}\label{sec:Numerical methodology}
We now turn back to the previous two examples, the Camassa-Holm equation and Euler's equation, to numerically model their self-semidirect product extensions in comparison with the direct product. 

\subsection{Geodesics on \texorpdfstring{$\operatorname{Diff}(S^1)\ltimes \operatorname{Diff}(S^1)$}{}}\label{sec:geodesics_on_diff_semi_diff}
The geodesic associated with $\operatorname{Diff}(S^1)\ltimes \operatorname{Diff}(S^1)$, with metric $\ell(u,v) = \frac{1}{2}\left( ||u||^2_{H^1_{\alpha}}+||v||^2_{H^1_{\beta}}\right)$, can be derived by combining \cref{sec: semidirect products,sec: volume preserving diffeomorphisms} and corresponds to a semidirect product coupled Camassa Holm - Camassa Holm (CHCH) system
\begin{align}
\partial_t m + (um)_x +u_x m + (vn)_x +v_x n &= 0,\label{eq:SDP_CH_1}\\
m - (u-\alpha^2 u_{xx}) &=0,\label{eq:SDP_CH_2}\\
\partial_t n  + ((u+v)n)_x +(u+v)_x n &= 0,\label{eq:SDP_CH_3}\\
n -( v - \beta^2 v_{xx}) &=0.\label{eq:SDP_CH_4}
\end{align}

This set of equations has been considered in \cite{escher2011euler}, where it is shown analytically that the system admits peakon solutions. We wish to verify this numerically and determine what type of nonlinear interactions occur between the peakons. By starting from non-peakon initial conditions, we investigate whether peakons are an emergent behaviour of the coupled system. We consider a Monolythic Continuous Galerkin(CG) Finite Element(FE) energy-preserving discretisation of the weak form of the above CHCH system, following the numerical approach in \cite{bendall2021perspectives}. 

Let $L^2(\Omega)$ denote the space of square integrable
functions on domain $\Omega$, whose norm and inner product are denoted $||\cdot||_{2},\langle a,b\rangle_{2}$. Let $H^1(\Omega)$ to refer to the Sobolev space $W^{1,2}(\Omega)$, whose norm and inner product are denoted $||\cdot||$, $\langle \cdot , \cdot \rangle_{H^1}$. Let $H^1_{\alpha}(\Omega)$ refer to a weighted Hilbert space with norm $||\cdot||_{H^1_{\alpha}} := \int u^2 + \alpha^2  u^2_x \mathrm{d}x$ with a weighted inner product
\begin{align}
\langle q, u\rangle_{H^{1}_{\alpha}} = \int qu + \alpha^2 q_x u_x \mathrm{d}x.
\end{align}

We define a weak solution $(m,u,n,v)\in [H^1(S^1)]^4 = V$ of the semidirect product system, by testing against arbitrary functions $(p,q,s,t) \in V$.
\begin{align}
\langle \partial_t m , p\rangle_{2} + \langle m u_x + n v_x, p \rangle_{2} - \langle  m u + n v, p_x \rangle_{2} &= 0,\quad \forall p \in H^1(S^1),\\
\langle m, q\rangle_{2} - \langle u, q\rangle_{H^{1}_{\alpha}}   &= 0, \quad \forall q \in H^1(S^1),\\
\langle \partial_t n, t\rangle_{2} + \langle  (u_x +v_x) n, t\rangle_{2} - \langle (u + v)n,t_x\rangle_{2} &= 0, \quad \forall t \in H^1(S^1),\\
\langle  n,s \rangle_{2} - \langle v,s\rangle_{H^{1}_{\beta}}  &= 0,\quad \forall s\in H^1(S^1).
\end{align}

To discretise in time we define the intermediate (implicit) variables, 
$m^{av} = \frac{1}{2}(m^{n+1} + m^{n})$,
$u^{av} = \frac{1}{2}(u^{n+1} + u^{n})$,
$n^{av} = \frac{1}{2}(n^{n+1} + n^{n})$,
$v^{av} = \frac{1}{2}(v^{n+1} + v^{n})$
and use a monolithic trapesium-rule time integration with first-order continuous Galerkin finite elements $
\text{CG}_{1}(\mathbb{T}):=\left\{v \in H^1(\mathbb{T}):\left.v\right|_K \in P_1(K),\quad \forall K \in \mathcal{M}_h\right\}
$. Where $\mathcal{M}_h$ denote an equispaced mesh of the periodic line $S^1=\mathbb{T}$ and $P_{i}(K)$ denote the $i$-th order polynomial, on the cell $K\in M$. 

\begin{align}
\langle ( m^{n+1}-m^{n})\Delta t^{-1} , p\rangle_{2} + \langle m^{av} u^{av}_x + n^{av} v^{av}_x, p \rangle_{2} - \langle  m^{av} u^{av} + n^{av} v^{av}, p_x \rangle_{2} &= 0,\quad \forall p \in \operatorname{CG}_1(\mathbb{T}),\\
\langle m^{n+1}, q\rangle_{2} - \langle u^{n+1}, q\rangle_{H^{1}_{\alpha}} &= 0, \quad \forall q \in \operatorname{CG}_1(\mathbb{T}),\\
\langle (n^{n+1}-n^{n})\Delta t^{-1}, t\rangle_{2} + \langle  (u^{av}_x +v^{av}_x) n^{av}, t\rangle_{2} - \langle (u^{av} + v^{av})n^{av},t_x\rangle_{2} &= 0, \quad \forall t \in \operatorname{CG}_1(\mathbb{T}),\\
\langle n^{n+1},s \rangle_{2} - \langle  v^{n+1},s\rangle_{H^{1}_{\beta}}  &= 0,\quad \forall s\in \operatorname{CG}_1(\mathbb{T}).
\end{align}

Where we have let $(m,u,n,v)\in \operatorname{CG}_1(\mathbb{T})^4$, and the test functions
$(q,p,s,t)\in \operatorname{CG}_1(\mathbb{T})^4$, and the inner products are understood in a discrete setting. 
Newton's iteration treats nonlinearity.
The initial conditions considered are 
$v = \frac{1}{5}\cosh(x-\frac{403}{15})$, $u =\frac{1}{2}\cosh(x-\frac{203}{15})$, and defined on the periodic interval $[0,40]$. The trapezium rule preserves quadratic invariants and the total energy of the system. The trapezium rule is doubly mutually conjugate to the symplectic implicit midpoint rule \cite{hairer2006geometric}. We consider parameters $\alpha = 1$, $\beta = \frac{1}{2}$, and use $1000$ elements, with a time increment of $\Delta t =0.0005$, over $nt = 64000$ timesteps. In order to facilitate a comparison, we also numerically solve the uncoupled direct-product $\operatorname{Diff}(S^1)\times \operatorname{Diff}(S^1)$ system
\begin{align}
\partial_t m + (u m)_x +u_x m &= 0,\label{eq:DP1}\\
m - (u-\alpha^2 u_{xx}) &=0,\label{eq:DP2}\\
\partial_t n  + (v n)_x +v_x n &= 0,\label{eq:DP3}\\
n -( v - \beta^2 v_{xx}) &=0,\label{eq:DP4}
\end{align}
using the same numerical method, parameters and initial conditions. Snapshots of the \textbf{\emph{direct product}} system \cref{eq:DP1,eq:DP2,eq:DP3,eq:DP4} are shown in \cref{FIG:DP_snapshots} as profiles in $(u,v)$ as the solution evolves, at $t = 0,4,8,...,32$. Snapshots of the \textbf{\emph{Semidirect}} product coupled system \cref{eq:SDP_CH_1,eq:SDP_CH_2,eq:SDP_CH_3,eq:SDP_CH_4} are shown in \cref{FIG:snapshots} as the solution evolves, at $t = 0,4,8,...,32$. Spacetime solution and energy-time plots of the \textbf{\emph{direct product}} system are plotted in \cref{fig:DP_CH-spacetime-and-energy}. Spacetime solution and energy-time plots of the \textbf{\emph{semidirect product}} system are plotted in \cref{fig:CH-spacetime-and-energy}.

\begin{figure}[H]
\centering
\begin{subfigure}[t]{0.3\textwidth}
\centering
\includegraphics[width=1\textwidth]{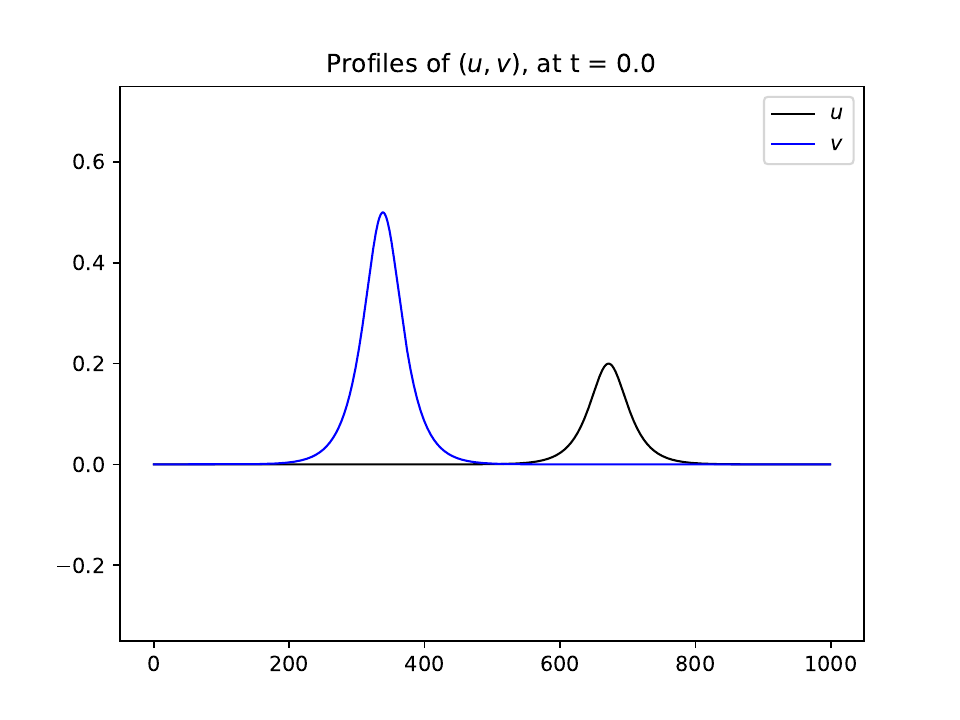}
\caption{}\label{FIG:DP_profile0_CH}
\end{subfigure}
\begin{subfigure}[t]{0.3\textwidth}
        \centering
\includegraphics[width=1\textwidth]{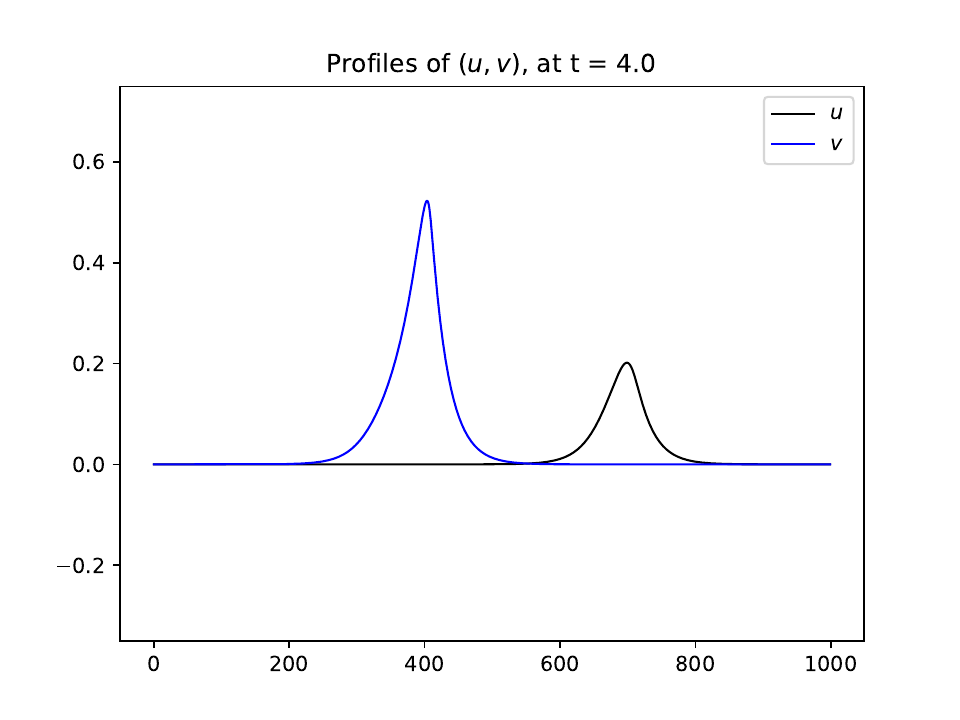}
\caption{}\label{FIG:DP_profile1_CH}
\end{subfigure}
\begin{subfigure}[t]{0.3\textwidth}
        \centering
\includegraphics[width=1\textwidth]{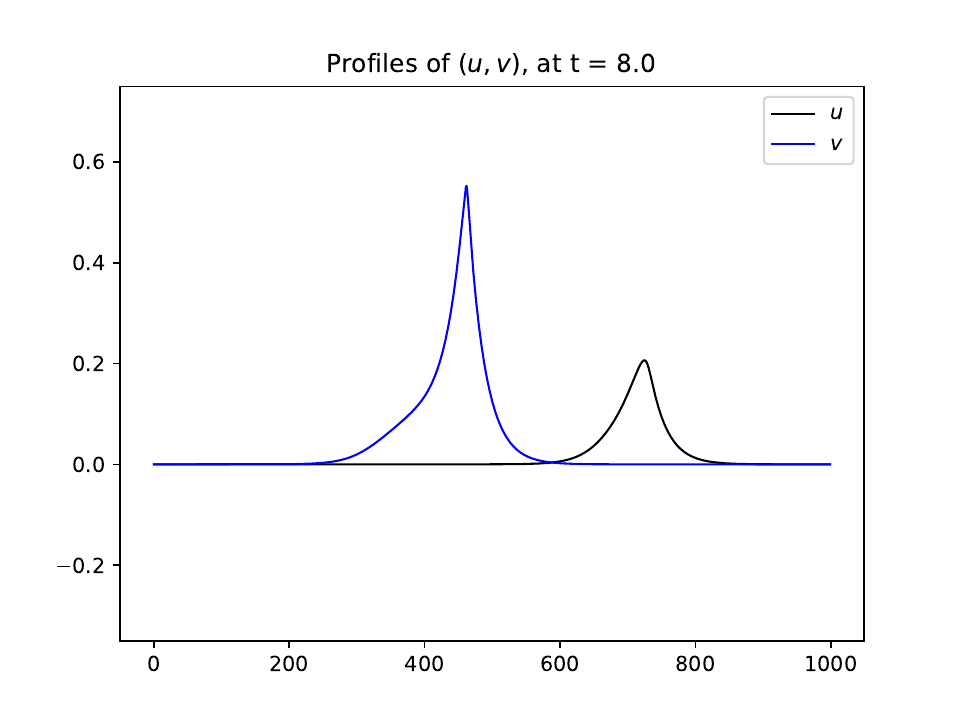}
\caption{}
\label{FIG:DP_profile2_CH}
\end{subfigure}
\begin{subfigure}[t]{0.3\textwidth}
        \centering
\includegraphics[width=1\textwidth]{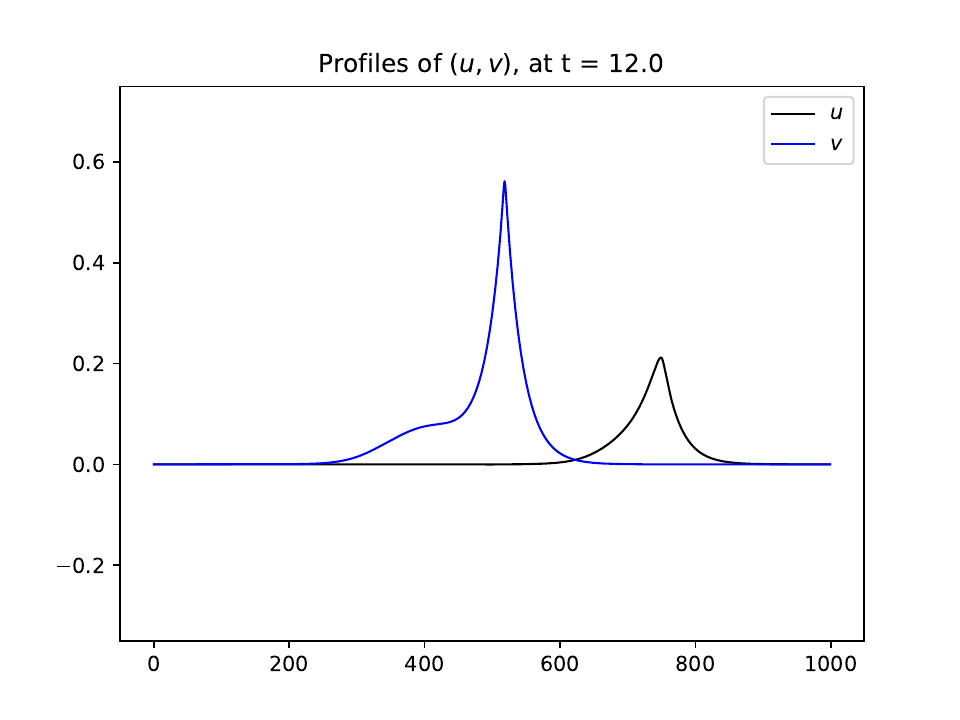}
\caption{}
\label{FIG:DP_profile3_CH}
\end{subfigure}
\begin{subfigure}[t]{0.3\textwidth}
        \centering
\includegraphics[width=1\textwidth]{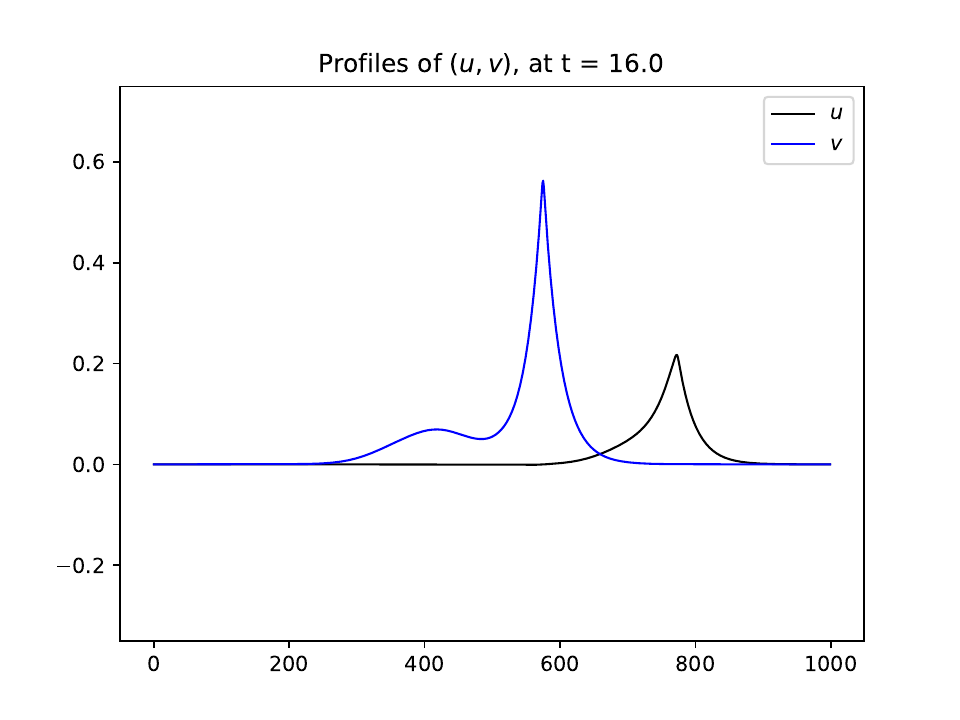}
\caption{}
\label{FIG:DP_profile4_CH}
\end{subfigure}
\begin{subfigure}[t]{0.3\textwidth}
        \centering
\includegraphics[width=1\textwidth]{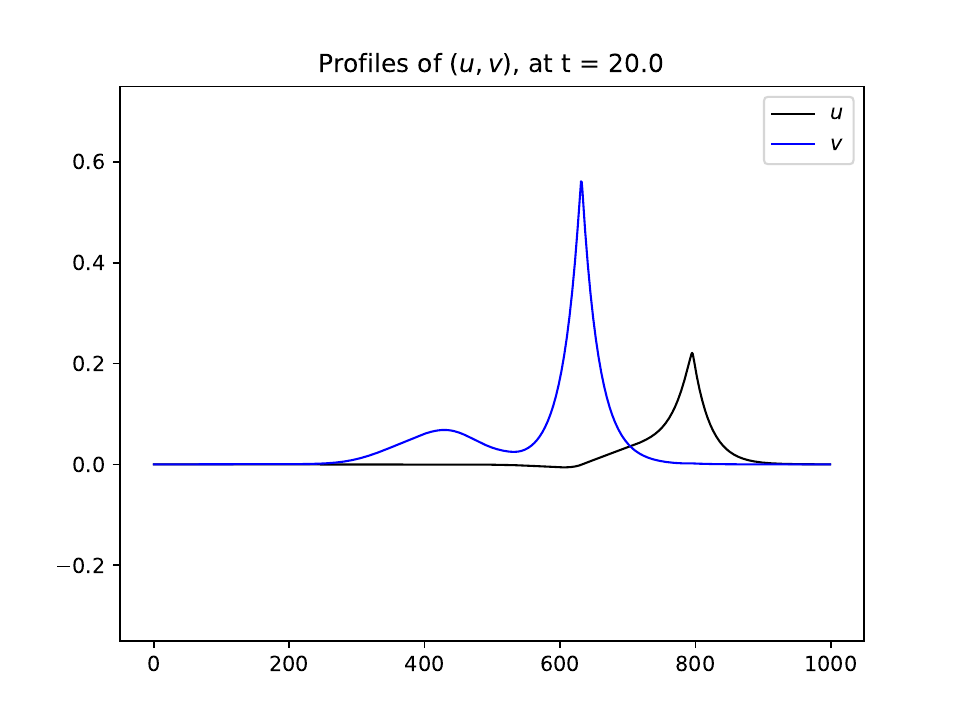}
\caption{}
\label{FIG:DP_profile5_CH}
\end{subfigure}
\begin{subfigure}[t]{0.3\textwidth}
        \centering
\includegraphics[width=1\textwidth]{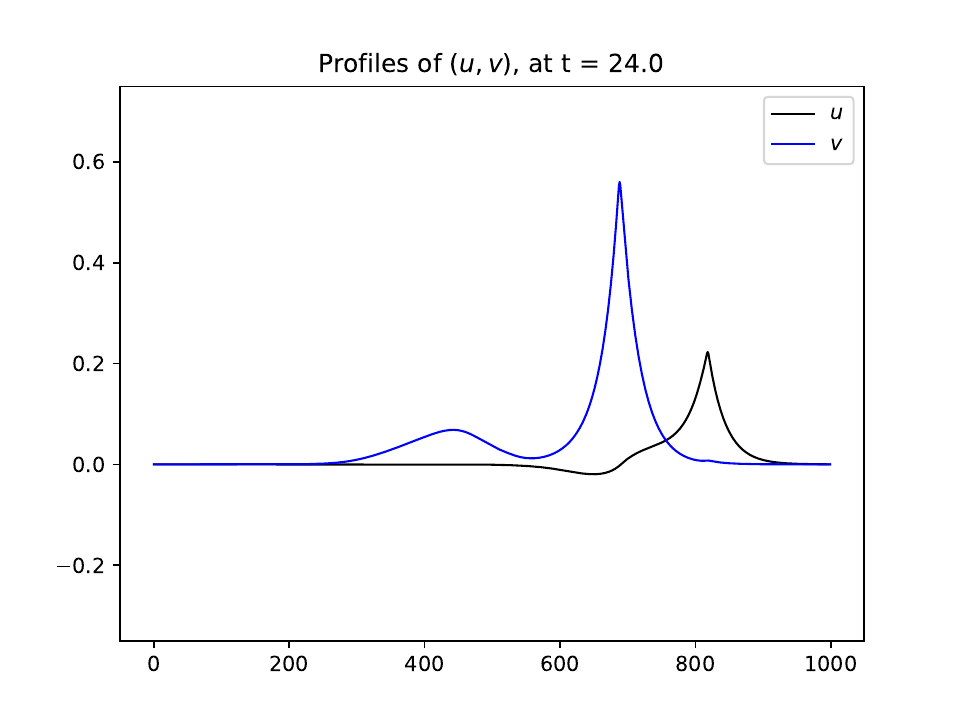}
\caption{}
\label{FIG:DP_profile6_CH}
\end{subfigure}
\begin{subfigure}[t]{0.3\textwidth}
        \centering
\includegraphics[width=1\textwidth]{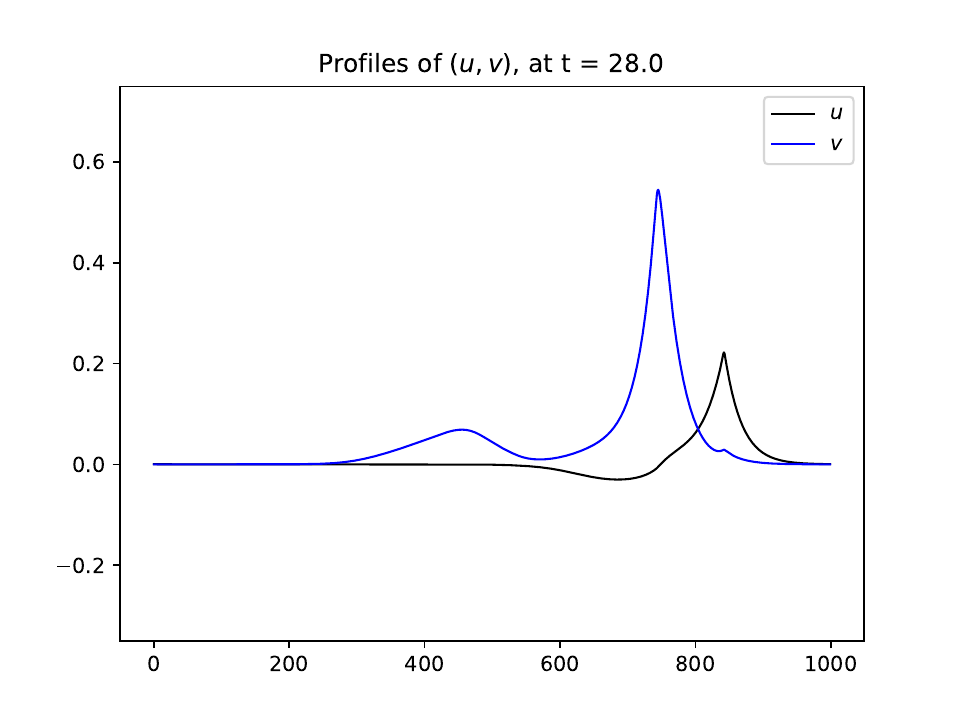}
\caption{}
\label{FIG:DP_profile7_CH}
\end{subfigure}
\begin{subfigure}[t]{0.3\textwidth}
        \centering
\includegraphics[width=1\textwidth]{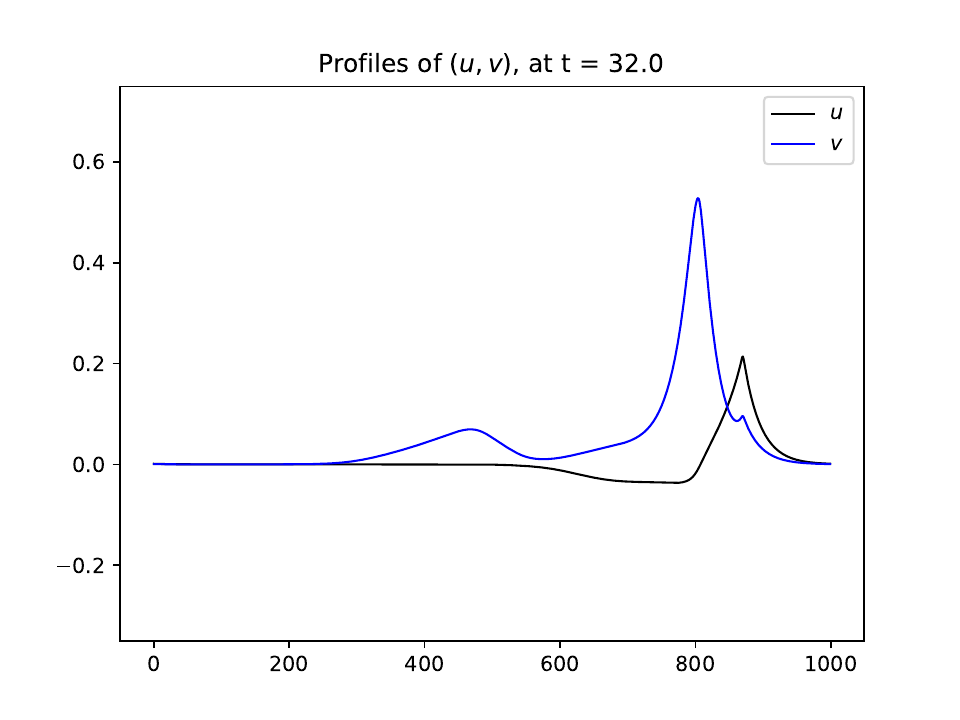}
\caption{}
\label{FIG:DP_profile8_CH}
\end{subfigure}
\caption{Finite element solution of the CH-CH \textbf{\emph{direct product}} system \cref{eq:DP1,eq:DP2,eq:DP3,eq:DP4} at $t=0,4,8,...,32$ in the variables $(u,v)$, plotted in black and blue respectively. We observe two uncoupled solutions of the Camassa-Holm equation, both initial conditions separate out into peakons. In blue the taller peakon travels faster and almost overtakes the slower peakon in black. }
\label{FIG:DP_snapshots}
\end{figure}

\begin{figure}[H]
\centering
\begin{subfigure}[H]{0.495\textwidth}
\centering
\includegraphics[width=1\textwidth]{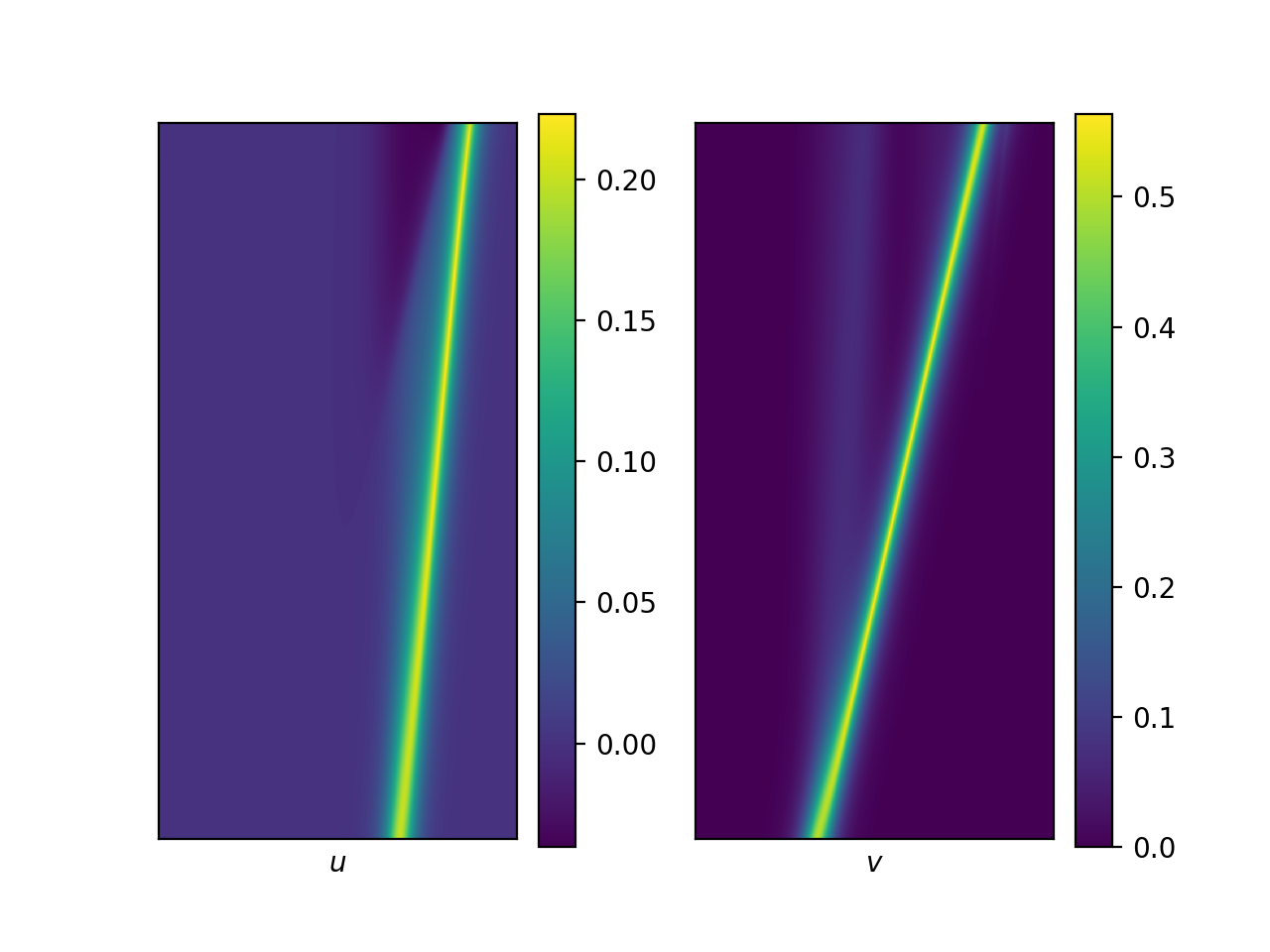}
\caption{Space-Time plot, of the \textbf{\emph{direct product}} CH-CH system $(u,v)$ respectively.}
\label{fig:DP_Space-Time plot}
\end{subfigure}
\begin{subfigure}[H]{0.495\textwidth}
\centering
\includegraphics[width=1\textwidth]{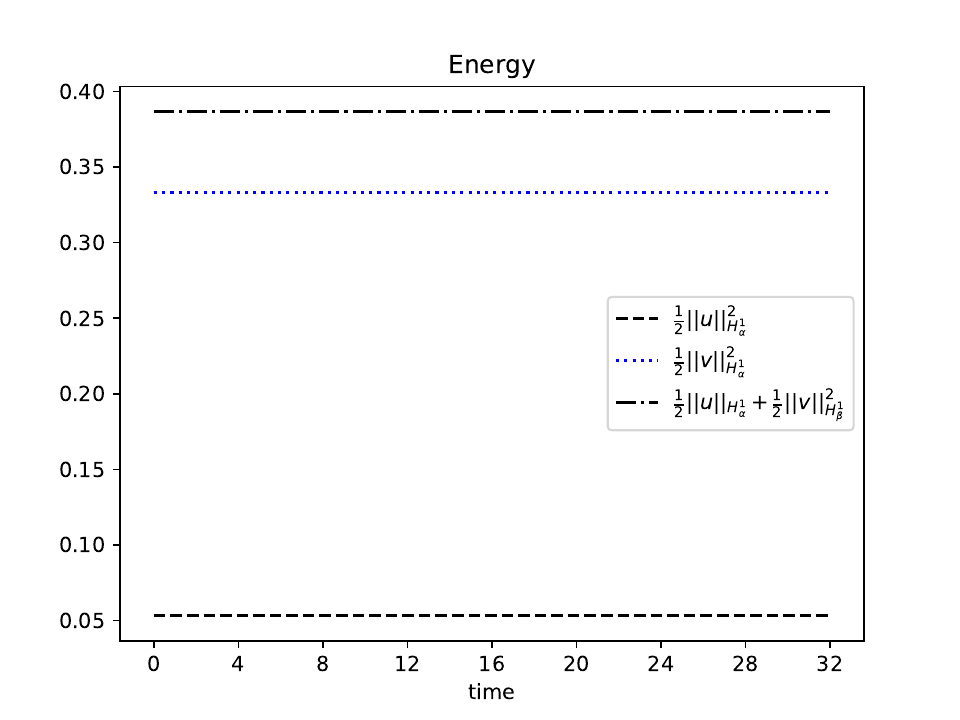}
\caption{Energy over $t\in[0,32]$ of \textbf{\emph{direct product}} CH-CH system}
\label{fig:DP_Energy}
\end{subfigure}
\caption{\textbf{\emph{Direct product}} uncoupled CH-CH system.  Space time plot \cref{fig:DP_Space-Time plot} and Energy plot \cref{fig:DP_Energy}. We observe energy preservation in each system. } 
\label{fig:DP_CH-spacetime-and-energy}
\end{figure}

\begin{figure}[H]
\centering
\begin{subfigure}[t]{0.3\textwidth}
\centering
\includegraphics[width=1\textwidth]{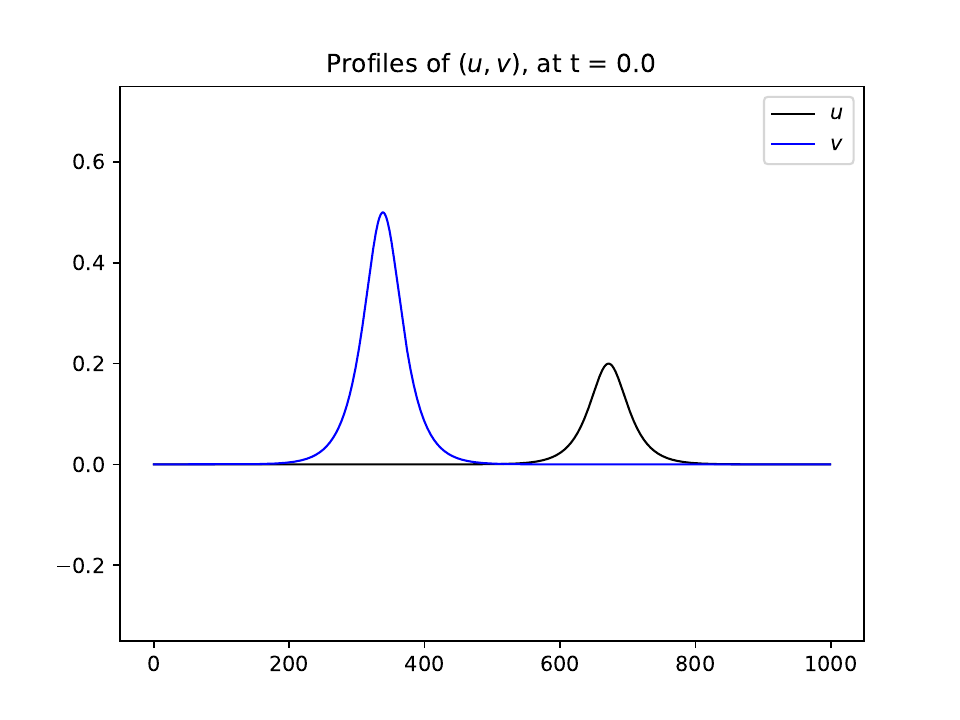}
\caption{}\label{FIG:profile0}
\end{subfigure}
\begin{subfigure}[t]{0.3\textwidth}
        \centering
\includegraphics[width=1\textwidth]{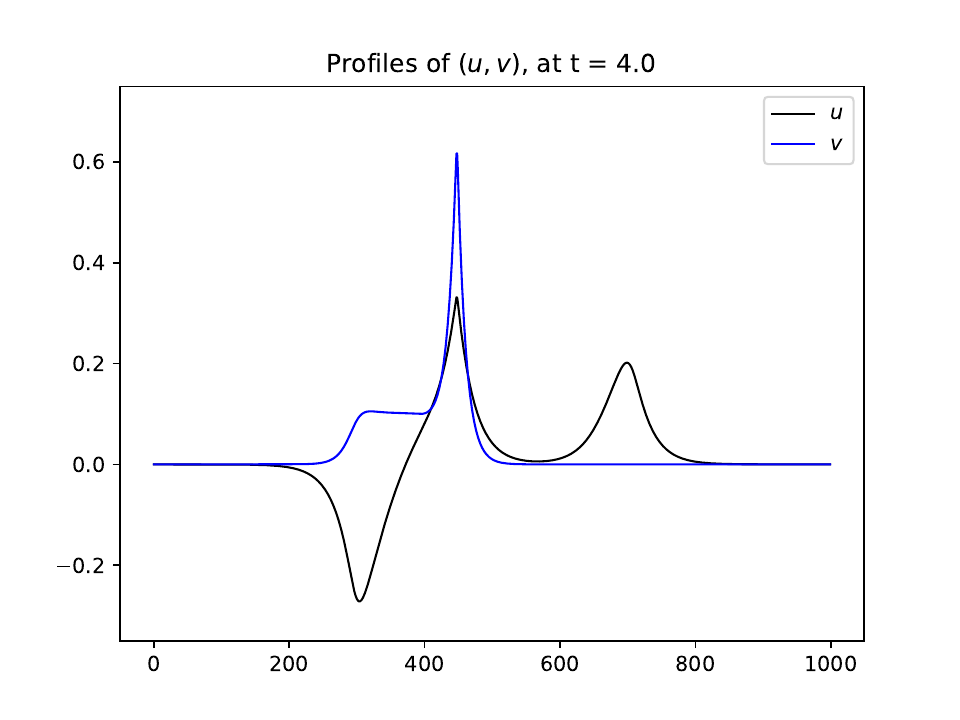}
\caption{}\label{FIG:profile1}
\end{subfigure}
\begin{subfigure}[t]{0.3\textwidth}
        \centering
\includegraphics[width=1\textwidth]{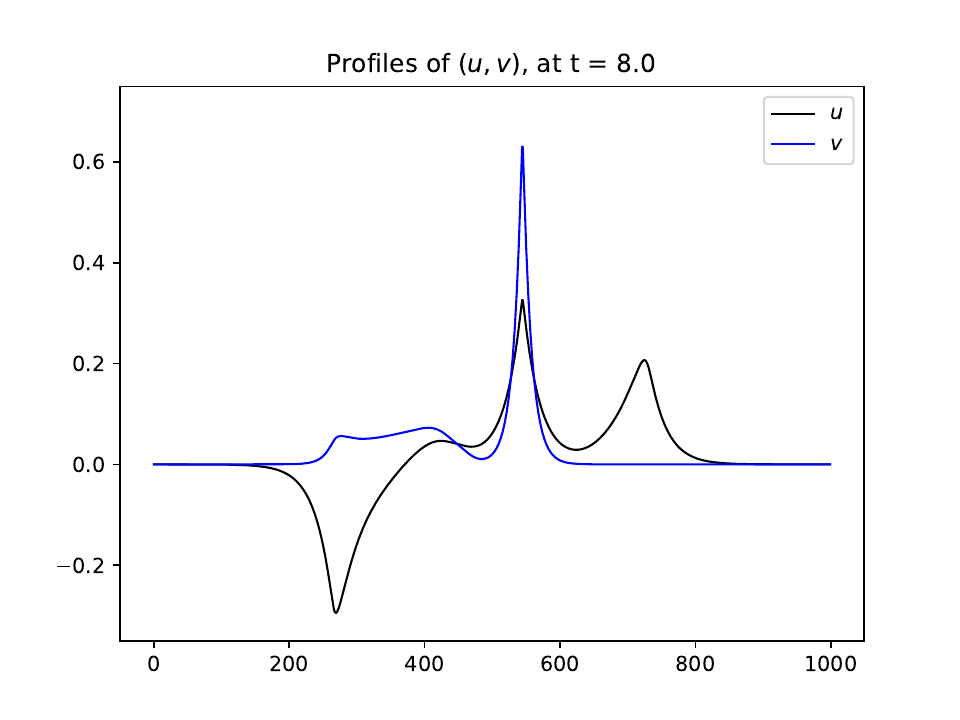}
\caption{}
\label{FIG:profile2}
\end{subfigure}
\begin{subfigure}[t]{0.3\textwidth}
        \centering
\includegraphics[width=1\textwidth]{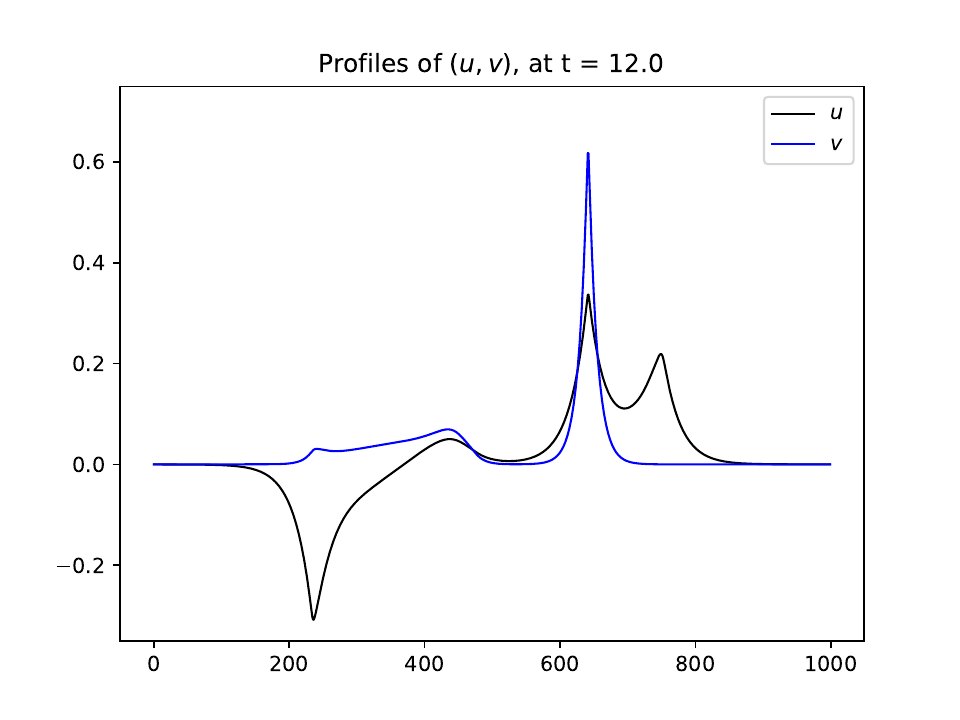}
\caption{}
\label{FIG:profile3}
\end{subfigure}
\begin{subfigure}[t]{0.3\textwidth}
        \centering
\includegraphics[width=1\textwidth]{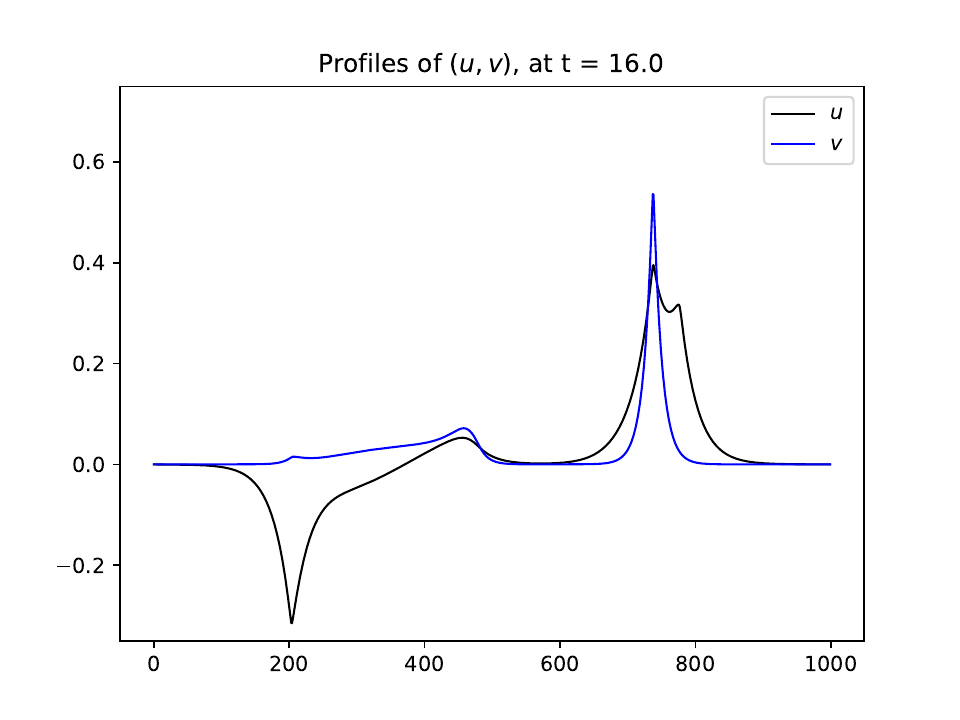}
\caption{}
\label{FIG:profile4}
\end{subfigure}
\begin{subfigure}[t]{0.3\textwidth}
        \centering
\includegraphics[width=1\textwidth]{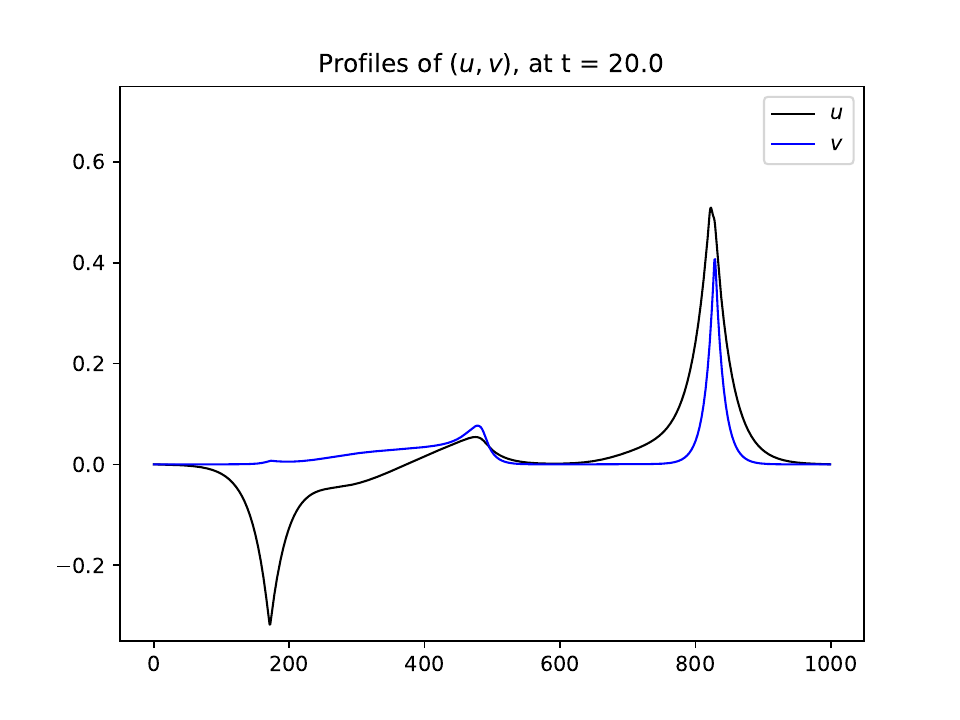}
\caption{}
\label{FIG:profile5}
\end{subfigure}
\begin{subfigure}[t]{0.3\textwidth}
        \centering
\includegraphics[width=1\textwidth]{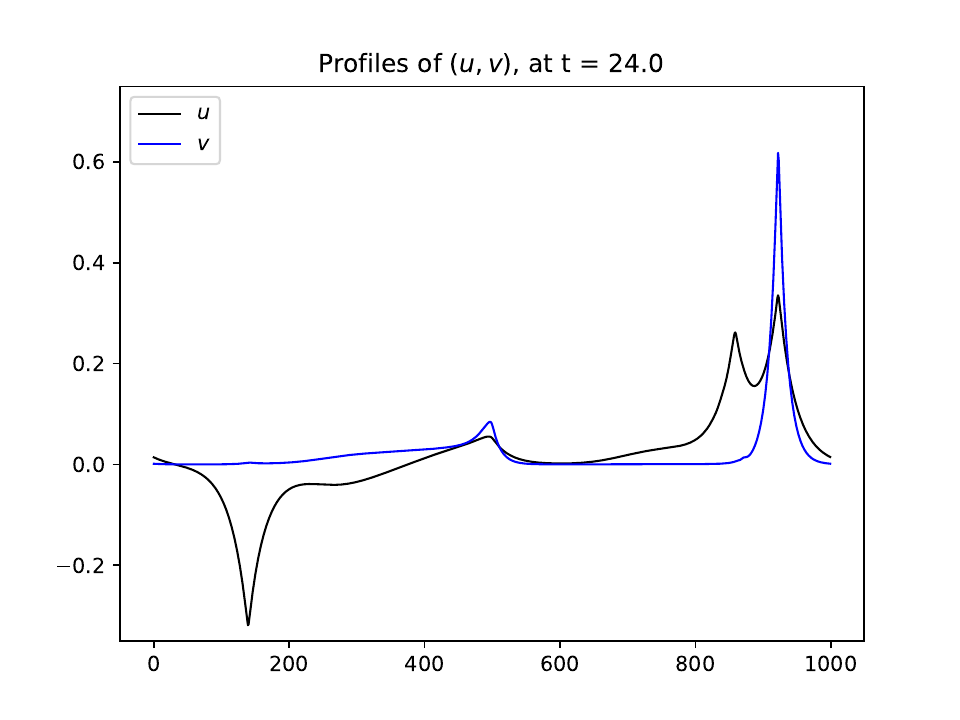}
\caption{}
\label{FIG:profile6}
\end{subfigure}
\begin{subfigure}[t]{0.3\textwidth}
        \centering
\includegraphics[width=1\textwidth]{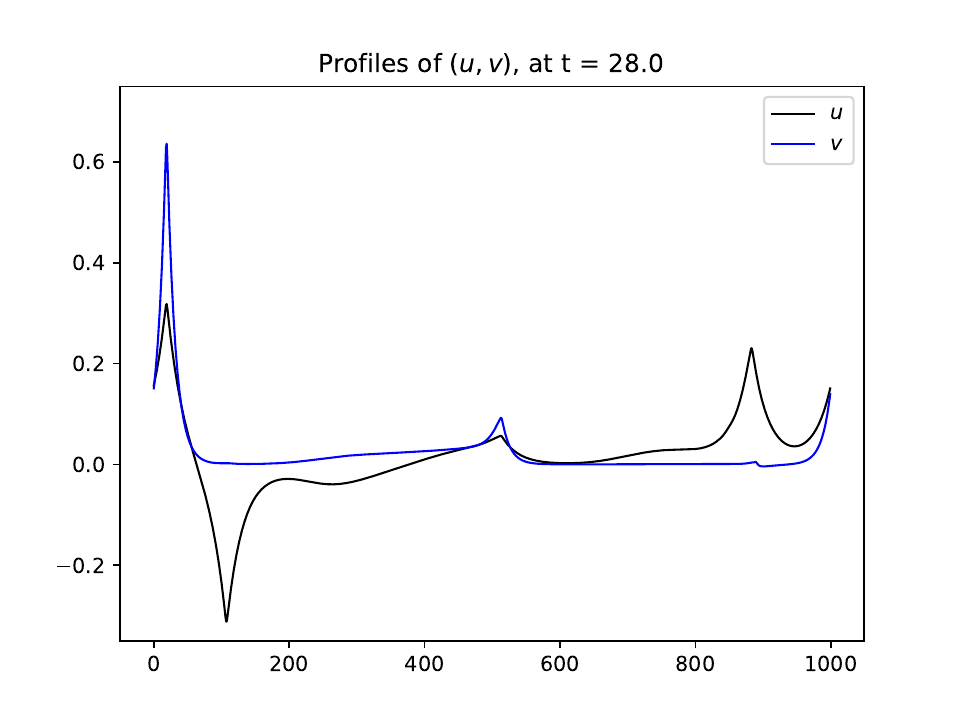}
\caption{}
\label{FIG:profile7}
\end{subfigure}
\begin{subfigure}[t]{0.3\textwidth}
        \centering
\includegraphics[width=1\textwidth]{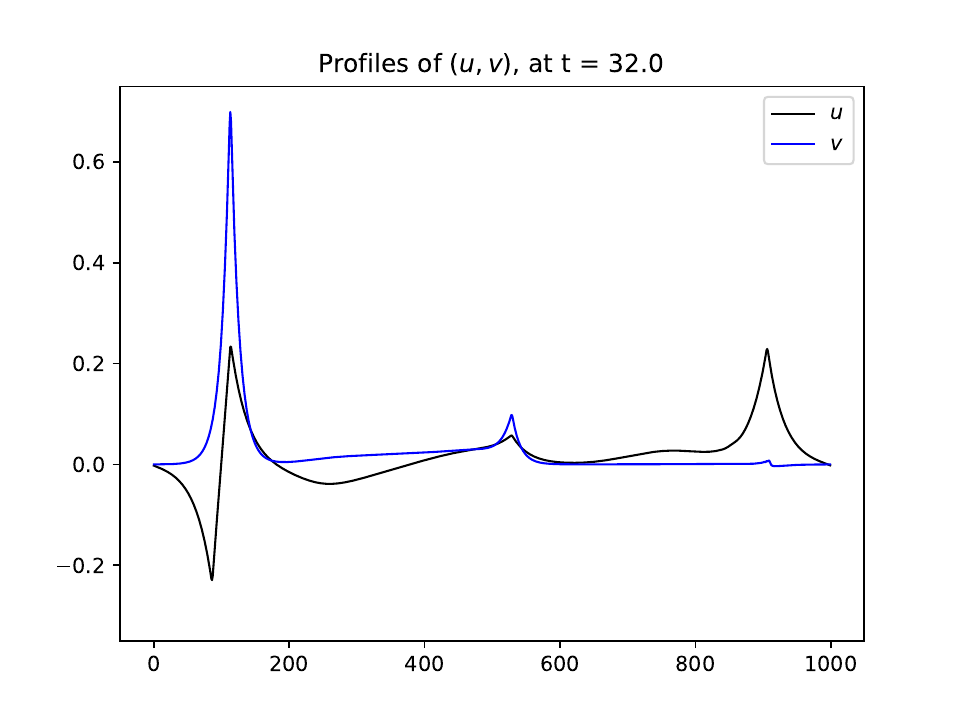}
\caption{}
\label{FIG:profile8}
\end{subfigure}
\caption{Finite element solution of the CH-CH \textbf{\emph{semidirect product}} system \cref{eq:SDP_CH_1,eq:SDP_CH_2,eq:SDP_CH_3,eq:SDP_CH_4} at $t=0,4,8,...,32$ in the variables $(u,v)$, plotted in black and blue respectively. At first, the blue initial condition $v$ created positive and negative peakons in the black $u$ variable. Later in the run, a peakon in blue and black travels with the same speed. }
\label{FIG:snapshots}
\end{figure}

\begin{figure}[H]
\centering
\begin{subfigure}[H]{0.495\textwidth}
\centering
\includegraphics[width=1\textwidth]{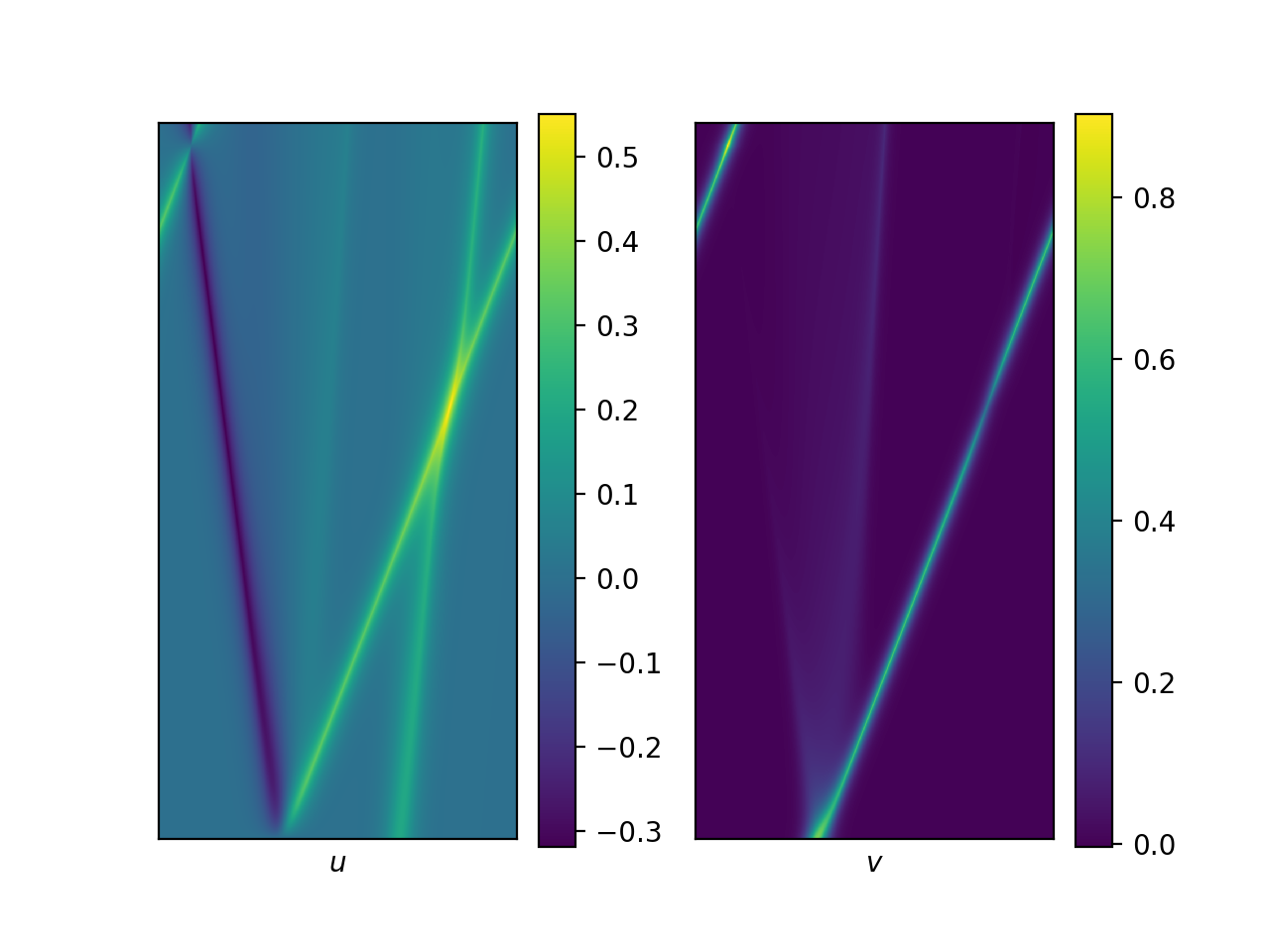}
\caption{Space-Time plot, of the coupled \textbf{\emph{semidirect product}} CH-CH system $(u,v)$ respectively.}
\label{fig:Space-Time plot}
\end{subfigure}
\begin{subfigure}[H]{0.495\textwidth}
\centering
\includegraphics[width=1\textwidth]{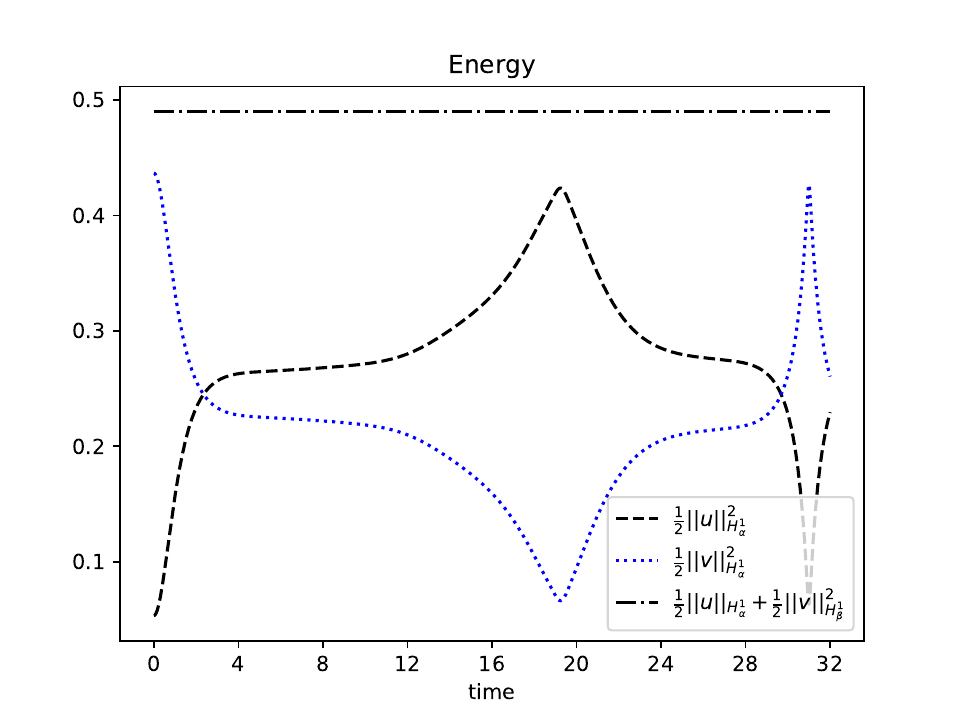}
\caption{Energy over $t\in[0,32]$ of \textbf{\emph{semidirect product}} CH-CH system}
\label{fig:Energy}
\end{subfigure}
\caption{\textbf{\emph{Semidirect product}}-CH-CH system.  Space time plot \cref{fig:Space-Time plot} and Energy plot \cref{fig:Energy}. We observe total energy preservation, and transfers of energy between the system in $u$ and $v$ associated with the nonlinear coupling.}
\label{fig:CH-spacetime-and-energy}
\end{figure}

\subsubsection{Results}
In \cref{fig:Energy} we plot discrete analogues of the energies $\frac{1}{2}||u||_{H^{1}_{\alpha}}^{2}$, $\frac{1}{2}||v||_{H^{1}_{\beta}}^{2}$, and $\frac{1}{2}||u||_{H^{1}_{\alpha}}^{2}+\frac{1}{2}||v||_{H^{1}_{\beta}}^{2}$ for the semidirect product coupled system. In \cref{fig:Energy}, the semidirect product coupled system's total energy is constant due to the discrete properties of the integrator \cite{hairer2006geometric}. In \cref{fig:Energy} we see several energy transfers as the coupled system interacts. 
In \cref{fig:DP_Energy} we plot discrete analogues of the energies $\frac{1}{2}||u||_{H^{1}_{\alpha}}^{2}$, $\frac{1}{2}||v||_{H^{1}_{\beta}}^{2}$, and $\frac{1}{2}||u||_{H^{1}_{\alpha}}^{2}+\frac{1}{2}||v||_{H^{1}_{\beta}}^{2}$ for the direct product uncoupled CH-CH system. In \cref{fig:DP_Energy}, all energies are preserved due to the integrator. No energy transfers are observed in the direct product system \cref{fig:DP_Energy}. 

Initially, for the semidirect product coupled system (on the interval $t\in [0,2]$), the $v$ system contributed energy to the $u$ system (\cref{fig:Energy}), creating positive and negative peakons (observed in \cref{FIG:profile1} at $t=4$). In \cref{FIG:profile1} $v$ created peakons and anti-peakons in $u$ whereas the $u$ system didn't transfer analogously to $v$, since $v$ acts on $u$ by forcing and $u$ acts on $v$ by transport. In comparing between the semidirect product coupled system \cref{FIG:SDP_snapshots} and the direct product uncoupled system in \cref{FIG:DP_snapshots} at $t=4$. We observe that the direct product system does not have a positive and negative peakon in the $u$ variable. Furthermore, we also observe the semidirect product coupled system has a peakon with a larger height and thinner width than the direct product system. 

In \cref{fig:Space-Time plot} we plot a space-time diagram of the coupled system, where the coupled peakon behaviour is transparent, in comparison to the direct product system space-time diagram in \cref{fig:DP_Space-Time plot}.

In \cite{escher2011euler} semidirect product CHCH equations are analytically shown to theoretically support peakons solutions from peakon initial condition. In \cref{FIG:snapshots} we observe peakon solutions, as an emergent phenomenon from non-peakon initial conditions, in addition to other nonlinear interactions. Further peakon interactions occur throughout this run, and most notably, a peakon in both the $u$ and $v$ variables appears to be travelling together. This is objectively interesting, as the peakons in $u$ and $v$ have different heights and would characteristically travel at different speeds in an uncoupled system, as observed in \cref{FIG:DP_snapshots}. 

We next investigate geodesics on $\operatorname{SDiff}(\mathbb{T}^2)\ltimes \operatorname{SDiff}(\mathbb{T}^2)$ numerically to see if this strong coupling behaviour is observed more generally in geodesic equations arising from a semidirect product of groups. 

\subsection{Geodesics on \texorpdfstring{$\operatorname{SDiff}(\mathbb{T}^2)\ltimes \operatorname{SDiff}(\mathbb{T}^2)$}{}}\label{sec:geodesics_on_Sdiff_semi_Sdiff}

Let $\ell = 1/2(||u||_2^2 + ||v||_2^2)$, the geodesic equations corresponding to the group $\operatorname{SDiff}(\mathbb{T}^2)\ltimes \operatorname{SDiff}(\mathbb{T}^2)$, can be derived by combining \cref{sec: semidirect products,sec: volume preserving diffeomorphisms} and in 2D gives a semidirect product Euler-Euler system which can be expressed in
vorticity-stream-function formulation as follows
\begin{align}
\partial_t\omega_1 + (u_1\cdot \nabla) \omega_1 + (u_2\cdot \nabla) \omega_2 = 0,\label{eq:semidirect product euler 1}\\
\partial_t\omega_2 + (u_1+u_2)\cdot \nabla \omega_2 = 0. \label{eq:semidirect product euler 2}
\end{align}
Where $u_i = -\nabla^{\perp}\psi_i$, $\omega_i = -\Delta \psi_i$ for $i=1,2$. Before solving numerically, we show that the vorticity is bounded for all time by the initial conditions.
The system can be written following \cref{eq:euler_arnold_matrix2} as 
\begin{align}
\partial_t(\omega_1-\omega_2) + (u_1 \cdot \nabla) (\omega_1-\omega_2) &= 0\\
\partial_t\omega_2 +  (u_1+u_2) \cdot \nabla \omega_2 &=0
\end{align}
This gives the following conservation laws for arbitrary $f,g$
\begin{align}
C_{f,g} = \int f(\omega_1-\omega_2) + g(\omega_2) d^2x.
\end{align}
Which imply for any norm or seminorm
\begin{align}
||\omega_1(t)-\omega_2(t)|| = ||\omega_1(0)-\omega_2(0)||,\quad 
||\omega_2(t)|| = ||\omega_2(0)||.
\end{align}
Therefore one has control over the growth of $\omega_1$, $\omega_2$. In particular using
\begin{align}
||\omega_1(t)|| \leq ||\omega_1(t)-\omega_2(t)||+||\omega_2(t)|| \leq ||\omega_1(0)-\omega_2(0)||+||\omega_2(0)||\leq ||\omega_1(0)||+2||\omega_2(0)||
\end{align}
the maximum value of vorticity can be bounded for all time by the initial condition as follows
\begin{align}
||\omega_1(t)||_{\infty} \leq ||\omega_1(0)||_{\infty} + 2||\omega_2(0)||_{\infty},  \quad ||\omega_2(t)||_{\infty}\leq ||\omega_2(0)||_{\infty}.
\end{align}
Furthermore, one could solve monotonically by discretising in the variables $(\omega_1-\omega_2,\omega_2)$, as to ensure such properties discretely, for example, using ideas presented in \cite{woodfield2024new}, using a linear invariant, nonlinear method. We numerically solve this equation in a different manner, the elliptic equations are solved spectrally using the discrete Fourier transform. The streamfunction is projected onto cell corners and then a discrete skew gradient is taken for the velocity field. This ensures a discrete divergence free property of the flow using a C-grid formulation. A flux form high-order upwind bias approach is taken to the reconstruction of the fluxes, ensuring small scale features can be resolved. The initial conditions are defined as follows,
\begin{align}
    \omega_1(x,y,0) &= 1/4 (2\cos(4 \pi x ) + \cos(2 \pi y))\label{eq:ic_q1}\\
    \omega_2(x,y,0) &=
\begin{cases}
\cos(0.5 \pi r_1/R )^2\quad \text{where}\quad r_1<R \\
\cos(0.5 \pi r_2 / R)^2\quad \text{where}\quad r_2<R \\
0,\quad\text{else}.
\end{cases} \label{eq:ic_q2}
\end{align}
Where
\begin{align}
r_1 = \sqrt{\left(x-\frac{1}{2}-\frac{1}{8}\right)^2+\left(y-\frac{1}{2}\right)^2},\quad
r_2 = \sqrt{\left(x-\frac{1}{2}+\frac{1}{8}\right)^2+\left(y-\frac{1}{2}\right)^2},\quad R = 1/8.
\end{align}

In order to facilitate a comparison, we also solve the direct product uncoupled Euler system
\begin{align}
\partial_t\omega_1 + (u_1\cdot \nabla) \omega_1 = 0,\label{eq:Euler-DP1}\\
\partial_t\omega_2 + (u_2\cdot \nabla) \omega_2 = 0, \label{eq:Euler-DP2}
\end{align}
where $u_i = -\nabla^{\perp}\psi_i$, $\omega_i = -\Delta \psi_i$ for $i=1,2$, using the same numerical method.
\smallskip
\subsubsection{Results}\label{sec:results Euler.}

For the direct product uncoupled system \cref{eq:Euler-DP1,eq:Euler-DP2}, under the initial conditions \cref{eq:ic_q1,eq:ic_q2}, fluid one undergoes mixing, and in fluid two, initially compact vortices circle one another and slowly merge. Snapshots of the \textbf{\emph{direct product}} solution at the time points at $t=0,4,...,32$ are plotted in \cref{FIG:DP_snapshots_and_qquantities} where uncoupled solutions to Euler's equation appear, and do not interact. In \cref{FIG:DP_snapshots_and_qquantities} we plot the energy and integral of vorticity where no exchange of energy or vorticity is observed. The mimetic numerical method preserves local and global conservation of vorticity, but allows energy dissipation. However, on the time window of interest, energy doesn't dissipate noticeably due to the accuracy of the scheme as seen in \cref{FIG:DP_snapshots_and_qquantities}.

In \cref{FIG:SDP_snapshots} we plot snapshots of vorticity $(\omega_1,\omega_2)$ for the \textbf{\emph{semidirect product}} system in \cref{eq:semidirect product euler 1,eq:semidirect product euler 2}, at the times $t = 0,4,..., 32$. The solution in $\omega_2$ is transported by the first fluid and prevents the merging of the vortices. The vorticity in fluid one $\omega_1$ appears to couple to the vorticity in $\omega_2$. In \cref{fig:SDP_energy-Space-Time plot}, energy is initially transferred between the fluids. 
In \cref{FIG:SDP_snapshots}, we observe vorticity in fluid one $\omega_1$ growing and becoming more concentrated in regions of high vorticity as compared with \cref{FIG:DP_snapshots_and_qquantities}.
We also observe in \cref{FIG:SDP_snapshots} higher vorticity values than in \cref{FIG:DP_snapshots_and_qquantities}, 
this is in part explained by the more relaxed growth bounds on $||\omega_1(t)||_{\infty} \leq ||\omega_1(0)||_{\infty} + 2||\omega_2(0)||_{\infty},$ for the semidirect product system, as compared with the growth bounds on the direct product system $||\omega_1(t)||_{\infty} \leq ||\omega_1(0)||_{\infty}$. Similar to the discussion on semidirect product coupled CHCH equations, initially on $t\in [0,4]$, we see that $\omega_2$ created positive and negative vorticity in $\omega_1$. In \cref{FIG:SDP_snapshots} fluid one transports the vorticity in fluid two, and fluid two forces fluid one. 
\smallskip

In summary, geodesics of semidirect product systems tend to couple flow dynamics, we observed peakons carry peakons and regions of high vorticity carry other regions of high vorticity. Fluid one transports momentum in fluid two, and fluid two forces fluid one, such that the momentum difference is transported by fluid two. 

\begin{figure}[H]
\centering
\begin{subfigure}[t]{0.495\textwidth}
\centering
\includegraphics[width=1\textwidth]{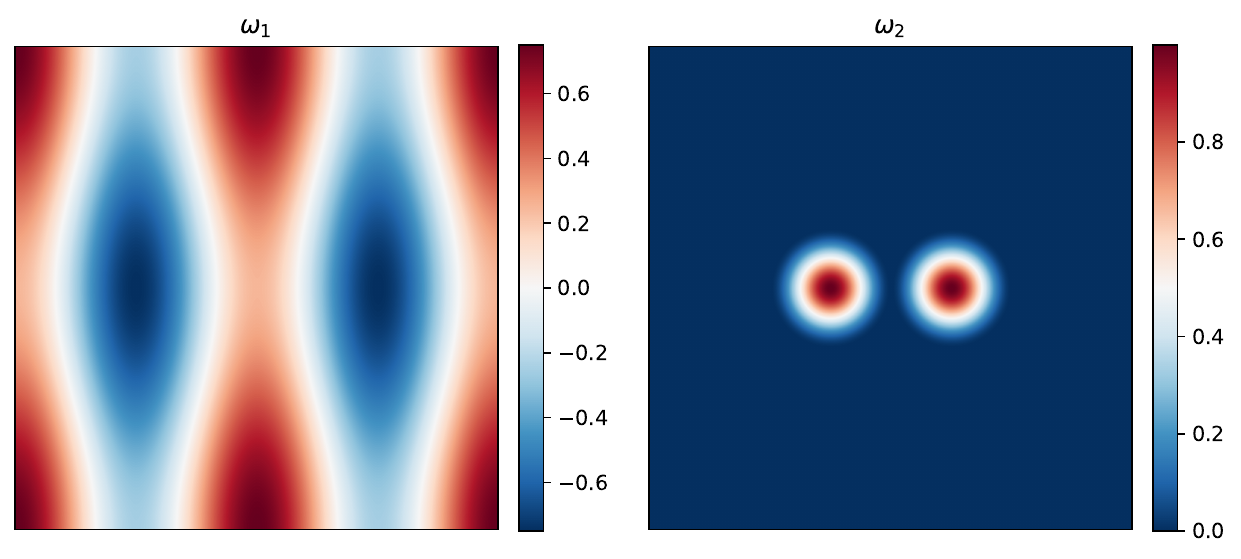}
\caption{We have $t=0$, $(\omega_1,\omega_2)$ snapshots.}\label{FIG:DP_profile0}
\end{subfigure}
\begin{subfigure}[t]{0.495\textwidth}
        \centering
\includegraphics[width=1\textwidth]{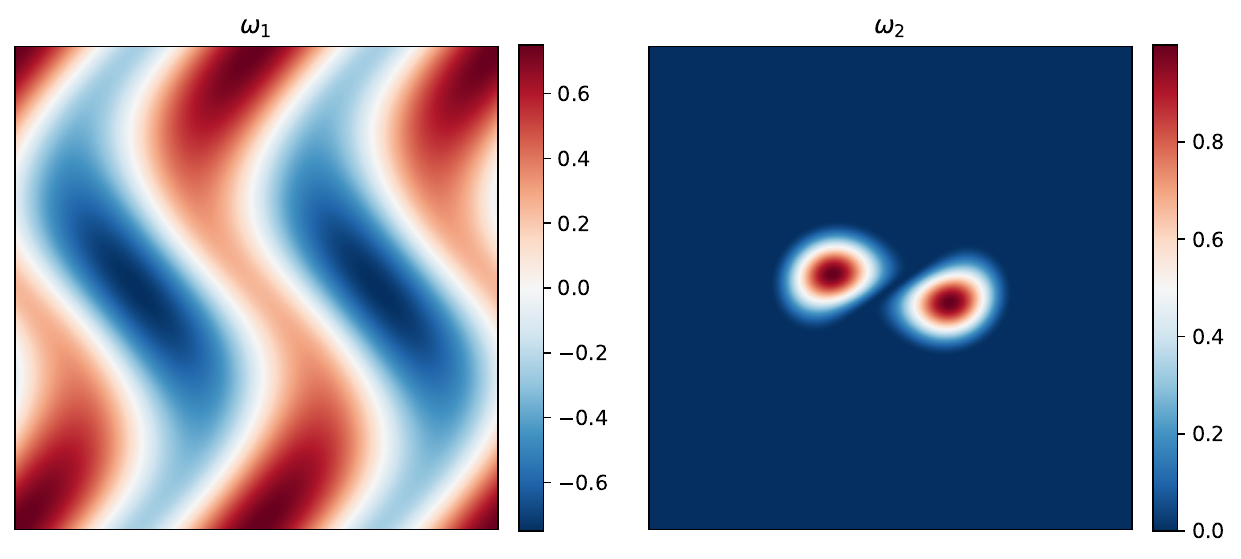}
\caption{We have $t=4$, $(\omega_1,\omega_2)$ snapshots.}\label{FIG:DP_profile1}
\end{subfigure}
\begin{subfigure}[t]{0.495\textwidth}
        \centering
\includegraphics[width=1\textwidth]{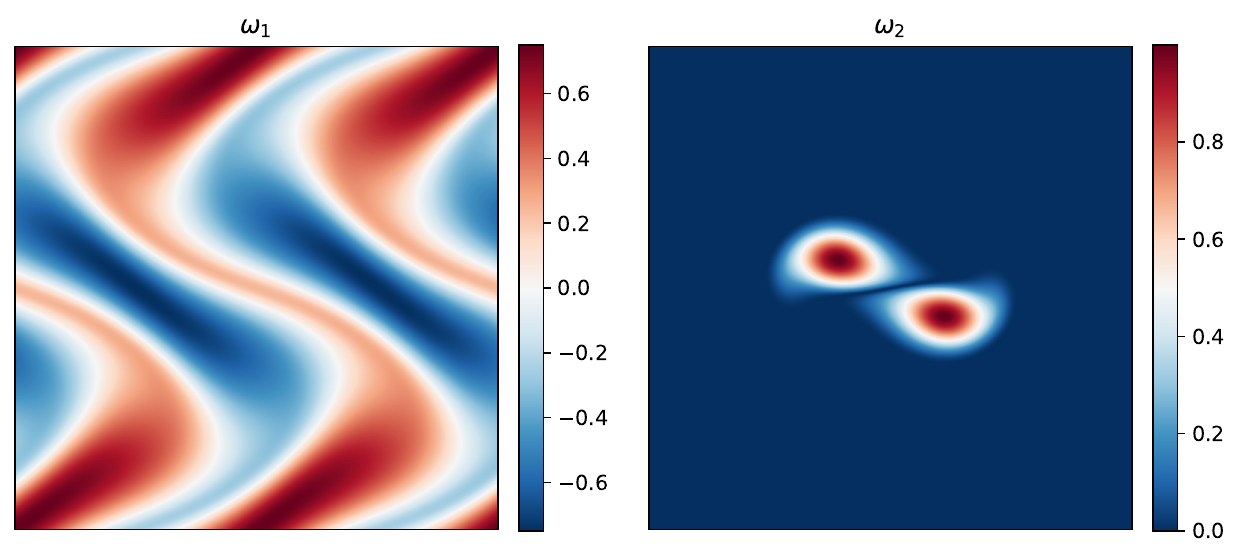}
\caption{We have $t=8$, $(\omega_1,\omega_2)$ snapshots.}
\label{FIG:DP_profile2}
\end{subfigure}
\begin{subfigure}[t]{0.495\textwidth}
        \centering
\includegraphics[width=1\textwidth]{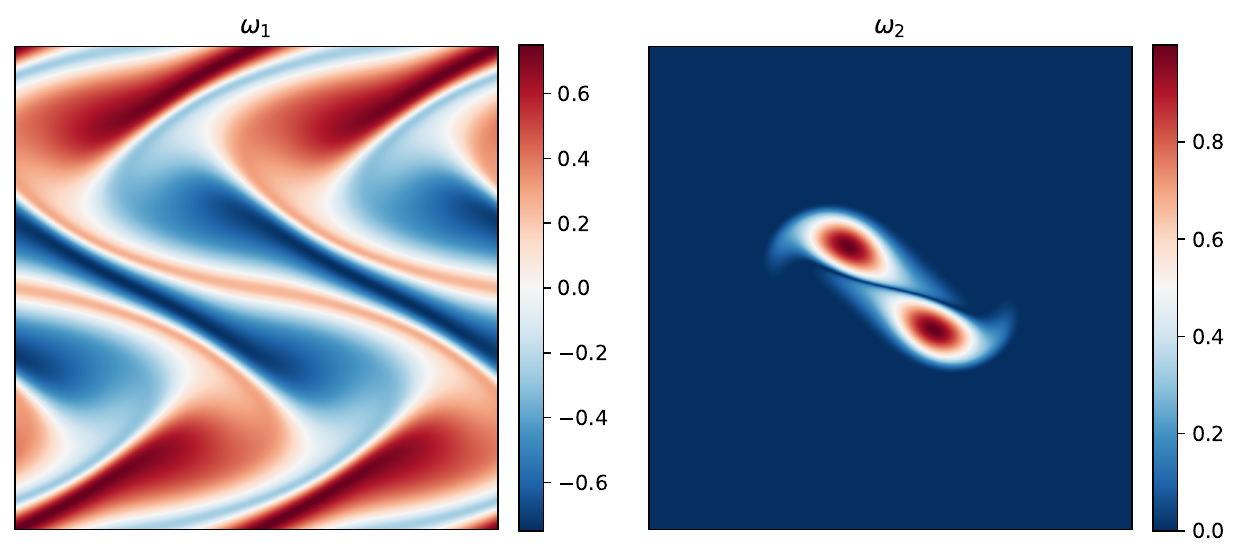}
\caption{We have $t=12$, $(\omega_1,\omega_2)$ snapshots.}
\label{FIG:DP_profile3}
\end{subfigure}
\begin{subfigure}[t]{0.495\textwidth}
        \centering
\includegraphics[width=1\textwidth]{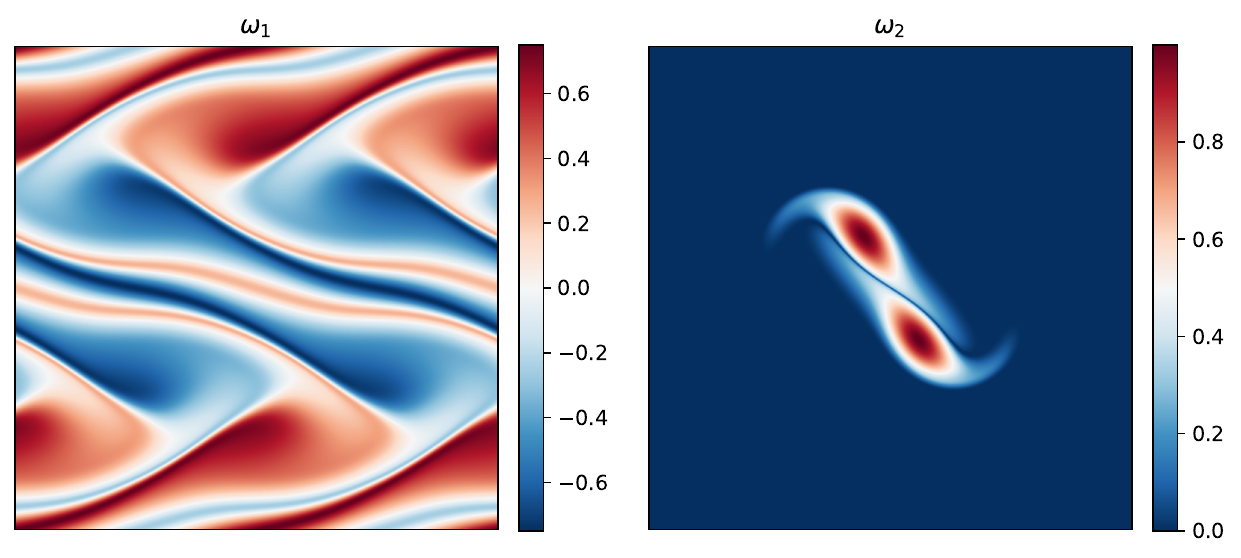}
\caption{We have $t=16$, $(\omega_1,\omega_2)$ snapshots.}
\label{FIG:DP_profile4}
\end{subfigure}
\begin{subfigure}[t]{0.495\textwidth}
        \centering
\includegraphics[width=1\textwidth]{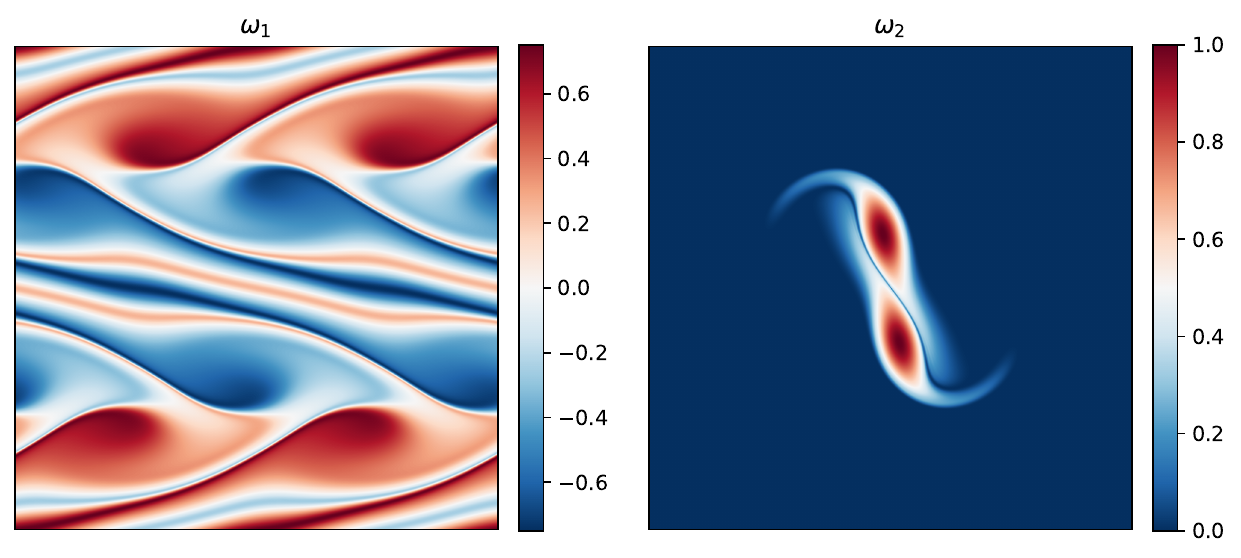}
\caption{We have $t=20$, $(\omega_1,\omega_2)$ snapshots.}
\label{FIG:DP_profile5}
\end{subfigure}
\begin{subfigure}[t]{0.495\textwidth}
        \centering
\includegraphics[width=1\textwidth]{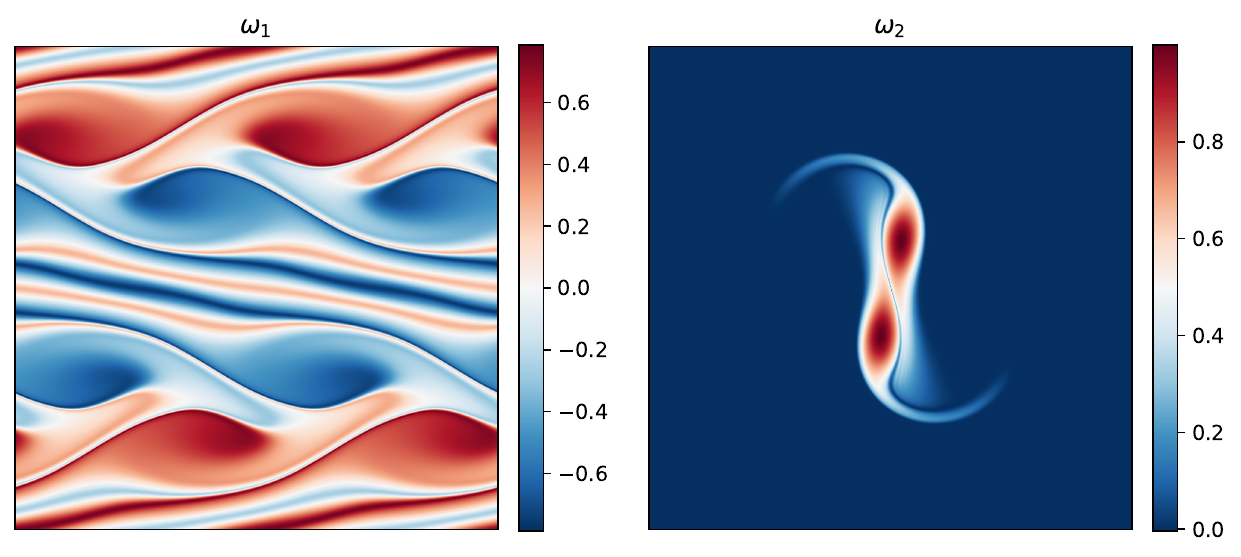}
\caption{We have $t=24$, $(\omega_1,\omega_2)$ snapshots.}
\label{FIG:DP_profile6}
\end{subfigure}
\begin{subfigure}[t]{0.495\textwidth}
        \centering
\includegraphics[width=1\textwidth]{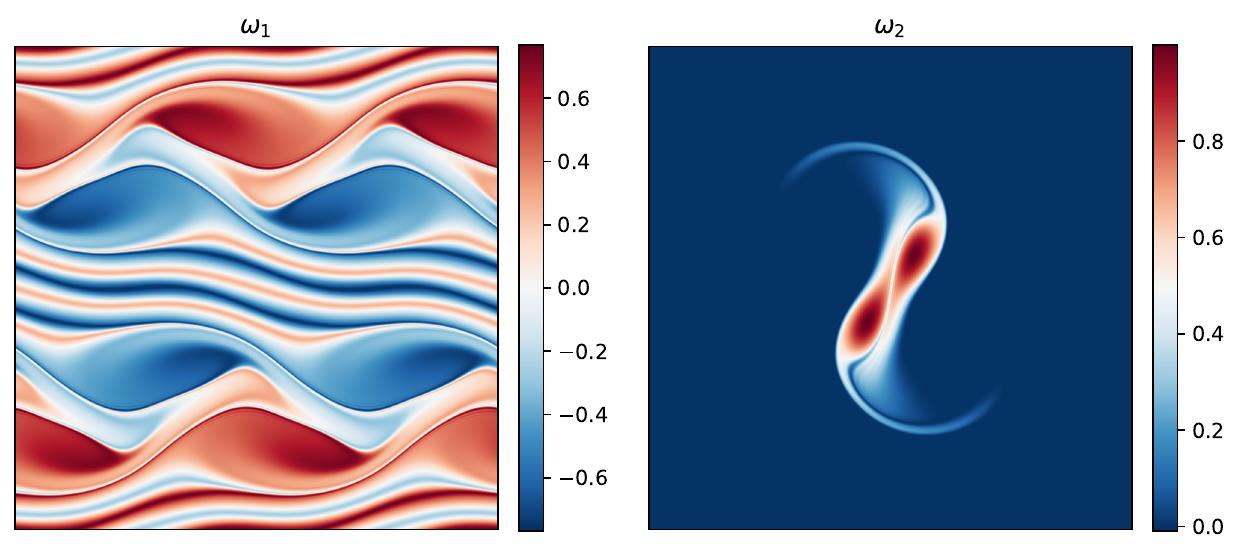}
\caption{We have $t=28$, $(\omega_1,\omega_2)$ snapshots.}
\label{FIG:DP_profile7}
\end{subfigure}
\begin{subfigure}[t]{0.495\textwidth}
        \centering
\includegraphics[width=1\textwidth]{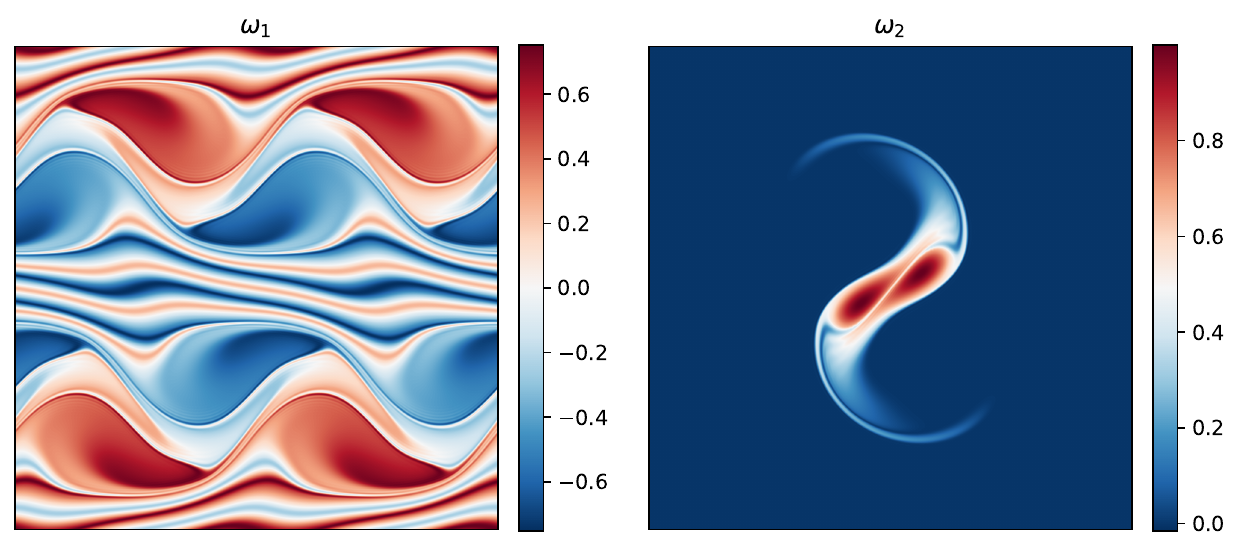}
\caption{We have $t=32$, $(\omega_1,\omega_2)$ snapshots.}
\label{FIG:DP_profile8}
\end{subfigure}
\begin{subfigure}{0.235\textwidth}
\centering
\includegraphics[width=1\textwidth]{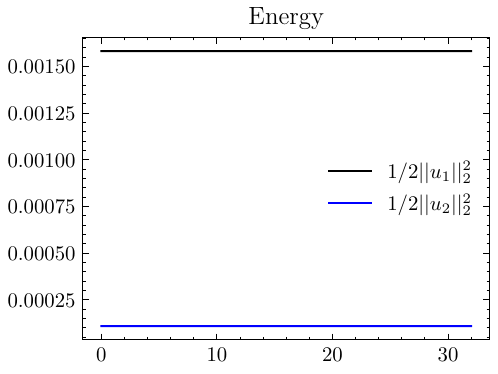}
\caption{Energy.}
\label{fig:DP_energy}
\end{subfigure}
\begin{subfigure}{0.235\textwidth}
\centering
\includegraphics[width=1\textwidth]{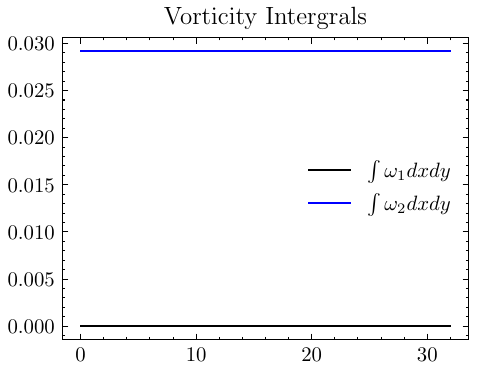}
\caption{Vorticity integrals.}
\label{fig:DP_vorticity}
\end{subfigure}
\label{fig:enter-label}
\caption{Solution of the Euler-Euler direct product uncoupled system at $t=0,4,...,32$ in the variables $(\omega_1,\omega_2)$. In $\omega_1$ a shear flow is observed, in $\omega_2$ two regions of positive vorticity interact and merge. There is no interaction between the two fluids in the direct product system. In \cref{fig:DP_energy} energy in each fluid is approximately conserved and not transfered between the fluids, in \cref{fig:DP_vorticity}
vorticity is preserved exactly.}
\label{FIG:DP_snapshots_and_qquantities}
\end{figure}

\begin{figure}[H]
\centering
\begin{subfigure}[t]{0.495\textwidth}
\centering
\includegraphics[width=1\textwidth]{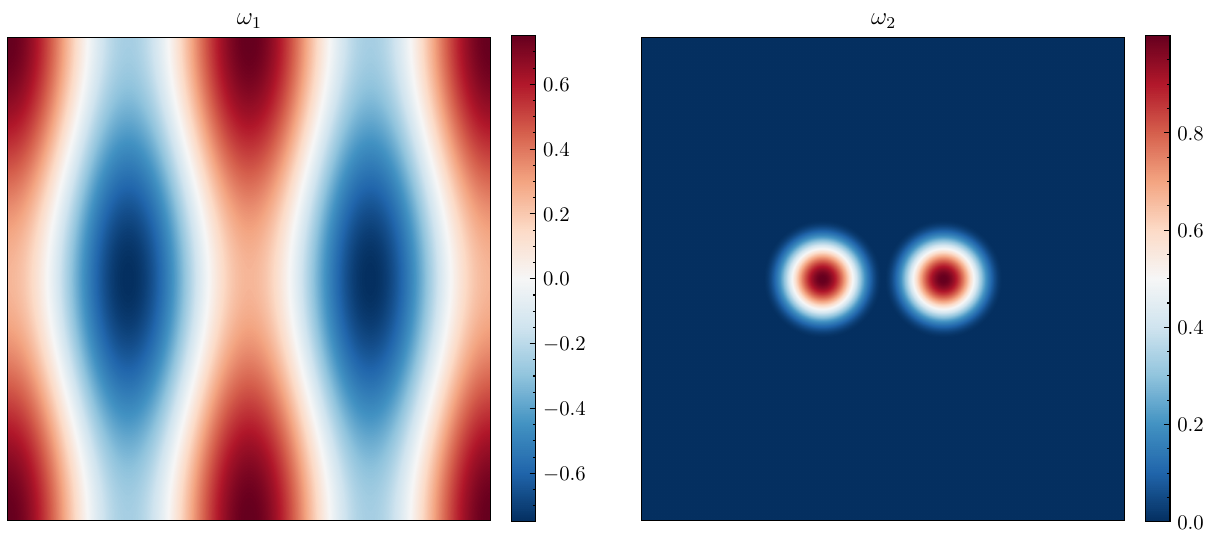}
\caption{We have $t=0$, $(\omega_1,\omega_2)$ snapshots.}\label{FIG:Euler_profile0}
\end{subfigure}
\begin{subfigure}[t]{0.495\textwidth}
        \centering
\includegraphics[width=1\textwidth]{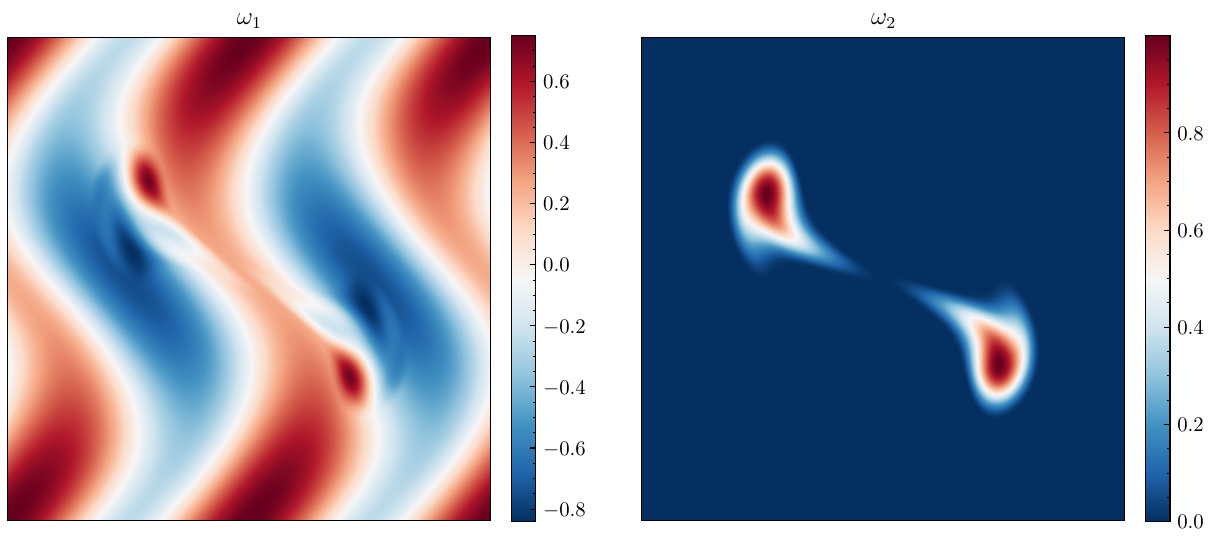}
\caption{We have $t=4$, $(\omega_1,\omega_2)$ snapshots.}\label{FIG:Euler_profile1}
\end{subfigure}
\begin{subfigure}[t]{0.495\textwidth}
        \centering
\includegraphics[width=1\textwidth]{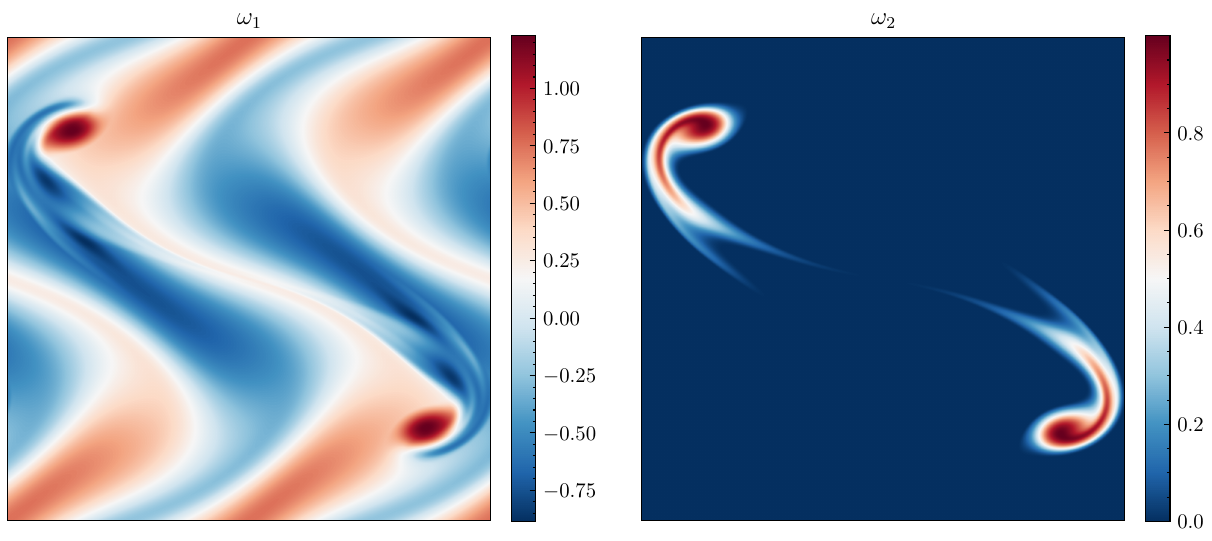}
\caption{We have $t=8$, $(\omega_1,\omega_2)$ snapshots.}
\label{FIG:Euler_profile2}
\end{subfigure}
\begin{subfigure}[t]{0.495\textwidth}
        \centering
\includegraphics[width=1\textwidth]{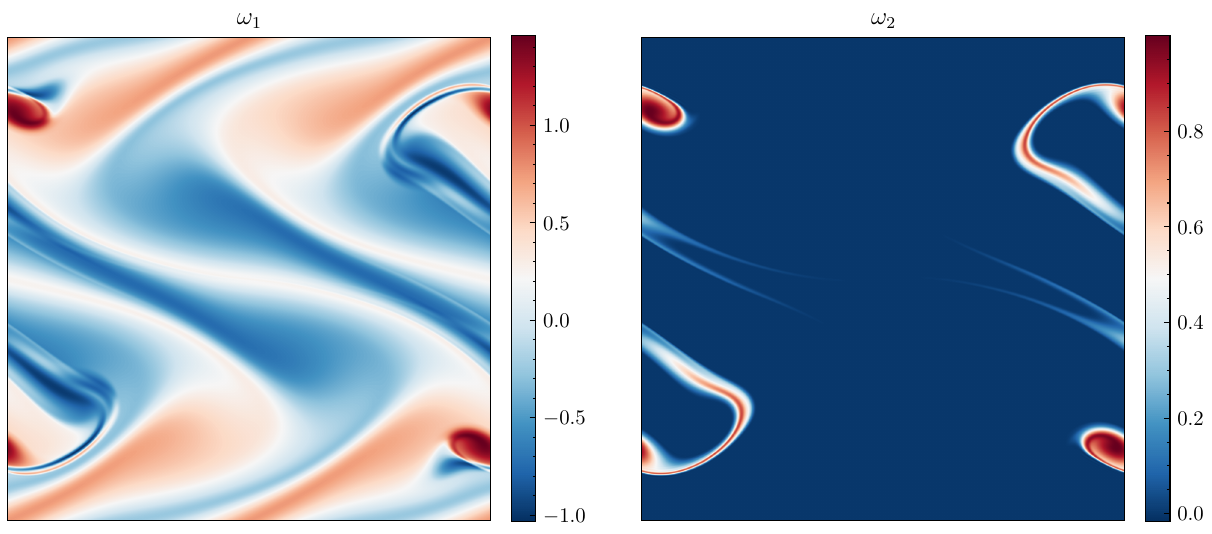}
\caption{We have $t=12$, $(\omega_1,\omega_2)$ snapshots.}
\label{FIG:Euler_profile3}
\end{subfigure}
\begin{subfigure}[t]{0.495\textwidth}
        \centering
\includegraphics[width=1\textwidth]{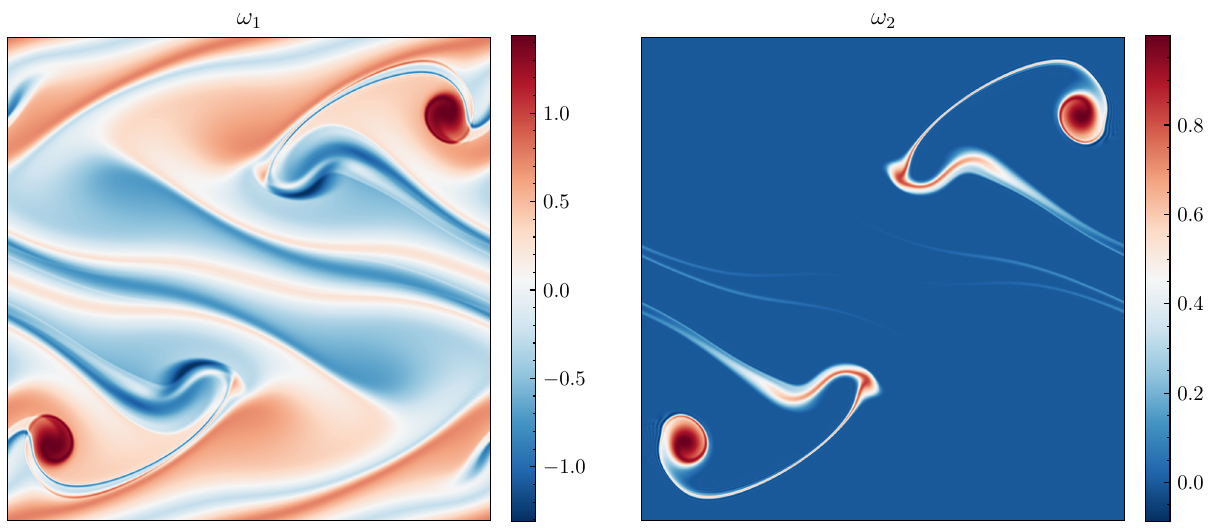}
\caption{We have $t=16$, $(\omega_1,\omega_2)$ snapshots.}
\label{FIG:Euler_profile4}
\end{subfigure}
\begin{subfigure}[t]{0.495\textwidth}
        \centering
\includegraphics[width=1\textwidth]{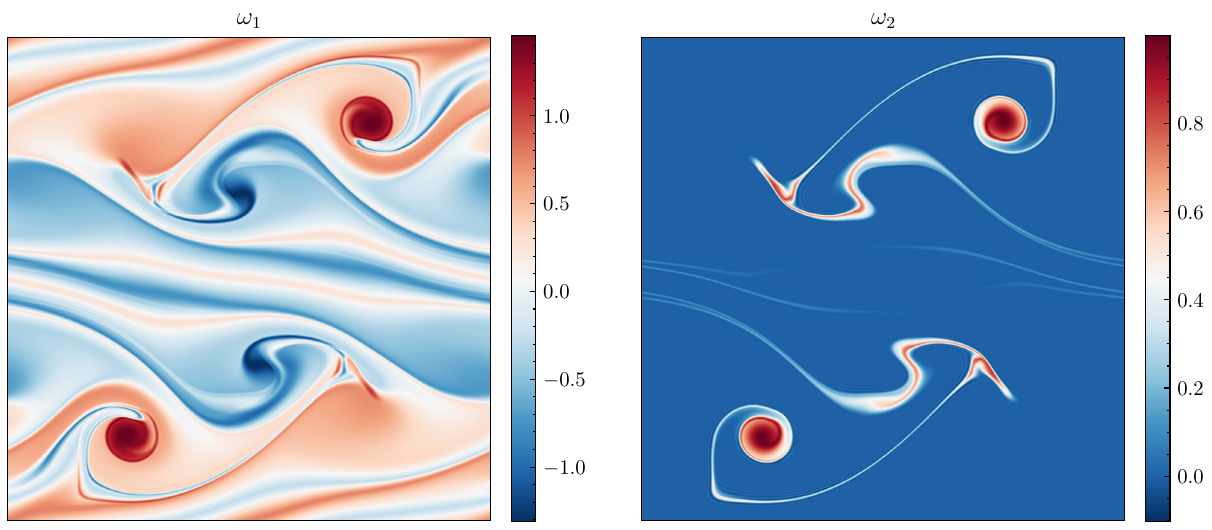}
\caption{We have $t=20$, $(\omega_1,\omega_2)$ snapshots.}
\label{FIG:Euler_profile5}
\end{subfigure}
\begin{subfigure}[t]{0.495\textwidth}
        \centering
\includegraphics[width=1\textwidth]{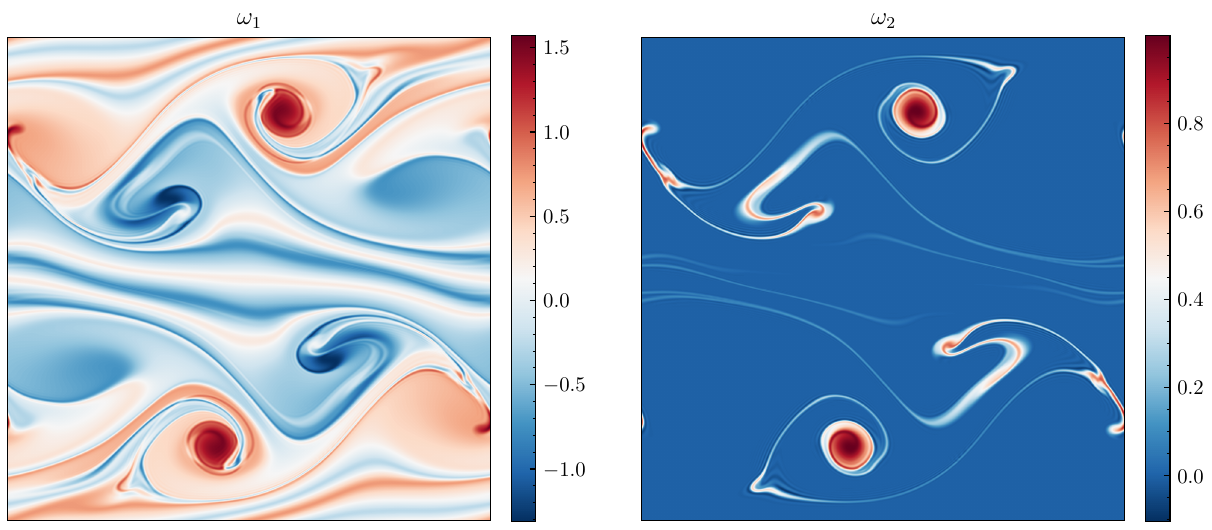}
\caption{We have $t=24$, $(\omega_1,\omega_2)$ snapshots.}
\label{FIG:Euler_profile6}
\end{subfigure}
\begin{subfigure}[t]{0.495\textwidth}
        \centering
\includegraphics[width=1\textwidth]{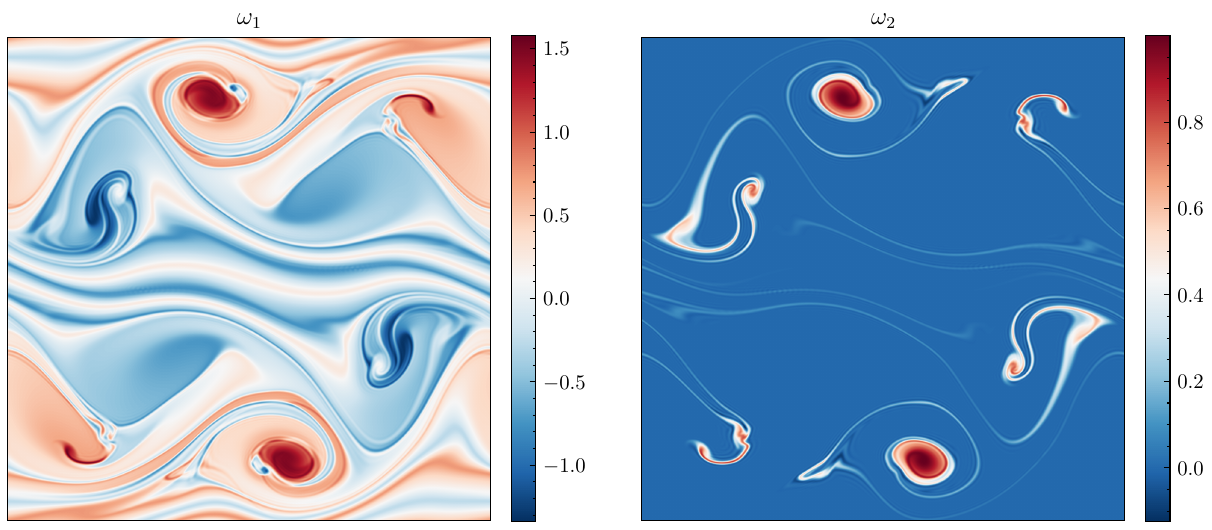}
\caption{We have $t=32$, $(\omega_1,\omega_2)$ snapshots.}
\label{FIG:Euler_profile7}
\end{subfigure}
\begin{subfigure}[t]{0.495\textwidth}
        \centering
\includegraphics[width=1\textwidth]{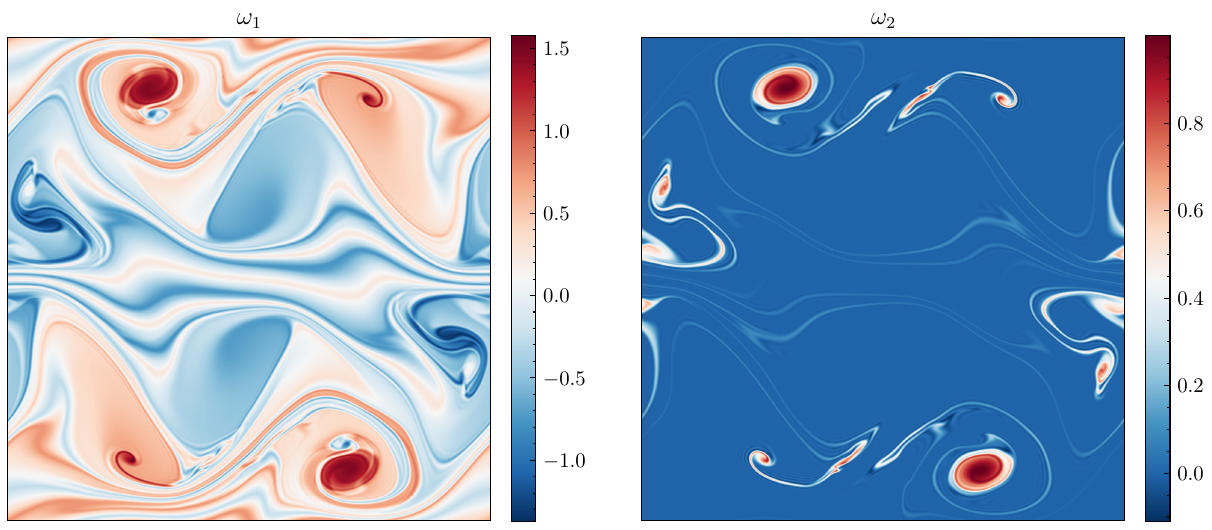}
\caption{We have $t=32$, $(\omega_1,\omega_2)$ snapshots.}
\label{FIG:Euler_profile8}
\end{subfigure}
\begin{subfigure}{0.235\textwidth}
\centering
\includegraphics[width=1\textwidth]{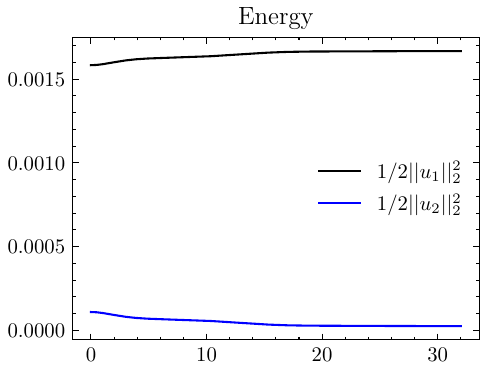}
\caption{Energy.}
\label{fig:SDP_energy-Space-Time plot}
\end{subfigure}
\begin{subfigure}{0.235\textwidth}
\centering
\includegraphics[width=1\textwidth]{pics/figure_vorticity_Euler_SDP.pdf}
\caption{Vorticity integrals.}
\label{fig:SDP_vorticity}
\end{subfigure}
\caption{Solution of the Euler-Euler semidirect product coupled system (\cref{eq:semidirect product euler 1,eq:semidirect product euler 2}) at $t=0,4,...,32$ in the variables $(\omega_1,\omega_2)$. As compared to the direct product system in \cref{FIG:DP_snapshots_and_qquantities}, there is coupling between the fluids. $\omega_2$ is transported by the the $\b u_1$ velocity, and $\omega_1$ is forced by $\b u_2 \cdot \nabla \omega_2$ term. We see the additional transport prevents the merging of the vortices in $\omega_2$. At $t=4$ we see that $\omega_2$ generates positive and negative vorticity in $\omega_1$ similar to the CH-CH system. In \cref{fig:SDP_energy-Space-Time plot} we see that the fluid 2 transfers energy to fluid 1 over this time interval. Later in the run we see vorticity in each fluid coupled together and on this time window less energy is transferred between the fluids. This is similar to the same interaction in the CH-CH system. In \cref{fig:SDP_vorticity} the integral of vorticity in each variable is preserved exactly.  }
\label{FIG:SDP_snapshots}
\end{figure}

\section{Further extensions}\label{sec:Further extensions}
We will discuss how one can use the structure of the semidirect product of groups to infer the geodesic equations of centrally extended groups $\widehat{G}\ltimes \widehat{G}$, and semidirect product groups of the type $(G\ltimes V)\ltimes (G\ltimes V)$. 

\subsection{Central extension \texorpdfstring{$\widehat{G}$}{} and \texorpdfstring{$\widehat{G} \ltimes \widehat{G}$}{}}\label{sec: semidirect product of centrall extension}

The central extension of the Lie group $(G,\cdot_{G})$, by the normal abelian Lie group $(N,+)$ is denoted $(\widehat{G},\hat{\cdot})$ and is formalised in terms of a short exact sequence $0 \rightarrow N \xhookrightarrow{i} \widehat{G} \overset{s}{\rightarrow} G \rightarrow e_G$, where $\xhookrightarrow{i}$ denotes injective homomorphism, $\overset{s}{\rightarrow}$ denotes surjective homomorphism, $0$ is the identity of $N$, $e_G$ is the identity of $G$. Such that the normal subgroup $N$ lies in the centre, where the centre of $G$ is defined as $\mathrm{Z}(G)=\lbrace z \in G | \quad \forall g \in G, z g=g z\}$, such that $N$ is invariant under left conjugation by group action (inner automorphism). The group multiplication structure is given by
\begin{align}
(g,a)\hat{\cdot}(h,b) = (gh,a + b + \Sigma(g,h)),\quad g,h\in G,\quad a,b \in \mathbb{R}.\label{eq:group product central extension}
\end{align}

Conditions on $\Sigma:G\times G \mapsto N$ are required for the centrally extended group to satisfy the usual group axioms.
\begin{proposition}\label{proposition:necessary sufficient conditions}
It is necessary and sufficient  for $\Sigma$, to satisfy the 2-co-boundary condition $\Sigma(h, e)=\Sigma(e, h)=\Sigma(e,e) = 0$ for the group $\widehat{G}$ to have identity. It is necessary and sufficient that $\Sigma\left(f, f^{-1}\right)=\Sigma\left(f^{-1},f\right)$ for all $f \in G$ for $\widehat{G}$ to have inverse.
Associativity of $\widehat{G}$ requires $\Sigma$ to satisfy the following group 2-cocycle identity
\begin{align}
\Sigma(g,h)+\Sigma(f,gh) = \Sigma(f,g)+\Sigma(fg,h). \label{eq:2cocycle identity}
\end{align}
\end{proposition}
\begin{proof}
These facts are known \cite{cendra2001lagrangian}, the details for necessary and sufficient conditions are presented in \cref{sec:GroupAxiomsCentral}.
\end{proof}

\begin{proposition}[AD-Ad-ad-joint relations on a centrally extended group and coadjoint representation.]\label{thm:Ad-central-extension}
Let $(g,a)\in \widehat{G}$, $(h,b)\in \widehat{G}$. Let $(\eta,\beta) = \frac{d}{d\epsilon}\big|_{\epsilon=0}(h_{\epsilon},b_{\epsilon})\in \widehat{\mathfrak{g}}$, $(\xi,\alpha) = \frac{d}{ds}\big|_{s=0}(g_{s},a_{s})\in \widehat{\mathfrak{g}}$ belong in the Lie algebra of $\widehat{G}$ denoted $\widehat{\mathfrak{g}} = \mathfrak{g}\oplus \mathfrak{n}$. 
Assuming the pairing between Lie-algebra and its dual as
$\langle (\xi,\alpha),(m,n)\rangle_{\widehat{\mathfrak{g}}^*} = \langle \xi,m\rangle_{\mathfrak{g}^*} + \langle \alpha ,n\rangle_{\mathfrak{n}}$. Then the adjoint and coadjoint operators are given by  
\begin{align}
\operatorname{AD}_{(g,a)}(h,b)&:= (g,a)(h,b)(g,a)^{-1} = (ghg^{-1}, b + \Sigma(g,h) -\Sigma(g,g^{-1})+ \Sigma(gh,g^{-1})), \label{eq:ADCE}\\
\operatorname{Ad}_{(g,a)}(\eta,\beta) & = (\operatorname{Ad}_{g} \eta , \beta + \partial_{2} \Sigma(g,e_G)\cdot \eta + \partial_1\Sigma(g,g^{-1})\cdot g \cdot \eta ),\label{eq:AdCE}\\
\operatorname{ad}_{(\xi,\alpha)}(\eta,\beta) &= [(\xi,\alpha),(\eta,\beta)]_{\widehat{\mathfrak{g}}} = ([\xi,\eta]_{\mathfrak{g}}, \sigma(\xi,\eta)),\label{eq:adCE}\\
\operatorname{ad}^{*}_{(\xi,a)}(m,n) &=  (\operatorname{ad}_{\xi}^{*}m + \sigma^*(\xi,n),0).\label{eq:coadCE}
\end{align}
In these equations, $[\xi,\eta]_{\mathfrak{g}}$ denotes the Lie algebra bracket (the negative commutator of vector-fields if $G$ is the diffeomorphism group), $\Sigma(g,h)$ denotes the group 2-cocycle and $\sigma(\xi,\eta):\mathfrak{g}\times\mathfrak{g}\rightarrow \mathfrak{n}$ is the Lie algebra 2-cocycle. The Lie algebra 2-cocycle can sometimes be related or calculated from a group 2-cocycle by the relation 
\begin{align}
    \sigma(\xi,\eta) = \partial_1 \partial_2 \Sigma(e_G,e_G) \cdot \eta \cdot\xi  - \partial_2\partial_1 \Sigma(e_G,e_G) \cdot\eta \cdot\xi 
    \label{eq:2cocycle-algebra-group-relationship}
    \end{align}
    where 
\begin{align}
& \partial_1 \partial_2 \Sigma(e_G,e_G) \cdot \eta \cdot \xi=\left.\left.\frac{\partial}{\partial t}\right|_{t=0} \frac{\partial}{\partial s}\right|_{s=0} \Sigma(g_t, h_s), \quad \partial_2 \partial_1 \Sigma(e_G,e_G) \cdot \eta \cdot \xi=\left.\left.\frac{\partial}{\partial t}\right|_{t=0} \frac{\partial}{\partial s}\right|_{s=0} \Sigma(h_s,g_t) .
\end{align}
Here we have used $e_G$ to denote the identity of $G$, and $\cdot$ denotes multiplication.
Note that the Lie algebra 2-cocycle $\sigma:\mathfrak{g}\times\mathfrak{g}\rightarrow \mathfrak{n}$ may exist without a corresponding group 2-cocycle. See, e.g., \cite{kriegl1997convenient} pg 506. In this case, the Lie algebra 2-cocycle is defined as algebraically satisfying sufficient conditions for the Lie algebra extension $\widehat{\mathfrak{g}}$, to remain a Lie algebra. Namely, the bilinear map $\sigma: \mathfrak{g} \times \mathfrak{g} \rightarrow \mathfrak{n}$ must be continuous, skew-symmetric and satisfy the 2-Cocycle condition 
\begin{align}
\sigma(u,[v,w]) + \sigma(v,[w,u]) + \sigma(w,[u,v]) = 0 , \quad \forall u,v,w \in \mathfrak{g} \label{eq:Algebra-2cocycle-Jacobi},
\end{align}  
which is equivalent to the centrally extended algebra satisfying Jacobi's identity.
The coadjoint contribution of the centrally extended algebra 2-cocycle is denoted $\sigma^*$ and is defined through the pairing 
\begin{align}
\langle (0, \sigma(\xi,\eta)),(m,n)\rangle_{\widehat{\mathfrak{g}}^*} = \langle n,\sigma(\xi,\eta)\rangle_{\mathfrak{n}}= \langle (\sigma^*(\xi,n),0),(\eta,\beta)\rangle_{\widehat{\mathfrak{g}}^*}. \label{eq:coadjoint contribution calculation}
\end{align} 
\end{proposition}
Proposition \ref{thm:Ad-central-extension} is exemplified in the following series of direct calculations, using a collection of standard material  \cite{michor1998geometry,arnol'd2009topological,cendra2001lagrangian,arnold2014topological,kriegl1997convenient}. Using the definition of the inverse \cref{eq:Inverse Group extension} and group action \cref{eq:group product central extension}, we attain
the ADjoint action as
\begin{align}
\operatorname{AD}_{(g,a)}(h,b) = (ghg^{-1},  b + \Sigma(g,h) -\Sigma(g,g^{-1})+ \Sigma(gh,g^{-1})),
\end{align} 
given in \cref{eq:ADCE}.
The Ad-joint action \cref{eq:AdCE} is derived, by linearising the ADjoint action \cref{eq:ADCE} on the identity of the group,
\begin{align}
\operatorname{Ad}_{(g,a)}(\eta,\beta) = \frac{d}{d\epsilon}\bigg|_{\epsilon=0} \operatorname{AD}_{(g,a)}(h(\epsilon),b(\epsilon))\quad\text{where}\quad (\eta,\beta) = (\dot{h}(0),\dot{b}(0)) \in \widehat{\mathfrak{g}}.
\end{align}
This is computed by using the chain rule to differentiate the 2-cocycle, where $h(0)=e_G$, $b(0)=0$, and $\partial_{1}\Sigma$, $\partial_{2}\Sigma$ denotes the partial derivative in the first and second component of the group 2-cocycle
\begin{align}
\frac{d}{d\epsilon} \bigg|_{\epsilon=0}\Sigma(g,h_{\epsilon}) &= \partial_{2} \Sigma(g,e_G)\cdot \eta, \quad \frac{d}{d\epsilon} \bigg|_{\epsilon=0}\Sigma(gh_{\epsilon},g^{-1}) = \partial_1\Sigma(g,g^{-1})\cdot g \cdot \eta .
\end{align}
Giving
\begin{align}
\operatorname{Ad}_{(g,a)}(\eta,\beta) = (\operatorname{Ad}_g \eta, \beta + \partial_{2} \Sigma(g,e_G)\cdot \eta + \partial_1\Sigma(g,g^{-1})\cdot g \cdot \eta ).
\end{align}

The adjoint action \cref{eq:adCE} is computed by differentiating \cref{eq:AdCE} as follows, 
\begin{align}
\operatorname{ad}_{(\xi,\alpha)}(\eta,\beta) = \frac{d}{d\epsilon}\bigg|_{\epsilon=0} \operatorname{Ad}_{(g_{\epsilon}
,a_{\epsilon})}(\eta,\beta)\quad\text{where}\quad (\xi,\alpha) = (\dot{g}(0),\dot{a}(0)) \in \widehat{\mathfrak{g}}.
\end{align}

From direct calculation one derives the adjoint relationship 
\begin{align}
\operatorname{ad}_{(\xi,\alpha)}(\eta,\beta) = \left(\operatorname{ad}_{\xi} \eta, \partial_1 \partial_2  \Sigma(e_G,e_G) \cdot \eta \cdot\xi  + \partial_1\partial_1 \Sigma(e_G,e_G) \cdot \eta \cdot \xi -\partial_2\partial_1 \Sigma(e_G,e_G) \cdot\eta \cdot\xi +\partial_1 \Sigma(e_G,e_G)\cdot \xi \cdot \eta \right).
\end{align}
Taking the derivative of the 2-co-boundary condition $\Sigma(g_{\epsilon},e_G) = \Sigma(h_{\epsilon},e_G)=0$, reveals the identities
\begin{align}
\partial_1 \Sigma(e_G,e_G)\cdot \eta=\partial_1 \Sigma(e_G,e_G)\cdot \xi  = 0.
\end{align}
Simplifying the adjoint relationship to the following
\begin{align}
\operatorname{ad}_{(\xi,\alpha)}(\eta,\beta) = (\operatorname{ad}_{\xi}\eta , \partial_1 \partial_2  \Sigma(e_G,e_G) \cdot \eta \cdot\xi  - \partial_2\partial_1 \Sigma(e_G,e_G) \cdot\eta \cdot\xi ),
\end{align}
and can be denoted as follows
\begin{align}
\operatorname{ad}_{(\xi,\alpha)}(\eta,\beta) = [(\xi,\alpha),(\eta,\beta)]_{\widehat{\mathfrak{g}}} = ([\xi,\eta]_{\mathfrak{g}}, \sigma(\xi,\eta)) = (\operatorname{ad}_{\xi} \eta, \sigma(\xi,\eta)).
\label{co-cycle ad notation}\end{align} 
Defining $\sigma^{*}(\xi,n):\mathfrak{g}\times \mathfrak{n}^{*}\rightarrow\mathfrak{g}^{*}$ as the coadjoint contribution of the 2-cocycle (\cref{eq:coadjoint contribution calculation}), one can compute
\begin{align}
\langle (m,n), \operatorname{ad}_{(\xi,a)}(\eta,\beta)\rangle_{\widehat{\mathfrak{g}}} &= \langle m , \operatorname{ad}_{\xi}\eta \rangle_{\mathfrak{g}} + \langle n, \sigma(\xi,\eta)\rangle_{\mathfrak{n}} \\
&= \langle \operatorname{ad}_{\xi}^{*}m + \sigma^*(\xi,n), \eta \rangle_{\mathfrak{g}} \\
&= \langle \operatorname{ad}_{\xi}^{*}m + \sigma^*(\xi,n), \eta \rangle_{\mathfrak{g}} + \langle 0 , \beta \rangle_{\mathfrak{n}} \\
&= \langle (\operatorname{ad}_{\xi}^{*}m + \sigma^*(\xi,n) , 0) , (\eta, \beta) \rangle_{\widehat{\mathfrak{g}}} 
\end{align}
to deduce the central extension equips the usual coadjoint operator with  $\sigma^*$ as follows, 
\begin{align}
\operatorname{ad}^{*}_{(\xi,a)}(m,n) =  (\operatorname{ad}_{\xi}^{*}m + \sigma^*(\xi,n),0).
\end{align}

The Euler-Poincaré variational principle for a centrally extended algebra is calculated, as follows.

\begin{align}
0 & =\delta \int \ell(u, a) d t  =\int\left\langle\frac{\delta \ell}{\delta(u, a)}, \delta(u, a)\right\rangle_{\widehat{\mathfrak{g}}^* \times \widehat{\mathfrak{g}}} d t \\
&=\int\left\langle\frac{\delta \ell}{\delta(u, a)}, \frac{d}{d t}(v, b)-\operatorname{a d}_{(u, a)}(v, b)\right\rangle d t \\
& =\int\left\langle-\frac{d}{d t} \frac{\delta \ell}{\delta(u, a)}-\operatorname{a d}_{(u, a)}^* \frac{\delta \ell}{\delta(u, a)},(v, b)\right\rangle d t+\left.\left\langle(v, b), \frac{\delta \ell}{\delta(u, a)}\right\rangle\right|_ 0 ^T \\
& =\int\left\langle-\frac{d}{d t}\left(\frac{\delta \ell}{\delta u}, \frac{\delta \ell}{\delta a}\right)-\operatorname{a d}_{(u, a)}^*\left(\frac{\delta \ell}{\delta u}, \frac{\delta \ell}{\delta a}\right),(v, b)\right\rangle d t+\left.\left\langle(v, b),\left(\frac{\delta \ell}{\delta u}, \frac{\delta \ell}{\delta a}\right)\right\rangle\right|_ 0 ^T
\end{align}

Here, the momentum is $(m, c)=\left(\frac{\delta \ell}{\delta u}, \frac{\delta \ell}{\delta a}\right)$ and the Euler-Poincaré equations are given by $\frac{d}{d t}(m, c) + \operatorname{ad}_{(u, a)}^*(m, c)=0$. By using \cref{co-cycle ad notation}, one finds
\begin{align}
    \frac{d}{dt} m + \operatorname{ad}^{*}_{\xi} m +\sigma^{*}(\xi,c) = 0
    \quad \hbox{and} \quad \frac{d}{dt}c= 0.
\label{eq:EP-centrally etended}
\end{align}

\begin{remark}
In matrix form \cref{eq:coadCE,eq:EP-centrally etended} can be written as a coadjoint operator acting on the dual space or in Lie Poisson form on the Lie algebra, as follows
\begin{align}
\left(\operatorname{ad}^{*}_{(\xi,a)}(m,n)\right)^{T} = 
\begin{bmatrix}
\operatorname{ad}^{*}_{\xi}(\cdot) &\sigma^{*}(\xi , \cdot)\\
0&0
\end{bmatrix}
\begin{bmatrix}
m\\
n
\end{bmatrix}
= 
\begin{bmatrix}
\operatorname{ad}^{*}_{(\cdot)}m+\sigma^{*}(\cdot,n)&0\\
0&0
\end{bmatrix}
\begin{bmatrix}
\xi\\
a
\end{bmatrix}.
\end{align}
\end{remark}
We can inspect the coadjoint and Lie-Poisson structure for the semidirect product $(\widehat{G}\ltimes \widehat{G})$ Euler-Poincaré equations following \cref{sec: semidirect products} 
as follows
\begin{align}
\frac{d}{dt}
\begin{bmatrix}
m\\
l\\
n\\
o
\end{bmatrix}
&= -
\begin{bmatrix}
\operatorname{ad}_{(\cdot)}^*m + \sigma^*(\cdot,l)& 0& \operatorname{ad}_{(\cdot)}^*n + \sigma^*(\cdot,o)&0\\
0&0&0&0\\
\operatorname{ad}_{(\cdot)}^*n+ \sigma^*(\cdot,o)& 0&\operatorname{ad}_{(\cdot)}^*n+ \sigma^*(\cdot,o)&0\\
0&0&0&0
\end{bmatrix}
\begin{bmatrix}
    \xi_h\\
    \alpha\\
    \xi_k \\
    \beta
\end{bmatrix}
= -
\begin{bmatrix}
\operatorname{ad}^{*}_{\xi_{h}}(\cdot) &
\sigma^{*}(\xi_h,\cdot)
&\operatorname{ad}^{*}_{\xi_{k}}(\cdot)
&\sigma^{*}(\xi_k,\cdot)\\
0&0&0&0\\
0&0&\operatorname{ad}^{*}_{\xi_{h}+\xi_{k}}(\cdot)& \sigma^*(\xi_{h} + \xi_{k},\cdot)\\
0&0&0&0
\end{bmatrix}
\begin{bmatrix}
m\\
l\\
n\\
o
\end{bmatrix}\label{eq:semidirectproduct_of_central_extension}
\end{align}

\subsection{Example 0: Geodesics on \texorpdfstring{$\widehat{\operatorname{Diff}}(S^1)$}{}}\label{sec:Example 0: Geodesics}
It is known, that the diffeomorphism group on the circle admits the central extension
$
\widehat{\operatorname{Diff}}(S^1)=\operatorname{Diff}\left(S^1\right) \times S^1
$, through the Thurston-Bott group 2-cocycle \cite{bott1977characteristic}, 
\begin{align}
\Sigma\left(g_1, g_2\right):=\frac{1}{2} \int_{S^1} \log \left(\frac{\partial}{\partial x}\left(g_1 \circ g_2\right)\right) \mathrm{d} \log \left(\frac{\partial g_2}{\partial x}\right). \label{eq:TB-2cocycle}
\end{align}
The Virasoro algebra $\mathfrak{vir}(S^1) = \operatorname{Vect}(S^1)\oplus S^1$ is a unique (up to isomorphism) non-trivial (modulo 2-co-boundaries) central extension of $\operatorname{Vect}(S^1)$, it also is the Lie algebra of the Lie group $\widehat{\operatorname{Diff}}\left(S^1\right)$. The Virasoro algebra has the Gefland-Fuchs algebra 2-cocycle \begin{align}
\sigma(\xi,\eta) = \int_{S^1} \xi_x \eta_{xx}\mathrm{d} x.  \label{eq:GF-2cocycle}
\end{align}
The relationship between the Thurston-Bott group 2-cocycle $\Sigma$, and the Gefland-Fuchs algebra 2-cocycle $\sigma$ is given in \cref{eq:2cocycle-algebra-group-relationship} and calculated explicitly in \cite{khesin2009geometry} (proposition 2.4). By direct calculation (integration by parts) under the pairing in \cref{eq:coadjoint contribution calculation}, the coadjoint contribution of the algebra 2-cocycle is 
\begin{align}
    \sigma^{*}(\xi,n) = n\xi_{xxx}.
\end{align} 
This is well established, see for example \cite{holm2018new,segal1991geometry,ovsienko1987korteweg,michor1998geometry,holm2018new,khesin2009geometry}, and leads to the following equation
\begin{align}
\partial_t m + (um)_x + u_x m + lu_{xxx} =0,
\quad m = \frac{\delta \ell}{\delta u}.
\end{align}
Where $\partial_t l = 0$, examples include the KdV and Camassa-Holm wave models.

\subsection{Example 1: Geodesics on \texorpdfstring{$\widehat{\operatorname{Diff}}(S^1)\ltimes \widehat{\operatorname{Diff}}(S^1)$}{}}\label{sec:Example 1: Geodesics}
Let $((f,a),(g,b)) \in \widehat{\operatorname{Diff}}_{H}(S^1)\ltimes \widehat{\operatorname{Diff}}_{K}(S^1)$ where $\widehat{\operatorname{Diff}}(S^1)$ is the central extension of $\operatorname{Diff}(S^1)$, defined by fibring the group over its identity element, the subscripts $H$, $K$ are defined to specify the semidirect group product structure. 

 
Since, this is a composition of two group extensions, through central extension and then semidirect product, we can inspect the following extended Lie-Poisson structure on $\mathfrak{vir}(S^1)\ltimes \mathfrak{vir}(S^1)$, and its formulation in its dual space $(\mathfrak{vir}(S^1)\ltimes \mathfrak{vir}(S^1))^*$ as, critical points of the following action principle
\begin{align}
0=\delta \int \ell(u,a,v,b) dt =\int \left\langle 
-\frac{d}{dt}\left(\frac{\delta\ell}{\delta u},\frac{\delta\ell}{\delta a},\frac{\delta\ell}{\delta v},\frac{\delta\ell}{\delta b}\right) - \operatorname{ad}^{*}_{(u,a,v,b)}\left(\frac{\delta\ell}{\delta u},\frac{\delta\ell}{\delta a},\frac{\delta\ell}{\delta v},\frac{\delta\ell}{\delta b}\right)
, (w,c,z,d)\right\rangle dt + \text{B.C's}
\end{align}
Where we denote the momenta $(m, l,n,o)=\left( \frac{\delta \ell}{\delta u}, \frac{\delta \ell}{\delta \alpha}, \frac{\delta \ell}{\delta v},\frac{\delta \ell}{\delta \beta} \right)$. Then by using the general structure in \cref{eq:semidirectproduct_of_central_extension}
and identifying $\operatorname{ad}^*_u m = (um)_x + u_x m$,  $\sigma^*(u,l) = l u_{xxx}$, one may derive the following Euler-Poincaré equations
\begin{align}
\partial_t {m} + (um)_x + u_x m + lu_{xxx} + (vn)_x + v_x n + o v_{xxx} = 0, \quad
\partial_t{l}=0,\\
\partial_t{n} + ((u+v)n)_x + (u+v)_x n + o (u+v)_{xxx} = 0, \quad \partial_t{o}=0.
\end{align}
Where $l,o$ are arbitrary constants.
Defining $m-n = D$, $l-o = E$, and subtracting one row from the other leads to the following system
\begin{align}
\frac{d}{dt}(D) + \operatorname{ad}^{*}_{u}(D) +\sigma^*(u,E) &= 0,\\
\frac{d}{dt}(n) + \operatorname{ad}^{*}_{u+v}(n) +\sigma^*(u+v,o) &= 0.
\end{align}
One then deduces that the difference in momentum variables $m-n$ is evolved by another geodesic flow on $\mathfrak{vir}(S^1)$. Some examples of semidirect product coupled systems are given below,
\begin{example}
Burgers-Burgers, $\operatorname{Diff}(S^1)\ltimes \operatorname{Diff}(S^1)$ , 
$\ell = 1/2 \left( ||u||_2^2 + ||v||_2^2 \right)$ 
\begin{align}
\partial_t u + 3uu_x + 3vv_x  &= 0,\\
\partial_t v + 3vv_x +2u_x v + u v_x & =0.
\end{align}
\end{example}
\begin{example}
KdV-KdV, $\widehat{\operatorname{Diff}}(S^1)\ltimes \widehat{\operatorname{Diff}}(S^1)$, $\ell  = 1/2||u||_2^2 + 1/2||v||_2^2 + 1/2 \alpha^2 + 1/2 \beta^2$.
\begin{align}
\partial_t u + 3uu_x +lu_{xxx}+ 3vv_x +o v_{xxx} &= 0,\\
\partial_t v+3vv_x +2u_x v + u v_x + o(u+v)_{xxx}& =0.
\end{align}
\end{example}
\begin{example}
KdV-Burgers ($o = 0$)
\begin{align}
\partial_t u + 3uu_x +lu_{xxx}+ 3vv_x  &= 0,\\
\partial_t v+3vv_x +2u_x v + u v_x  & =0. 
\end{align}
\end{example}
\begin{example}
Burgers-KDV ($l = 0$)
\begin{align}
\partial_t u + 3uu_x + 3vv_x +o v_{xxx} &= 0,\\
\partial_t v+3vv_x +2u_x v + u v_x + o(u+v)_{xxx}& =0,
\end{align}
\end{example}

\begin{example}
CH-CH ($\widehat{\operatorname{Diff}}(S^1)\ltimes \widehat{\operatorname{Diff}}(S^1)$ , $\ell  = 1/2||u||_{H^1_{\alpha}}^2 + 1/2||v||_{H^1_{\beta}}^2 + 1/2 \alpha^2 + 1/2 \beta^2$)
\begin{align}
    \partial_t m + (um)_x + u_x m + lu_{xxx} + (vn)_x + v_x n + o v_{xxx} &= 0, \quad m = u-\alpha u_{xx},\\
\partial_t n + ((u+v)n)_x + (u+v)_x n + o (u+v)_{xxx} &= 0,\quad n = v-\beta v_{xx},
\end{align}
\end{example}

\begin{example}[Burgers-CH]
\begin{align}
\partial_t u + 3uu_x + 3vv_x  &= 0,\\
\partial_t n + ((u+v)n)_x + (u+v)_x n + o (u+v)_{xxx} &= 0, \quad n = v- \beta v_{xx},
\end{align}
\end{example}

\begin{example}[CH-Burgers]
\begin{align}
\partial_t m + (um)_x + u_x m + lu_{xxx} + (vn)_x + v_x n + o v_{xxx} &= 0,\quad m = u-\alpha u_{xx},\\
\partial_t v + ((u+v)v)_x + (u+v)_x v  &= 0.\\
\end{align}
\end{example}
Here one notes there are additional coupling terms arising from the central extension, most visible in the Burgers-KdV equation where dispersive coupling arises. In \cref{Centrally extending the semidirect product group}, we consider a 2-component ``central extension" of the semidirect product Lie algebra which does not give such dispersive coupling terms as to exemplify noncommutativity of group extensions.

\subsection{Semidirect products \texorpdfstring{$G\ltimes V$}{}, and \texorpdfstring{$(G\ltimes V)\ltimes (G\ltimes V)$}{}} \label{sec:semi_direct_GV}

Adjoint actions associated with the group $G\ltimes V$ can be found in detail in \cite{marsden1984semidirect,guillemin1982geometric}, we briefly recall how these are used in the construction of the Euler-Poincaré equations on $G\ltimes V$, as to introduce notation required for its further self-semidirect product extension.

Let $G$ be a Lie group acting on the left on the vector space $V$. The outer semidirect product group $G\ltimes V$ can be constructed from the semidirect product group action
\begin{align}
(g,v)\cdot_{\ltimes} (h,w) = (g h,v + g w).
\end{align}
Where we employ concatenation notation to denote left group action of $G$ on $G$ by $gh := g\cdot_{G}h$, and $gw=g\cdot w = \rho(g)w$, is the left action of $G$ on $V$. The identity of the group is $(e_G,0)$. The inverse of the group is $(g^{-1},-g^{-1}v)$. One can directly compute the $\operatorname{AD}$joint action
\begin{align}
\operatorname{AD}_{(g,v)}(h,w) = (g,v)\cdot_{\ltimes} (h,w)\cdot_{\ltimes}(g^{-1},-g^{-1}v) &= (ghg^{-1}, v+g\cdot w + (gh)\cdot (-g^{-1}\cdot v))\\
&= (\operatorname{AD}_{g}h, v+ g\cdot w - (\operatorname{AD}_{g}h) \cdot v ).
\end{align}
We have used the fact that the left representation $\rho$ has the following property
\begin{align}
\rho(g_1)\rho(g_2)v = \rho(g_1 g_2)v.
\end{align}

By differentiating $\operatorname{AD}_{(g,v)}(h(t),w(t))$, with respect to $t$, and evaluating at $t=0$, where $(h(0),w(0)) = (e_G,0)$ is the identity of the group, and $(\dot{h}(0),\dot{w}(0)) = (\eta_{h},\eta_{w})\in \mathfrak{g}\ltimes V$ is in the algebra we can find the $\operatorname{Ad}$joint representation of the group on the algebra,
\begin{align}
\operatorname{Ad}_{(g,v)}(\eta_{h},\eta_{w}) = (\operatorname{Ad}_{g}\eta_{h}, g\cdot \eta_{w} - \operatorname{Ad}_{g}\eta_{h} \cdot v),
\end{align}
where the first dot in $g\cdot \eta_{w}$ denotes how $G$, acts on $V$, specified by the representation $\rho$. The second dot denotes how $\mathfrak{g}$ acts on $V$, this is an induced representation denoted $\alpha$. 
We now denote the left representation of the group on the vector space, and the induced Lie algebra representation on the vector space as follows,
\begin{align}
\rho: G \times V \rightarrow V, \quad  \rho (g, v) = \rho (g) v = (g \cdot v),\\
\alpha: \mathfrak{g} \times V \rightarrow V,\quad \alpha (\eta, v)= \alpha(\eta) v=(\eta \cdot v).
\end{align}

Next we compute the adjoint action
\begin{align}
\operatorname{ad}_{(\xi_g,\xi_v)}(\eta_h,\eta_w) &= \frac{d}{dt}\bigg|_{t=0} \operatorname{Ad}_{(g(t),v(t))}(\eta_h,\eta_w) \\ 
&=\frac{d}{dt}\bigg|_{t=0}\left( \operatorname{Ad}_{g(t)}\eta_h , g(t) \cdot \eta_{w} - (g(t)\eta_h g^{-1}(t)) \cdot v(t) \right),\\
&= (\operatorname{ad}_{\xi_g}\eta_h, \xi_{g}\cdot \eta_{w} - \frac{d}{dt}\bigg|_{t=0}(g(t)\eta_h g^{-1}(t)) \cdot v(t))\\
&= (\operatorname{ad}_{\xi_g}\eta_h, \xi_{g}\cdot \eta_{w} - \alpha([\xi_g,\eta_h]) 0 - \eta_h  \cdot \xi_v)\\
&= (\operatorname{ad}_{\xi_g}\eta_h , \xi_{g}\cdot \eta_{w} - \eta_h  \cdot \xi_v)
\end{align}
Where the induced Lie algebra representation on the vector space identity element gives no contribution. For notational convenience denote $\xi_v,\eta_w = v,w\in V$ and $\xi_g,\eta_h=\xi,\eta \in \mathfrak{g}$, and write the adjoint action as
\begin{align}
\operatorname{ad}_{(\xi,v)}(\eta,w)  
&= (\operatorname{ad}_{\xi}\eta , \xi \cdot w -\eta \cdot v).
\end{align}

Let $(m,n)\in \mathfrak{g}^{*}\times V^*$, $(\xi,v)\in \mathfrak{g}\times V$, $(\eta,w)\in \mathfrak{g}\times V$,
then defining $\diamond:V^*\times V\mapsto \mathfrak{g}$ as the operator satisfying $\langle n,\alpha(\eta)v\rangle_{V^*\times V} = -\langle n\diamond v,\eta \rangle_{\mathfrak{g}^*\times\mathfrak{g}}$ and denote $\alpha^*:\mathfrak{g}\times V^* \mapsto V^*$ as the operator satisfting $\langle n,\alpha(\xi)w\rangle_{V^*\times V} = \langle \alpha^*(\xi)n,w\rangle_{V^*\times V}$, then with the pairing 
\begin{align}
\langle (m,n),(\xi,v) \rangle_{\mathfrak{g}\times V} = \langle m,\xi \rangle_{\mathfrak{g}} + \langle n,v \rangle_{V},
\end{align}
one attains the coadjoint motion
\begin{align}
\langle \operatorname{ad}^{*}_{\xi,v}(m,n),(\eta,w)\rangle &= \langle (m,n),\operatorname{ad}_{\xi,v}(\eta,w) \rangle\\
&=\langle (m,n),(\operatorname{ad}_{\xi}\eta , \xi \cdot w -\eta \cdot v)\rangle\\
&=\langle m ,\operatorname{ad}_{\xi}\eta\rangle +\langle n, \xi \cdot w -\eta \cdot v)\rangle\\
&=\langle m ,\operatorname{ad}_{\xi}\eta\rangle_{\mathfrak{g}^*\times \mathfrak{g}} +\langle n, \alpha(\xi) w\rangle_{V^*\times V} -\langle n,\alpha(\eta)v\rangle_{V^*\times V}\\
&=\langle \operatorname{ad}^*_{\xi} m ,\eta\rangle_{\mathfrak{g}^*\times \mathfrak{g}} +\langle \alpha^*(\xi)n,  w\rangle_{V^*\times V} + \langle n \diamond v,\eta\rangle_{\mathfrak{g}^*\times \mathfrak{g}}\\
&=\langle (\operatorname{ad}^*_{\xi} m +n \diamond v ,\alpha^*(\xi)n ),(\eta,w)\rangle_{(\mathfrak{g}^*\times V^*)\times(\mathfrak{g}\times V)}
\end{align}

The contragradient representation $\alpha^*$ of $\xi\in \mathfrak{g}$ on $n\in V^*$ and the $\diamond:V^*\times V\mapsto \mathfrak{g}^*$ operator are typically determined in a case-by-case setting, dependent on the Lie algebra and vector space. 

\begin{remark}
In matrix form the coadjoint operator \cref{eq:coadCE} can be written as acting on the dual space $\mathfrak{g}^* \ltimes V^*$ or in Lie Poisson form on $\mathfrak{g} \ltimes V$ as follows
\begin{align}
\left(\operatorname{ad}^{*}_{(\xi,v)}(m,n)\right)^{T} = 
\begin{bmatrix}
\operatorname{ad}^{*}_{\xi}(\cdot) & (\cdot)\diamond v\\
0&\alpha^*(\xi)(\cdot)
\end{bmatrix}
\begin{bmatrix}
m\\
n
\end{bmatrix}
= 
\begin{bmatrix}
\operatorname{ad}^{*}_{(\cdot)}m& n \diamond (\cdot) \\
\alpha^*(\cdot)n&0
\end{bmatrix}
\begin{bmatrix}
\xi\\
v
\end{bmatrix}.
\end{align}
\end{remark}


We consider the group  $(G\ltimes V)\ltimes (G\ltimes V)$ and inspect the following Lie-Poisson structure and coadjoint representation of the Euler-Poincaré equation from arguments in \cref{sec: semidirect products}
\begin{align}
\frac{d}{dt}
\begin{bmatrix}
m\\
l\\
n\\
o
\end{bmatrix}
&= -
\begin{bmatrix}
\operatorname{ad}_{(\cdot)}^*m  & l\diamond (\cdot)& \operatorname{ad}_{(\cdot)}^*n & o\diamond (\cdot)\\
\alpha^*(\cdot)l &0&\alpha^*(\cdot)o&0\\
\operatorname{ad}_{(\cdot)}^*n & o\diamond (\cdot)&\operatorname{ad}_{(\cdot)}^*n&o\diamond (\cdot)\\
\alpha^*(\cdot)o &0&\alpha^*(\cdot)o&0\\
\end{bmatrix}
\begin{bmatrix}
    \xi_h\\
    f\\
    \xi_k \\
    h
\end{bmatrix}
= -
\begin{bmatrix}
\operatorname{ad}^{*}_{\xi_h}(\cdot) & (\cdot)\diamond f & \operatorname{ad}^{*}_{\xi_k}(\cdot) & (\cdot)\diamond h\\
0&\alpha^*(\xi_h)(\cdot)& 0 & \alpha^*(\xi_k)(\cdot)\\
0&0&\operatorname{ad}^{*}_{\xi_{h}+\xi_{k}}(\cdot)& (\cdot)\diamond (f+h)\\
0&0&0&\alpha^*(\xi_h+\xi_k)(\cdot)
\end{bmatrix}
\begin{bmatrix}
m\\
l\\
n\\
o
\end{bmatrix}. \label{eq:semidirectproduct_of_semidirect_products}
\end{align}

\begin{example}
Consider the group $\operatorname{Diff}(S^1)\ltimes C^{\infty}(S^1)$, with Lie algebra $\operatorname{Vect}(S^1)\ltimes C^{\infty}(S^1)$, with adjoint action 
\begin{align}
\operatorname{ad}_{(u(x)\partial_x,f(x))}(v(x)\partial_x,g(x)) = -([u,v]\partial_x, ug_x-v f_x).
\end{align}
Under the $L^2$ pairing 
\begin{align}
\langle (m\mathrm{d}x\otimes \mathrm{d}x,l \mathrm{d}x),(v\partial_x,g) \rangle = \int m(x)v(x) + l(x)g(x) dx, \label{eq:l2 pairing}
\end{align}
one derives the coadjoint operator on $\left(\Lambda^1(S^1)\otimes\Lambda^1(S^1)\right)\times \Lambda^1(S^1)$ as,
\begin{align}
\operatorname{ad}^*_{(u\partial_x,f)}(m \mathrm{d}x^2,l\mathrm{d}x) = \left( (mu)_x + mu_x + lf_x, (lu)_x\right).
\end{align}
This is expressible in matrix form as 
\begin{align}
\left(\operatorname{ad}^{*}_{(\xi,f)}(m,l)\right)^{T} = 
\begin{bmatrix}
\mathcal{L}_{\xi}(\cdot) & (\cdot)\partial_x f\\
0&\partial_x(\cdot \xi)
\end{bmatrix}
\begin{bmatrix}
m\\
l
\end{bmatrix}
= 
\begin{bmatrix}
\mathcal{L}_{(\cdot)}(m \,\mathrm{d}x^2)& l \partial_x (\cdot) \\
\partial_x (\cdot l)&0
\end{bmatrix}
\begin{bmatrix}
\xi\\
f
\end{bmatrix}.
\end{align}
\end{example}

\begin{example}Using the above example and the general form in \cref{eq:semidirectproduct_of_semidirect_products} one can inspect the geodesics (Euler-Poincaré equations) for $(\operatorname{Diff}(S^1)\ltimes C^{\infty}(S^1))\ltimes (\operatorname{Diff}(S^1)\ltimes C^{\infty}(S^1))$, under the canonical pairing
\begin{align}
\left\langle \left((m \mathrm{d}x\otimes \mathrm{d}x,l\mathrm{d}x),(n\mathrm{d}x\otimes \mathrm{d}x,o\mathrm{d}x)\right),((u\partial_x,f),(v\partial_x,g) ) \right\rangle = \int_{S^1} m(x)u(x) +n(x)v(x) + l(x)f(x) + o(x)g(x) dx, 
\end{align}
between Lie algebra $(\operatorname{Vect}(S^1)\ltimes C^{\infty}(S^1))\ltimes (\operatorname{Vect}(S^1)\ltimes C^{\infty}(S^1))$ and dual $((\Lambda^1(S^1)\otimes\Lambda^1(S^1))\times \Lambda^1(S^1))^{\times 2}$. As
the coupled equations
\begin{align}
\partial_t m + (um)_x + mu_x + l f_x + (vn)_x + nv_x + o g_x = 0,\\
\partial_t l + \partial_x(ul) + \partial_x(vo) = 0,\\
\partial_t n + ((u+v)m)_x + m(u_x+v_x) + o (f_x+g_x)=0,\\
\partial_t o + \partial_x ((u+v)o) = 0.
\end{align}
Where $(m,n)= \left(\frac{\delta\ell}{\delta u},\frac{\delta\ell}{\delta v}\right)$.
A self-semidirect product coupled system of Itô equations \cite{ito1982symmetries} emerges for the $L^2$ norm $(m=u,n=v)$ and a self-semidirect product coupled system of CH2-equations for the $H^1$ norm when $(m=u-u_{xx},n=v-v_{xx})$.
\end{example}

We elaborate on the coordinate-free Ideal Incompressible MHD equations in arbitrary dimensions derived in \cite{arnol'd2009topological,khesin1989invariants}, as to then model its direct product as Hall-MHD and its self-semidirect product equations, in \cref{sec:MHD}. Similar ideas are presented in \cite{araki2015differential,holm202431}, regarding Lie-Poisson structures for MHD.

\section{Conclusion}\label{sec:conclusion}
Summary: We have used an energy-preserving monolithic continuous Galerkin finite element code to numerically investigate the equations and coupling mechanism associated with the geodesic flow associated with the semidirect product of the diffeomorphism group on itself. We observed coupled peakons of different heights travelling together with the same speed in the numerical solution, and emergent peakon behaviour. We investigated the coupling phenomenon numerically for the semidirect product of the volume preserving diffeomorphism group with itself in two spatial dimensions, where similar nonlinear coupling is observed in the vorticity variables. The momentum  difference is carried by velocity in fluid one
and the sum of velocities carries the momentum in fluid two.
\smallskip

Open problems: The numerical study of the semidirect product of groups could in the future be extended to the case where the group has been augmented by group extension or semidirect product. The coupled systems of CH-CH, CH-Burgers, Burgers-KdV, KdV-Burgers, KdV-KdV, Burgers-Burgers, and other coupled equations of Itô and CH2 type in \cref{sec:Further extensions} may be interesting for further study, as well as the Euler-EPDiff or MHD-MHD equations. Group extensions typically do not commute. Consequently, the central extension of the semidirect product does not necessarily equal the semidirect product of centrally extended groups. The implications of this fact about central extensions of semidirect product groups are explored in the case of vector spaces in \cite{ovsienko1994extensions}, we preliminarily explore this in \cref{Centrally extending the semidirect product group} for the central extension.
\smallskip

Further challenges: Additional work is required to understand the physical relevance of such models. In particular, different choices of the Lagrangian in semidirect product systems results in a variety of equations whose physical applications may be interesting to identify. The usual Euler-Poincaré equations are recovered from the semidirect product Euler-Poincaré equations in \cref{sec: composition of maps}, by parametrising the metric one may link the semidirect product sytem to the original EP equation. It is also possible to extend the scope of this work beyond geodesic flow, and adopt the more general framework of reduction by stages  \cite{hochschild1965structure,kupershmidt1983canonical,ortega2013momentum,marsden2007hamiltonian}, whose application is more general.

\section*{Acknowledgements}
JW happily acknowledges many helpful and enjoyable discussions with Darryl Holm, Oliver Street and Ruiao Hu. In addition to several other useful discussions with H. Wang and T. Diamantakis. JW gratefully acknowledges the use of the SciencePlots library \cite{SciencePlots} by John D. Garrett. 
During this work, JW has been partially supported by the European Research Council (ERC) Synergy grant ``Stochastic Transport in Upper Ocean Dynamics" (STUOD) -- DLV-856408.

\bibliography{reference}
\bibliographystyle{abbrv}
\appendix

\begin{appendices}
\section{Appendix}
\subsection{Group Homomorphism}

A fundamental construction in this work is a group homomorphism, it is defined as follows.
\begin{definition}[Group Homomorphism]\label{def:group homomorphism}
Let $(H,\cdot_H)$,$(K,\cdot_K)$ be groups, a Group Homomorphism from $(H,\cdot_H)$, to $(K,\cdot_K)$, is a function $\phi:H\rightarrow K$, such that $\forall h_1,h_2\in H$ we have
\begin{align}
    \phi(h_1\cdot_H h_2) = \phi(h_1)\cdot_K \phi(h_2).
\end{align}
\end{definition}
Group Homomorphisms inherit the following properties, 
\begin{enumerate}
    \item $\phi\left(e_H\right)=e_K$.
    \item $\phi\left(h^{-1}\right)=\phi(h)^{-1}$.
\end{enumerate}
from the group homomorphism property. $\phi$ maps the unit element $e_H$ of $H$ to the unit element $e_K$ of $K$, this property follows from the homomorphism property $
\phi \left(e_H\right)=\phi\left(e_H \cdot_{H} e_H\right)=\phi\left(e_H\right) \cdot_K \phi \left(e_H\right)
$. $\phi$ maps the inverse in $H$ to inverse in $K$, this follows from the previous property since
$
\phi(h) \cdot_{K} \phi\left(h^{-1}\right)=\phi\left(h \cdot_{H} h^{-1}\right)=\phi\left(e_H\right)=e_K.$

\subsection{Group Axioms: Centrally extended group}\label{sec:GroupAxiomsCentral}
We verify that group 2-cocycle conditions on $\Sigma$, are necessary and sufficient for $\widehat{G}$ to satisfy the group axioms. 
\begin{enumerate}
\item Identity. 
For $(g,b)$ to be the left and right identity of $\widehat{G}$, the following conditions must hold
\begin{align}
(f,a)=(f,a)\hat{\cdot}(g,b) = (fg,a+b+\Sigma(f,g)) \iff g=e_G,\quad \text{and}\quad b=-\Sigma(f,e),\\
(f,a)= (g,b)\hat{\cdot}(f,a) = (gf,b+a + \Sigma(g,f)) \iff g=e_G,\quad \text{and} \quad b = -\Sigma(e,f).
\end{align}
Therefore $-b=\Sigma(f,e)=\Sigma(e,f),\quad\forall f\in G$. Consider the following property of the identity
\begin{align}
(e,-\Sigma(e,f)) = (e,-\Sigma(e,f))\hat{\cdot}(e,-\Sigma(e,f)) = (ee,-\Sigma(e,f) -\Sigma(e,f) + \Sigma(e,e))\quad \iff \quad \Sigma(e,e)=\Sigma(e,f).
\end{align} 
Therefore $(e,-\Sigma(e,e))$ is the identity, and $-b=\Sigma(e,e)=\Sigma(e,f)=\Sigma(f,e)$. Now consider the element $(e,0)\in \widehat{G}$.
\begin{align}
(e,0)= (e,0)\hat{\cdot}(e,-\Sigma(e,e)) = (e,-\Sigma(e,e)).
\end{align}
So that $\Sigma(e,e)= 0 = \Sigma(e,f)=\Sigma(f,e)$ is a neccessary and sufficient condition on the Group 2-cocycle for the group $\widehat{G}$ to have a unique identity element
\begin{align}
e_{\widehat{G}} = (e_G,0)
\end{align}. 

\item Inverse. 
From the group composition $(f,a)\hat{\cdot}(g,b) = (fg,a+b+\Sigma(f,g))$, right identity of $(g,b)$ requires $g=f^{-1}$, and then $a+b+\Sigma(f,f^{-1})=0$. Left identity of $(g,b)\hat{\cdot}(f,a) = (gf,b+a+\Sigma(g,f))$ requires $g = f^{-1}$ and $0=b+a+\Sigma(f^{-1},f)$. Therefore, since
$b = -a -\Sigma(f,f^{-1})$ and $b = -a -\Sigma(f^{-1},f)$, we must require that for uniqueness of the inverse $\Sigma(f^{-1},f)=\Sigma(f,f^{-1})$.
So the inverse is \begin{align}
(f,a)^{-1}  = (f^{-1},-a - \Sigma(f,f^{-1}))\in \widehat{G}. \label{eq:Inverse Group extension}
\end{align}
\item Associativity. Associativity in $\widehat{G}$ requires
\begin{align}
(f,a)\hat{\cdot}((g,b)\hat{\cdot}(h,c)) = ((f,a)\hat{\cdot}(g,b))\hat{\cdot}(h,c) \label{eq:associativity of central extension}
\end{align}
Consider the right hand side of \cref{eq:associativity of central extension} 
\begin{align}
(fg,a+b+\Sigma(f,g))\hat{\cdot}(h,c) = (fgh,a+b+c+\Sigma(f,g)+\Sigma(fg,h))
\end{align}
and equate it to the left hand side of \cref{eq:associativity of central extension}, 
\begin{align}
 (f,a)\hat{\cdot}(gh,b+c+\Sigma(g,h)) = (fgh,a+b+c+\Sigma(g,h)+\Sigma(f,gh))
\end{align}
observe the 2cocycle condition \cref{eq:2cocycle identity}.
\end{enumerate}

\subsection{Group Axioms: Semidirect product group}\label{sec:group-axioms}
We verify that group identity, inverse and associativity of the semidirect product group
\begin{enumerate}
    \item Identity. $e_{G} = (e_H,e_K)$, is the identity, as can be verified through the group homomorphism properties 
\begin{align}
(e_H,e_K)\cdot_{\ltimes}(h,k) = (e_{H}\cdot_{H} H , e_K \cdot_{K}\phi_{e_H}k)= (h,k) \\
(h,k)\cdot_{\ltimes}(e_H,e_K) = (h\cdot_{H}e_H , k \cdot_{K}\phi_{h}e_k)= (h,k) 
\end{align}
\item Inverse.
The inverse can be infered from the group multiplication structure \cref{eq:group product structure}
\begin{align}
    g^{-1} = (h,k)^{-1} = (h^{-1},h^{-1}\cdot k^{-1}) \label{eq:inverse}
\end{align}
The right inverse property is verified using \cref{eq:property1} as follows
$gg^{-1} = (h,k)(h^{-1},h^{-1}\cdot k^{-1}) = (e,k (h\cdot (h^{-1}\cdot k^{-1}))) = (e,e)$. The left inverse property is verified using \cref{eq:property2} as follows $g^{-1}g = (h^{-1},h^{-1}\cdot k^{-1})(h,k) = (e,h^{-1}\cdot k^{-1}(h^{-1} \cdot k) )= (e,(h^{-1}\cdot k^{-1})(h^{-1} \cdot k ) ) = (e,e)$. 
\item Associativity.
The group associativity property holds by direct computation
\begin{align}
g_1(g_2 g_3) &= g_1(h_2,k_2)(h_3,k_3) = g_1 (h_2 h_3 , k_2 (h_2 \cdot k_3) ) =  (h_1 h_2 h_3 , k_1 h_1 \cdot ( k_2 (h_2 \cdot k_3) ) ) \\
(g_1 g_2) g_3 &= (h_1,k_1)(h_2,k_2)g_3 = (h_1 h_2 , k_1 (h_1 \cdot k_2))g_3
= (h_1 h_2 h_3 , k_1 (h_1 \cdot k_2) (h_1 h_2 \cdot k_3)), \label{eq:SDP_associativity_eq2}
\end{align}
the right hand sides are equivalent under the use of \cref{eq:property2}. 
\end{enumerate}

\subsection{Adjoint calculations for a semidirect product of groups}\label{sec:adjoints for H semi K}

\begin{proposition}
[AD-Ad-adjoint for semidirect product of groups]
Let $G=H\ltimes K$, then the AD-Ad-ad-joint actions are
\begin{align}
\operatorname{AD}_{(h_1,k_1)}(h_2,k_2)
&= (\operatorname{AD}_{h_1}h_2,k_1 \cdot_{K} (\phi_{h_1} k_2 ) \cdot_{K}(\phi_{\operatorname{AD}_{h_{1}} h_{2} } k_{1}^{-1} )),\label{eq:AD_new1}\\
\operatorname{Ad}_{(h,k)}{(\eta_h,\eta_k)} 
&=(\operatorname{Ad}_{h}\eta_{h}, \operatorname{Ad}_{k}(\Phi_{h}\eta_{k})  + \Psi_{k} (\operatorname{Ad}_{h}\eta_h)  - \operatorname{Ad}_{h}\eta_h),\label{eq:AD_new2}\\
\operatorname{ad}_{(\xi_{h},\xi_{k})}(\eta_h,\eta_k) &= ( \operatorname{ad}_{\xi_h} \eta_h, \xi_h \cdot \eta_k + \xi_k \cdot \eta_h + \operatorname{ad}_{\xi_k} \eta_k)\label{eq:AD_new3}.
\end{align}
Here the $\cdot$ denotes Lie algebra action. 
If it is assumed $K = H$, then the Adjoint action simplifies as follows
\begin{align}
\operatorname{Ad}_{g}{\eta_g} &= (\operatorname{Ad}_{h}\eta_{h},\operatorname{Ad}_{kh}(\eta_k+\eta_h) -\operatorname{Ad}_{h}\eta_h) \\
\operatorname{ad}_{\xi_{g}}\eta_{g} 
&= (\operatorname{ad}_{\xi_h}\eta_h, \operatorname{ad}_{\xi_k+\xi_h} (\eta_{h}+\eta_{k})- \operatorname{ad}_{\xi_h}\eta_h )
\end{align}

If it is assumed that the inner product structure is of the form
\begin{align}
    \langle (m,n),(u,v) \rangle_{\mathfrak{g}^*\times\mathfrak{g}} = \langle m,u\rangle_{\mathfrak{h}^*\times\mathfrak{h}}+ \langle n,v\rangle_{\mathfrak{k}^*\times\mathfrak{k}}.
\end{align}
Then the co-adjoint action is 
\begin{align}
\operatorname{ad}^{*}_{(\xi_h,\xi_k)} (m,n) &= (\operatorname{ad}^{*}_{\xi_h} m+ \operatorname{ad}_{\xi_k}^{*} n, \operatorname{ad}_{(\xi_k + \xi_h)}^{*} n  ). \label{eq:adstar}
\end{align}
\end{proposition}
The above proposition is known in the more general case and detailed exposition on \cref{eq:AD_new1,eq:AD_new2,eq:AD_new3} can be found in \cite{hochschild1965structure,marsden2007hamiltonian,ortega2013momentum}, below we detail calculations required for the case $H=K$ derivation of \cref{eq:adstar}.
The ADjoint action of $G$ on $G$ can be calculated from recalling the group associative identity \cref{eq:SDP_associativity_eq2} $g_1 g_2 g_3 = (h_1 h_2 h_3 , k_1 (h_1 \cdot k_2) (h_1 h_2 \cdot k_3))$, where $g_3=g_1^{-1}$ and substituting in \cref{eq:inverse} by
setting $h_3=h_1^{-1}$ and $k_{3} = h^{-1}_1\cdot k^{-1}_1$, gives the following $\operatorname{AD}$joint action
\begin{align}
\operatorname{AD}_{g_1}g_2 &= g_1 g_2 g^{-1}_1 = (h_1,k_1)(h_2,k_2)(h_1,k_1)^{-1}\\
&=(h_1,k_1)(h_2,k_2)(h^{-1}_1,h^{-1}_1\cdot k^{-1}_1)\\
& = (h_1 h_2 h_1^{-1} , k_1 (h_1 \cdot k_2) (h_1 h_2 \cdot (h^{-1}_1\cdot k^{-1}_1)))\\
& = (h_1 h_2 h_1^{-1} , k_1 (h_1 \cdot k_2) (h_1\cdot h_2) \cdot (k^{-1}_1)))\\
&= (\operatorname{AD}_{h_1}h_2,k_1 \operatorname{AD}_{h_1}k_2 \phi_{\operatorname{AD}_{h_1}h_2} k_1^{-1})
\end{align}
Where we have used the homomorphism property \cref{eq:property1}.
We now compute the Adjoint action
\begin{align}
\operatorname{Ad}_{(h_1,k_1)}(\eta_{h},\eta_{k}) := \frac{d}{d\epsilon}\bigg\vert_{\epsilon=0} \operatorname{AD}_{(h_1,k_1)}(h_{2,\epsilon},k_{2,\epsilon}), \quad\text{where}\quad (\eta_h,\eta_k) := \frac{d }{d\epsilon}\bigg|_{\epsilon=0} (h_{2,\epsilon},k_{2,\epsilon}) \in \mathfrak{g}=T_e G.
\end{align}
where for notational convenience $\eta_{g} = (\eta_{h},\eta_k):= (\dot{h}_2(0),\dot{k}_2(0)) \in \mathfrak{g}$ and $e_G = (e_h,e_K) = (h_2(0),k_2(0))$. The first component of $\operatorname{Ad}_{g_{1}} \eta_{2}$ is trivially given by $\operatorname{Ad}_{h_1}\eta_{h}$. For the second component of $\operatorname{Ad}_{g_{1}} \eta_{2}$ we expand out the conjugations
\begin{align}
k_1 \operatorname{AD}_{h_1}k_2 \phi_{\operatorname{AD}_{h_1}h_2}k_1^{-1} = k_1 \cdot_{K} h_1 k_2(\epsilon) h_1^{-1} \cdot_{K} (h_1h_2(\epsilon)h_1^{-1}) k^{-1}_{1} (h_1 h^{-1}_2(\epsilon) h_1^{-1}),
\end{align}
and take the derivative at the group identity
\begin{align}
    \frac{d}{ds}\bigg|_{s=0} [k_1 (h_1 k_2(s) h_1^{-1}) (h_1 h_2(s) h_1^{-1} k_1^{-1} h_1 h_2(s)^{-1} h_1^{-1})]
    &= k_1 (h_1 \dot{k}_2(0) h_1^{-1}) (h_1 e_{H} h_1^{-1} k_1^{-1} h_1 e_H h_1^{-1})  \\
    & + k_1 (h_1 e_K h_1^{-1}) (h_1 \dot{h}_2(0) h_1^{-1} k_1^{-1} h_1 e_H h_1^{-1})  
    \\
    & - k_1 (h_1 e_K h_1^{-1}) (h_1 e_H h_1^{-1} k_1^{-1} h_1 e_H \dot{h}_2(0) e_H h_1^{-1}).
\end{align}
Group homomorphisms have special properties on evaluation at identities and inverses, and the above expression simplifies to
\begin{align}
\frac{d}{ds}\bigg|_{s=0} k_1 \operatorname{AD}_{h_1}k_2 \phi_{\operatorname{AD}_{h_1}h_2}k_1^{-1}&= k_1 (h_1 \dot{k}_2(0) h_1^{-1}) ( k_1^{-1})  + k_1 (h_1 \dot{h}_2(0) h_1^{-1} k_1^{-1}) - (h_1 \dot{h}_2(0)  h_1^{-1})\\
    &= k (h\cdot \eta_k) k^{-1} + k(h \cdot \eta_h )k^{-1} - h \cdot \eta_h\\
    &= k \cdot(h\cdot \eta_k)  + k\cdot(h \cdot \eta_h ) - h \cdot \eta_h.
\end{align}
The above conjugate actions have all been denoted by $\cdot$, we briefly recall how these could be informally reinterpreted for general Lie groups $H,K$. 
Here $h\cdot\eta_k$ denotes the induced (by $\phi$) conjugate action of $h\in H$ on $\eta_{k}\in \mathfrak{k}$, denoted more explicitly as $\Phi_{h}(\eta_k)$. Similarly $h\cdot \eta_h$ is Adjoint action of $H$ on $\mathfrak{h}$, denoted $\operatorname{Ad}_h\eta_h$. Finally, $k\cdot \eta_h$ is the induced conjugate action of $k\in K$ on $\eta_h\in \mathfrak{h}$ denoted by $\Psi_k h$. In the specific instance $H=K$, simplification can be achieved, as all group actions and representations are Lie group homomorphisms given by conjugation, and elements in the Lie algebras can be added. For example
\begin{align}
 k \cdot (h\cdot \eta_k)  + k \cdot (h \cdot \eta_h ) - h \cdot \eta_h = kh \cdot (\eta_k + \eta_h) - h\cdot \eta_h = \operatorname{Ad}_{kh}(\eta_k+\eta_h) -\operatorname{Ad}_{h}\eta_h.
\end{align}

The adjoint action can be computed similarly by expanding the conjugation mappings and taking the derivative and evaluating at the identity 
\begin{align}
\frac{d}{dt}\bigg|_{t=0}\left(k_t \cdot (h_t\cdot \eta_k)  + k_t \cdot (h_t \cdot \eta_h ) - h_t \cdot \eta_h\right) &= [\xi_k ,\eta_k] + [\xi_h , \eta_k] + [\xi_k,\eta_h] + [\xi_h,\eta_h] - [\xi_h,\eta_h]\\ &=  [\xi_k+\xi_h,\eta_k+\eta_h] - [\xi_h,\eta_h].
\end{align}
verifying 
\begin{align}
\operatorname{ad}_{(\xi_h,\xi_k)} (\eta_h,\eta_k) =(\operatorname{ad}_{\xi_h} \eta_h,\operatorname{ad}_{(\xi_k+\xi_h)} (\eta_k+\eta_h)  - \operatorname{ad}_{\xi_h} \eta_h).
\end{align}

We now define an element of the dual space $(m,n)\in \mathfrak{g}^* \equiv (\mathfrak{h}\ltimes \mathfrak{k})^{*}$, and derive coadjoint motion. 
\begin{align}
\langle \operatorname{ad}^{*}_{(\xi_h,\xi_k)} (m,n), (\eta_h,\eta_k)\rangle &= \langle (m,n),\operatorname{ad}_{(\xi_h,\xi_k)} (\eta_h,\eta_k) \rangle \\
&= \langle (m,n) , (\operatorname{ad}_{\xi_h} \eta_h ,\operatorname{ad}_{(\xi_k+\xi_h)} (\eta_k+\eta_h)  - \operatorname{ad}_{\xi_h} \eta_h )\rangle \\
&= \langle (m,n) , (\operatorname{ad}_{\xi_h} \eta_h ,\operatorname{ad}_{(\xi_k+\xi_h)} \eta_k + \operatorname{ad}_{(\xi_k+\xi_h)}\eta_h  - \operatorname{ad}_{\xi_h} \eta_h )\rangle \\
&= \langle (m,n) , (\operatorname{ad}_{\xi_h} \eta_h ,\operatorname{ad}_{(\xi_k+\xi_h)} \eta_k + \operatorname{ad}_{\xi_k}\eta_h   )\rangle\label{eq:coadstarderivation}
\end{align}
We now assume that there exists a pairing with the following  inner product structure
\begin{align}
    \langle (m,n),(u,v) \rangle_{\mathfrak{g}^*\times\mathfrak{g}} = \langle m,u\rangle_{\mathfrak{h}^*\times\mathfrak{h}}+ \langle n,v\rangle_{\mathfrak{k}^*\times\mathfrak{k}}.
\end{align}
Following on the calculation from \cref{eq:coadstarderivation}, we have that
\begin{align}
\langle (m,n) , (\operatorname{ad}_{\xi_h} \eta_h ,\operatorname{ad}_{(\xi_k+\xi_h)} \eta_k + \operatorname{ad}_{(\xi_k )}\eta_h  )\rangle &= 
\langle m , \operatorname{ad}_{\xi_h} \eta_h \rangle + \langle n,\operatorname{ad}_{(\xi_k+\xi_h)} \eta_k\rangle +\langle n,\operatorname{ad}_{\xi_k}\eta_h   \rangle\\
&= 
\langle \operatorname{ad}_{\xi_h}^{*} m ,  \eta_h \rangle + \langle \operatorname{ad}_{(\xi_k+\xi_h)}^{*}n , \eta_k \rangle +\langle \operatorname{ad}_{\xi_k}^{*}n,\eta_h   \rangle.
\end{align}
Giving the following representation, 
\begin{align}
\operatorname{ad}^{*}_{(\xi_h,\xi_k)} (m,n) = (\operatorname{ad}^{*}_{\xi_h} m+ \operatorname{ad}_{\xi_k}^{*} n, \operatorname{ad}_{(\xi_k + \xi_h)}^{*} n  ),
\end{align}
which we can write informally in co-adjoint matrix form or block diagonal Lie-Poisson form respectively as, 
\begin{align}
\left(\operatorname{ad}^{*}_{(\xi_h,\xi_k)} (m,n) \right)^{T}= \begin{bmatrix}
\operatorname{ad}_{\xi_h}^*(\cdot) & \operatorname{ad}_{\xi_k}^*(\cdot)\\
0&\operatorname{ad}_{\xi_h+\xi_k}^{*} (\cdot)
\end{bmatrix}
\begin{bmatrix}
    m\\
    n
\end{bmatrix} 
= 
\begin{bmatrix}
\operatorname{ad}_{(\cdot)}^*m & \operatorname{ad}_{(\cdot)}^*n \\
\operatorname{ad}_{(\cdot)}^*n& \operatorname{ad}_{(\cdot)}^*n
\end{bmatrix}
\begin{bmatrix}
    \xi_h\\
    \xi_k
\end{bmatrix}.
\end{align}

\begin{remark}
If one assumes a more general pairing (assuming $\mathfrak{h}=\mathfrak{k}$)
as follows
\begin{align}
    \langle (m,n),(u,v) \rangle_{\mathfrak{g}^*\times\mathfrak{g}} := a\langle m,u\rangle+b\langle m,v\rangle+
    c\langle n,u\rangle
    +d\langle n,v\rangle,
\end{align}
one attains the coadjoint motion
\begin{align}
\operatorname{ad}^{*}_{(\xi_h,\xi_k)} (m,n) = (\operatorname{ad}^{*}_{\xi_h} (am+cn) + \operatorname{ad}_{\xi_k}^{*} (dn+bm) , \operatorname{ad}_{(\xi_k + \xi_h)}^{*} (bm+dn)  ).
\end{align}
Which we can write informally in co-adjoint matrix form or Lie-Poisson form, respectively as, 
\begin{align}
\left(\operatorname{ad}^{*}_{(\xi_h,\xi_k)} (m,n) \right)^{T}= \begin{bmatrix}
\operatorname{ad}_{\xi_h}^*(\cdot) & \operatorname{ad}_{\xi_k}^*(\cdot)\\
0&\operatorname{ad}_{\xi_h+\xi_k}^{*} (\cdot)
\end{bmatrix}\left(
\begin{bmatrix}
    a&c\\
    b&d
\end{bmatrix} 
\begin{bmatrix}
    m\\
    n
\end{bmatrix} 
\right)
= 
\begin{bmatrix}
\operatorname{ad}_{(\cdot)}^*(am+cn) & \operatorname{ad}_{(\cdot)}^*(bm+dn) \\
\operatorname{ad}_{(\cdot)}^* (bm+dn)& \operatorname{ad}_{(\cdot)}^* (bm+dn)
\end{bmatrix}
\begin{bmatrix}
    \xi_h\\
    \xi_k
\end{bmatrix}.
\end{align}
The variational principle gives the same geodesic equations but for the transformed momentum variables 
\begin{align}
\begin{bmatrix}
    \widehat{m}\\
    \widehat{n}
\end{bmatrix}
=
\begin{bmatrix}
    a&c\\
    b&d
\end{bmatrix} 
\begin{bmatrix}
    m\\
    n
\end{bmatrix}. 
\end{align}
For an inner product it is typically required that $b=c, ad>b^2, a>0, d>0$. Considering the parametrised metric associated with $a=d=1/2$, $c=d=\theta/2$, one attains the semidirect product equations considered in this paper when $\theta = 0$, and two copies of the original Euler-Poincaré equations when $\theta=1$. 
\end{remark}

\subsection{Modelling framework, composition of maps}\label{sec: composition of maps}
We now provide some additional motivation for studying the semidirect product of groups. 
Suppose that $g_t \in \operatorname{Diff}(M)$ describes the action of a fluid taking a position at initial time to a later time $x_t = g_t  x_0$. Then taking the time derivative of position (denoted $\dot{x}_t$) and transforming into the current frame gives $\dot{x}_t = \dot{g}_t x_0 = \dot{g}_t  g^{-1}_t  x_t$, identifying
the Eulerian velocity as $\tilde{u}:=\dot{g}g^{-1}$. We suppose that the diffeomorphism $g$ is written as a composition of two other diffeomorphisms such that $x_t = (h\circ k) x_0$. One can naturally attain a decomposition of the Eulerian velocity through the chain rule, 
$\dot{x}_t = \dot{g}g^{-1}x_0 = \partial_t(h\circ k) (h\circ k)^{-1} x_0 = [\dot{k}  (h\circ k)+\dot{h}\circ k)]( k^{-1}\circ h^{-1})x_0  =(\dot{h}h^{-1} + h_{*}\dot{k}k^{-1})x_0$. The Eulerian velocity can be written as $\tilde{u} = u + v$ where $u:= \dot{h}h^{-1}$ is the velocity of $h$ and $v := h_{*}\dot{k}k^{-1}$ denotes the velocity of $k$ pushed-forward by $h$.
Supposing through a reduction by stages (see for example \cite{holm202431,marsden2007hamiltonian,ortega2013momentum}) there exists a reduced Lagrangian consisting of the velocity $u$ and the push forwarded velocity $v$, Hamilton's principle of critical action can be written as
\begin{align}
0=\delta \int \ell(\dot{h}h^{-1},h_{*}\dot{k}k^{-1})dt = \delta\int\ell(u,v)dt.
\end{align}
Let the flow maps $(h,k)$ be parametrised by $\epsilon$, and denote $h':=\partial_{\epsilon}|_{\epsilon=0}h_{\epsilon}$, and $k':=\partial_{\epsilon}|_{\epsilon=0}k_{\epsilon}$, the constrained variations are calculated \cite{holm202431,bruveris2012mixture} to be
\begin{align}
\delta u =\frac{d}{d\epsilon}\bigg|_{\epsilon=0}\dot{h}h^{-1} &= \dot{w}_1 -\operatorname{ad}_{u}w_1,\quad w_1:=h'h^{-1},\label{eq:constrained variations1} \\
\delta v = \frac{d}{d\epsilon}\bigg|_{\epsilon=0} (h_{*}\dot{k}k^{-1}) &= \dot{w}_2 -\operatorname{ad}_{u+v}w_2 -\operatorname{ad}_{u}w_1, \quad w_2 := h_*(k'k^{-1}).\label{eq:constrained variations2}
\end{align}
This reveals a $\mathfrak{g}\ltimes \mathfrak{g}$ semidirect product structure to the resulting Euler-Poincaré equations
\begin{align}
\partial_t (m,n) + (\operatorname{ad}^*_{u}(m) + \operatorname{ad}^*_{v}n,\operatorname{ad}^*_{u+v}(n)) = 0,\quad\quad \quad (m,n):= \left(\frac{\delta \ell}{\delta u},\frac{\delta \ell}{\delta v}\right),\label{eq:semidirect_product}
\end{align}
since the constrained variations align exactly with the adjoint actions of the semidirect product of groups in \cref{prop:adjoint actions for a semidirect product}.

If one defines the kinetic energy as the square of the total velocity  
\begin{align}
\ell(\tilde{u})= \frac{1}{2}\int\tilde{u}^2d^n x = \frac{1}{2}\int(\dot{h}h^{-1}+ h_{*}\dot{k}k^{-1})^2= \frac{1}{2}\int (u + v)^2 d^n x=\ell(u,v),
\end{align}
the momenta are equal 
\begin{align}
   \frac{\delta \ell}{\delta v}=\frac{\delta \ell}{\delta u} =m=n=(u+v)^{\flat}.
\end{align}
Hamilton's principle, subject to the constrained variations \cref{eq:constrained variations1,eq:constrained variations2} gives semidirect product Euler-Poincaré equations \cref{eq:semidirect_product}, with equal momenta becomes
\begin{align}
\partial_t (\tilde{m},\tilde{m}) + (\operatorname{ad}^*_{u+v}(\tilde{m}),\operatorname{ad}^*_{u+v}(\tilde{m})) = 0.
\end{align}
Recalling the definition of total fluid velocity $\tilde{u} = u+v$ and total momentum $\tilde{m}= \frac{\delta \ell(\tilde{u})}{\delta\tilde{u}}=(u+v)^{\flat}$, we recover two copies the usual Euler-Poincaré equation
\begin{align}
    \partial_t \tilde{m} + \operatorname{ad}^*_{\tilde{u}}\tilde{m}= 0.
\end{align}
Revealing that the decomposition of the diffeomorphism attains the original Euler-Poincaré equations, as one would expect. In this work, we consider the same constrained variations, but numerically investigate equations arising from Lagrangian's when one neglects the contribution of the cross terms in the Lagrangian, where an inner product of the form $\langle (,),(,)\rangle_{\mathfrak{h}\times \mathfrak{k}} = \langle, \rangle_{\mathfrak{h}} + \langle ,\rangle_{\mathfrak{k}}$ gives rise to a diagonal Lagrangian. This gives distinct equations to the Euler-Poincaré equations.

\subsection{Centrally extending the semidirect product}\label{Centrally extending the semidirect product group}
Defining $\tilde{\sigma}$ as 
\begin{align}
\tilde{\sigma}((\xi_1,\xi_2),(\eta_1,\eta_2)) := (\sigma(\xi_1,\eta_1),\sigma(\xi_2,\eta_2)),
\end{align}
where $\sigma:\mathfrak{g}\times \mathfrak{g}\mapsto \mathfrak{n}$ denotes the algebra 2-cocycle on $\mathfrak{g}$. 
One can verify that $\tilde{\sigma}$ satisfies the 2-cocycle relationship 
\begin{align}
\tilde{\sigma}((a_1,a_2),[(b_1,b_2),(c_1,c_2)]_{\mathfrak{g}\ltimes \mathfrak{g}}) + \text{cyclic permutations of a,b,c} = 0,\quad \forall a,b,c \in \mathfrak{g}\ltimes\mathfrak{g}\\ 
\tilde{\sigma}((a_1,a_2),[(b_1,b_2),(c_1,c_2)]_{\mathfrak{g}\times \mathfrak{g}}) + \text{cyclic permutations of $a$,b,c} = 0,\quad \forall a,b,c \in \mathfrak{g}\times\mathfrak{g}
\end{align}
for both the semidirect product Lie bracket and the direct product Lie bracket, from the original 2-cocycle condition on $\sigma$. Continuity and skew-symmetry similarly follow from $\sigma$. Therefore, assuming a canonical pairing
\begin{align}
\langle(m,n,l,o),(\xi_1,\xi_2,\alpha_1,\alpha_2)\rangle = \langle m,\xi_1 \rangle_{\mathfrak{g}} + \langle n,\xi_2 \rangle_{\mathfrak{g}} + \langle l,\alpha_1 \rangle_{\mathfrak{n}} + \langle o,\alpha_2 \rangle_{\mathfrak{n}},
\end{align}
one can derive the following co-adjoint operator for this central extension of the semidirect product, as follows
\begin{align}
\left\langle\operatorname{ad}^*_{(\xi_1,\xi_1,\alpha_1,\alpha_2)}(m,n,l,o),(\eta_1,\eta_2,\beta_1,\beta_2) \right\rangle &= \langle (m,n),\operatorname{ad}_{(\xi_1,\xi_2)} (\eta_1,\eta_2) \rangle + \langle (l,o),\tilde{\sigma}((\xi_1,\xi_2),(\eta_1,\eta_2))\rangle\\ &= \langle \operatorname{ad}_{(\xi_1,\xi_2)}^*(m,n) , (\eta_1,\eta_2,\beta_1,\beta_2) \rangle + \langle l,\sigma(\xi_1,\eta_1)\rangle+ \langle o,\sigma(\xi_2,\eta_2)\rangle\\
& = \langle\operatorname{ad}^*_{\xi_1}m + \operatorname{ad}^*_{\xi_2}n + \sigma^*(\xi_1,l), \eta_1 \rangle + \langle\operatorname{ad}^*_{\xi_1+\xi_2}n + \sigma^*(\xi_2,o), \eta_2\rangle \\
&= \langle \left( \operatorname{ad}_{\xi_1}^{*} m + \operatorname{ad}_{\xi_2}^{*} n + \sigma^*(\xi_1,l), \operatorname{ad}^*_{\xi_1+\xi_2}n + \sigma^*(\xi_2,o),0,0 \right),(\eta_1,\eta_2,\beta_1,\beta_2)\rangle.
\end{align}
This can be written (in terms of the previous notation in \cref{sec: semidirect product of centrall extension} ($\xi_1 = \xi_h$, $\xi_2=\xi_k$, $\alpha_1 = \alpha$, $\alpha_2=\beta$) ) in Lie-Poisson or coadjoint operator form as
\begin{align}
\frac{d}{dt}
\begin{bmatrix}
m\\
l\\
n\\
o
\end{bmatrix}
&= -
\begin{bmatrix}
\operatorname{ad}_{(\cdot)}^*m + \sigma^*(\cdot,l)& 0& \operatorname{ad}_{(\cdot)}^*n & 0\\
0&0&0&0\\
\operatorname{ad}_{(\cdot)}^*n & 0&\operatorname{ad}_{(\cdot)}^*n+ \sigma^*(\cdot,o)&0\\
0&0&0&0
\end{bmatrix}
\begin{bmatrix}
    \xi_h\\
    \alpha\\
    \xi_k \\
    \beta
\end{bmatrix}
= -
\begin{bmatrix}
\operatorname{ad}^{*}_{\xi_{h}}(\cdot) &
\sigma^{*}(\xi_h,\cdot)
&\operatorname{ad}^{*}_{\xi_{k}}(\cdot)
&0\\
0&0&0&0\\
0&0&\operatorname{ad}^{*}_{\xi_{h}+\xi_{k}}(\cdot)& \sigma^*(\xi_{k},\cdot)\\
0&0&0&0
\end{bmatrix}
\begin{bmatrix}
m\\
l\\
n\\
o
\end{bmatrix}\label{eq:central_extension_of_semidirectproduct}.
\end{align}
The coadjoint contribution of the 2-cocycle is strictly diagonal in the Lie-Poisson form, and not involved in the coupling of the equations. Choosing $\mathfrak{g}= \operatorname{Vect}(S^1)$, one attains the following coupled systems of equations as examples of centrally extended semidirect product EP-equations, without dispersive coupling where $(o,l)$ are constants.
\begin{example}
CH-CH, $\ell  = 1/2||u||_{H^1_{\alpha}}^2 + 1/2||v||_{H^1_{\beta}}^2 + 1/2 \alpha^2 + 1/2 \beta^2$
\begin{align}
    \partial_t m + (um)_x + u_x m + lu_{xxx} + (vn)_x + v_x n  &= 0, \quad m = u-\alpha u_{xx},\\
\partial_t n + ((u+v)n)_x + (u+v)_x n + o (v)_{xxx} &= 0,\quad n = v-\beta v_{xx},
\end{align}
\end{example}

\begin{example}
KdV-Burgers, $\ell  = 1/2||u||_{L^2}^2 + 1/2||v||_{L^2}^2 + 1/2 \alpha^2$
\begin{align}
\partial_t u + 3uu_x + lu_{xxx} + 3vv_x  &= 0, \\
\partial_t v + ((u+v)v)_x + (u+v)_x v  &= 0.
\end{align}
\end{example}

\begin{example}
Burgers-KdV, $\ell  = 1/2||u||_{L^2}^2 + 1/2||v||_{L^2}^2 + 1/2 \beta^2$
\begin{align}
\partial_t u + 3uu_x   + 3vv_x  &= 0, \\
\partial_t v + ((u+v)v)_x + (u+v)_x v + o v_{xxx} &= 0.
\end{align}
\end{example}
These equations are EP-equations on a centrally extended semidirect product Lie algebra, but it is not clear what the group action is. The previous EP-equations on a semidirect product of centrally extended Lie algebras would have arisen if one introduced
$
\bar{\sigma}((\xi_1,\xi_2),(\eta_1,\eta_2)) := (\sigma(\xi_1,\eta_1),\sigma(\xi_2,\eta_2) + \sigma(\xi_1,\eta_2) + \sigma(\xi_2,\eta_1))
$ instead.

\subsection{Ideal Incompressible MHD}\label{sec:MHD}

\begin{example}[Ideal Incompressible MHD]
A magnetic extension is a particular subcase of vector space extension constructed using the dual of the Lie algebra $G\ltimes \mathfrak{g}^*$. When $G=\operatorname{SDiff}(M)$,
the Lie algebra action of
$\mathfrak{g}\ltimes \mathfrak{g}^* = \operatorname{SVect}(M) \ltimes\left(\Lambda^1(M) / \mathrm{d} \Lambda^0(M)\right)$ on itself is given by 
\begin{align}
\operatorname{ad}_{(u,[b])} (v,[c]) = (\operatorname{ad}_{u}v, \operatorname{ad}^*_{v}[b]-\operatorname{ad}^*_u [c] )=(-\mathcal{L}_u v,\mathcal{L}_v [b] - \mathcal{L}_{u}[c]). 
\end{align}
So the co-adjoint action is calculable to be $(\mathcal{L}_u [m] - \mathcal{L}_B[b],\mathcal{L}_{u}B)$, as follows
\begin{align}
\langle\operatorname{ad}^*_{(u,[b])}([m],B),(v,[c])\rangle &= \langle \operatorname{ad}_{u}v, [m]\rangle - \langle \operatorname{ad}^*_{u}[c],B\rangle +\langle\operatorname{ad}_v^*[b],B\rangle\\
&= \langle  \operatorname{ad}_{u}^*[m],v\rangle - \langle \operatorname{ad}_u B,[c]\rangle +\langle B,\mathcal{L}_v [b]\rangle\\
&= \langle  \operatorname{ad}_{u}^*[m],v\rangle + \langle \mathcal{L}_u B,[c]\rangle -\langle B \diamond [b],v \rangle\\
&= \langle \left( \operatorname{ad}_u^* [m] - \mathcal{L}_{B}[b],\mathcal{L}_{u}B \right),  \left(v,[c]\right) \rangle.
\end{align}

To compute the operator $\diamond:V^*\times V\mapsto \mathfrak{g}^*$ (defined previously as $\langle n,\alpha(\eta)v \rangle_{\mathfrak{g}\times\mathfrak{g}^*} :=-\langle n\diamond v, \eta\rangle_{\mathfrak{g}^*\times
\mathfrak{g}})$, we used the contragradient representation of the Lie derivative as follows $-\langle B\diamond [b],v  \rangle = \langle \mathcal{L}_v [b], B\rangle = -\langle [b],\mathcal{L}_v B\rangle =  \langle [b],\mathcal{L}_{B}v \rangle = -\langle \mathcal{L}_{B}[b],v \rangle$. The resulting Euler-Poincaré equation is expressible in matrix form, as 
\begin{align}
\frac{d}{dt}
\begin{bmatrix}
    [m]\\
    B
\end{bmatrix}
&= 
- 
\begin{bmatrix}
    \mathcal{L}_{u}(\cdot) & -\mathcal{L}_{(\cdot)}[b] \\
    0 &  \mathcal{L}_{u}(\cdot)  \\
\end{bmatrix}
\begin{bmatrix}
    [m]\\
    B\\
\end{bmatrix}\label{2D-incompressible-MHD}=
-
\begin{bmatrix}
    \mathcal{L}_{(\cdot)}[m] & -\mathcal{L}_{B}(\cdot) \\
 \mathcal{L}_{(\cdot)}B & 0  \\
\end{bmatrix}
\begin{bmatrix}
     u  \\
    [b] 
\end{bmatrix}.
\end{align}
Here, the Lie derivative of the one-form is denoted $\mathcal{L}_{u}[m]$, and the Lie derivative of a vector field is the commutator of vector fields and denoted $\mathcal{L}_{u}B$. 
\end{example}

\begin{example}[Ideal Incompressible Hall MHD]
By combining \cref{sec:Direct products,2D-incompressible-MHD} the coordinate free Incompressible Hall-MHD equations can be constructed in arbitrary dimensions from the direct product \cite{holm1987hall}, and takes the following coadjoint and Lie-Poisson form 
\begin{align}
\frac{d}{dt}
\begin{bmatrix}
    [m]\\
    B\\
    [n]\\
    C
\end{bmatrix}
&= 
- 
\begin{bmatrix}
    \mathcal{L}_{u}(\cdot) & -\mathcal{L}_{(\cdot)}[b] & 0 & 0\\
    0 &  \mathcal{L}_{u}(\cdot) & 0 & 0 \\
    0&0&\mathcal{L}_{v}(\cdot)&-\mathcal{L}_{(\cdot)}[c]\\
    0&0&0&\mathcal{L}_{v}(\cdot)
\end{bmatrix}
\begin{bmatrix}
    [m]\\
    B\\
    [n]\\
    C
\end{bmatrix}\label{incompressible-HMHD}=
-
\begin{bmatrix}
    \mathcal{L}_{(\cdot)}[m] & -\mathcal{L}_{B}(\cdot) & 0& 0\\
     \mathcal{L}_{(\cdot)}B & 0 & 0 & 0 \\
    0&0&\mathcal{L}_{(\cdot)}[n]&-\mathcal{L}_{C}(\cdot)\\
    0&0&\mathcal{L}_{(\cdot)}C&0
\end{bmatrix}
\begin{bmatrix}
     u  \\
    [b] \\
     v  \\
    [c]
\end{bmatrix}.
\end{align}
These are geodesic equations for $(G\ltimes \mathfrak{g}^*)\times (G\ltimes \mathfrak{g}^*)$ where $G=\operatorname{SDiff}(M)$. 
\end{example}

\begin{example}[Self-Semi-direct product] 
Using the general form in \cref{eq:semidirectproduct_of_semidirect_products,2D-incompressible-MHD} one can inspect Euler-Poincaré equations for the self-semidirect product extension of the magnetically extended volume preserving group $\operatorname{SDiff}(M)$ group given by $
(\operatorname{SDiff}(M) \ltimes\left(\Lambda^1(M) / \mathrm{d} \Lambda^0(M)\right))\ltimes (\operatorname{SDiff}(M) \ltimes\left(\Lambda^1(M) / \mathrm{d} \Lambda^0(M)\right))
$. The Lie algebra is $$(u,[b],v,[c])\in 
(\operatorname{SVect}(M) \ltimes\left(\Lambda^1(M) / \mathrm{d} \Lambda^0(M)\right))\ltimes (\operatorname{SVect}(M) \ltimes\left(\Lambda^1(M) / \mathrm{d} \Lambda^0(M)\right))
$$ and it acts on the dual space $$([m],B,[n],C)\in \left(\left(\Lambda^1(M) / \mathrm{d} \Lambda^0(M)\right)
\times\operatorname{SVect}(M)\right)\times\left(\left(\Lambda^1(M) / \mathrm{d} \Lambda^0(M)\right)
\times\operatorname{SVect}(M)\right).
$$
Assuming a Lagrangian with canonical pairing between vector space and dual space gives the Euler-Poincaré equations in coadjoint and Lie-Poisson form
\begin{align}
\frac{d}{dt}
\begin{bmatrix}
    [m]\\
    B\\
    [n]\\
    C
\end{bmatrix}
&= 
- 
\begin{bmatrix}
    \mathcal{L}_{u}(\cdot) & -\mathcal{L}_{(\cdot)}[b] & \mathcal{L}_{v}(\cdot) & -\mathcal{L}_{(\cdot)}[c]\\
    0 &  \mathcal{L}_{u}(\cdot) & 0 &  \mathcal{L}_{v}(\cdot) \\
    0&0&\mathcal{L}_{u+v}(\cdot)&-\mathcal{L}_{(\cdot)}([b]+[c])\\
    0&0&0&\mathcal{L}_{u+v}(\cdot)
\end{bmatrix}
\begin{bmatrix}
    [m]\\
    B\\
    [n]\\
    C
\end{bmatrix}\label{incompressible-MHD}=
-
\begin{bmatrix}
    \mathcal{L}_{(\cdot)}[m] & -\mathcal{L}_{B}(\cdot) & \mathcal{L}_{(\cdot)}[n] & -\mathcal{L}_{C}(\cdot)\\
     \mathcal{L}_{(\cdot)}B & 0 & \mathcal{L}_{(\cdot)}C & 0 \\
    \mathcal{L}_{(\cdot)}[n]&-\mathcal{L}_{C}(\cdot)&\mathcal{L}_{(\cdot)}[n]&-\mathcal{L}_{C}(\cdot)\\
    \mathcal{L}_{(\cdot)}C&0&\mathcal{L}_{(\cdot)}C&0
\end{bmatrix}
\begin{bmatrix}
     u  \\
    [b] \\
     v  \\
    [c]
\end{bmatrix}
\end{align}
 The magnetic extension of $\operatorname{SDiff}(M)$ gives the incompressible MHD equations (given in the upper left $2\times 2$ block), discussed in \cite{khesin2009geometry,arnol'd2009topological} where $u$ denotes velocity field and $B$ denotes magnetic field. The system \cref{incompressible-MHD} denotes the self-semidirect product extension with another incompressible magnetic fluid with velocity $v$ and magnetic field $C$. 
\end{example}

\end{appendices}

\end{document}